\newif\ifdraft\drafttrue
 \draftfalse 

\documentclass[10pt,a4paper]{article}

\usepackage{a4wide}
\usepackage{url}
\usepackage{versions}
\includeversion{LONG}
\excludeversion{SHORT}

\newcommand{\ADG}[1]{\ifdraft\marginpar{ADG says: #1}\fi}
\newcommand{\JB}[1]{\ifdraft\marginpar{JB says: #1}\fi}
\newcommand{\MS}[1]{\ifdraft\marginpar{Marcin says: #1}\fi}

\usepackage{amsmath,amssymb}
\usepackage{stmaryrd}
\usepackage{mathrsfs}
\usepackage{bussproofs}
\usepackage{comment}

\newcommand{\ie}{{\it i.e.}}

\newcommand{\termone}{M}
\newcommand{\termtwo}{N}
\newcommand{\termthree}{L}
\newcommand{\termfour}{P}

\newcommand{\varone}{x}
\newcommand{\vartwo}{y}

\newcommand{\red}{\rightarrow}

\newcommand{\Nat}{\mathbb{N}}

\newcommand{\midd}{\; \; \mbox{\Large{$\mid$}}\;\;}
\newcommand{\NN}{\mathbb{N}}
\newcommand{\RR}{\mathbb{R}}
\newcommand{\RRp}{\mathbb{R}_{[0,1]}}
\newenvironment{varitemize}
{
\begin{list}{\labelitemi}
{
\setlength{\itemsep}{0pt}
 \setlength{\topsep}{0pt}
 \setlength{\parsep}{0pt}
 \setlength{\partopsep}{0pt}
 \setlength{\leftmargin}{15pt}
 \setlength{\rightmargin}{0pt}
 \setlength{\itemindent}{0pt}
 \setlength{\labelsep}{5pt}
 \setlength{\labelwidth}{10pt}}}
{
 \end{list}
}

\newcommand{\ratio}{0.6}
\newenvironment{display}[2][\ratio]{
  \begin{tabbing}
    \hspace{0.1em} \= \hspace{1.5em} \= \hspace{#1\linewidth-3.2em} \= \hspace{1.5em} \= \kill
    \textbf{#2}\\[-.8ex]
    \hbra\\[-.8ex]
  }
  {\\[-.8ex]\hket
  \end{tabbing}}

\newcommand{\clause}[3][]{\>$#2$\>#1\>#3}

\newcounter{rule}

\newcommand{\hbra}{
  \hbox to \columnwidth{\vrule width0.3mm height 1.8mm depth-0.3mm
      \leaders\hrule height1.8mm depth-1.5mm\hfill
          \vrule width0.3mm height 1.8mm depth-0.3mm}}

\newcommand{\hket}{
  \hbox to \columnwidth{\vrule width0.3mm height1.5mm
      \leaders\hrule height0.3mm\hfill
          \vrule width0.3mm height1.5mm}}

\makeatletter
 \newcommand{\addToLabel}[1]{%
   \protected@edef\@currentlabel{\@currentlabel#1}%
  }
\makeatother

\newcommand{\restaterule}[4][]{%
  \addToLabel{(\textsc{#2})}%
  $\begin{array}[b]{@{}l@{}}%
    \mbox{(\textsc{#2})#1}\\%
    \begin{array}{@{}l@{}}
      #3\\
      \hline      \raisebox{0ex}[2.5ex]{\strut}#4%
    \end{array}
  \end{array}$}

\usepackage[numbers,sectionbib]{natbib}

\newcommand{\Appref}[1]{Appendix~\ref{#1}}

\usepackage{balance}
\usepackage{textcomp}
\usepackage{xspace}
\usepackage{listings} 
\usepackage{color} 
\usepackage{textcomp} 
\usepackage{extarrows}

\newcounter{numberone}

\newenvironment{varenumerate}
{
\begin{list}{\arabic{numberone}.}
{
  \usecounter{numberone}
  \setlength{\itemsep}{0pt}
  \setlength{\topsep}{0pt}
  \setlength{\parsep}{0pt}
  \setlength{\partopsep}{0pt}
  \setlength{\leftmargin}{15pt}
  \setlength{\rightmargin}{0pt}
  \setlength{\itemindent}{0pt}
  \setlength{\labelsep}{5pt}
  \setlength{\labelwidth}{15pt}
}}
{
\end{list} 
}

\newenvironment{prog}{\begin{array}[t]{@{}l@{}}}{\end{array}}

\newcommand{\measterms}{\mathcal{M}}
\newcommand{\expset}{\{\fail\}}

\newcommand{\Abs}[1]{\lvert{#1}\rvert}

\newcommand{\distset}{\mathcal{I}}
\newcommand{\distone}{\mathsf{D}}
\newcommand{\disttwo}{\mathsf{E}} 
\newcommand{\dirac}[1]{\delta(#1)}
\newcommand{\pdf}[1]{\operatorname{pdf}_{#1}}

\newcommand{\Inject}[1]{\lceil{#1}\rceil}

\newcommand{\bnf}{::=}

\newcommand{\varset}{\mathcal{X}}
\newcommand{\abstr}[2]{\lambda #1.#2}

\newcommand{\ite}[3]{\mathtt{if}\ {#1}\ \mathtt{then}\ {#2}\ \mathtt{else}\ {#3}}
\newcommand{\score}[1]{\mathtt{score}(#1)}
\newcommand{\ttrue}{\mathtt{true}}
\newcommand{\tfalse}{\mathtt{false}}
\newcommand{\sbst}[3]{#1\{#3/#2\}}
\newcommand{\valone}{V}
\newcommand{\valtwo}{W}
\newcommand{\valset}{\mathcal{V}}
\newcommand{\gvalone}{G}
\newcommand{\gvaltwo}{H}
\newcommand{\gvalset}{\mathcal{GV}}

\newcommand{\tsetone}{A}
\newcommand{\tsettwo}{B}

\newcommand{\fail}{\mathtt{fail}}
\newcommand{\terror}{\mathtt{fail}}
\newcommand{\intfun}[1]{\sigma_{#1}}

\newcommand{\redterms}{R\Lambda_P}
\newcommand{\lamterms}{\Lambda}
\newcommand{\cterms}{C\Lambda}

\newcommand{\msone}{\mathscr{M}}

\newcommand{\setone}{X}
\newcommand{\settwo}{Y}

\newcommand{\funone}{f}

\newcommand{\algone}{\Sigma}
\newcommand{\algtwo}{\Sigma'}

\newcommand{\hole}{[\cdot]}
\newcommand{\ectxone}{E}

\newcommand{\ectxthree}{G}

\newcommand{\ct}[2]{#1[#2]}
\newcommand{\ctm}[2]{#1\{#2\}}
\newcommand{\emdistr}{\mathbf{0}}
\newcommand{\distrone}{\mathscr{D}}
\newcommand{\distrtwo}{\mathscr{E}}
\newcommand{\distrthree}{\mathscr{F}}
\newcommand{\distrfour}{\mathscr{G}}
\newcommand{\distrfive}{\mathscr{H}}
\newcommand{\distrsix}{\mathscr{I}}
\newcommand{\distrseven}{\mathscr{L}}
\newcommand{\restr}[2]{{#1}\rvert_{#2}}

\newcommand{\supp}[1]{\operatorname{supp}(#1)}

\newcommand{\funs}{\mathbf{F}}

\newcommand{\sisss}[3]{#1\Rightarrow_{#2} #3}
\newcommand{\redss}{\Rightarrow}

\newcommand{\sibss}[3]{#1\Downarrow_{#2} #3}
\newcommand{\sts}[1]{\llbracket #1\rrbracket_{\Rightarrow}}
\newcommand{\bts}[1]{\llbracket #1\rrbracket_{\Downarrow}}
\newcommand{\tsq}[1]{\llbracket #1\rrbracket}

\newcommand{\Tracedist}[1]{\langle\!\langle #1\rangle\!\rangle}
\newcommand{\Tracevaldist}[1]{\langle\!\langle  #1\rangle\!\rangle^{\!\valset}}

\newcommand{\valdist}[1]{\restr{\tsq{#1}}{\valset}}
\newcommand{\detred}{\xrightarrow{\text{det}}}

\newcommand{\oCPO}{\ensuremath{\omega\mathbf{CPO}}}

\newcommand{\Termsmp}[1]{\mathbf{P}_{\! #1}}
\newcommand{\Termsmv}[1]{\mathbf{P}^{\valset}_{\! #1}}
\newcommand{\termsmp}[2]{\Termsmp{#1}(#2)}
\newcommand{\termsmv}[2]{\Termsmv{#1}(#2)}
\newcommand{\Obssmp}[1]{\mathbf{O}_{#1}}
\newcommand{\obssmp}[2]{\Obssmp{#1}(#2)}
\newcommand{\sampseq}{\mathbb{S}}
\newcommand{\sampseqp}[1]{\mathbb{S}_{#1}}
\newcommand{\sbd}[1]{\llbracket #1\rrbracket_{\sampseq}}
\newcommand{\sbdp}[2]{\sbd{#1}^{#2}}

\newenvironment{restate}[1]%
{\begin{trivlist}\item[]{\normalsize\sc Restatement of #1.} \it}%
  {\end{trivlist}}

\newcommand{\condnocr}{}

\newcommand{\Rtwo}{\textsc{R2}\xspace}
\newcommand{\Church}{\textsc{Church}\xspace}
\newcommand{\Venture}{\textsc{Venture}\xspace}
\newcommand{\Anglican}{\textsc{Anglican}\xspace}
\newcommand{\WebChurch}{\textsc{Web Church}\xspace}
\newcommand{\WebPPL}{\textsc{WebPPL}\xspace}
\newcommand{\Scheme}{\textsc{Scheme}\xspace}
\newcommand{\JavaScript}{\textsc{JavaScript}}

\renewcommand{\emptyset}{\varnothing}
\newcommand{\Set}[1]{\{#1\}}                    % set constructor
\newcommand{\0}{\emptyset}
\newcommand{\mAlg}{\Sigma}
\newcommand{\Borel}{\mathcal{B}}

\newcommand{\Real}{\mathbb{R}}
\newcommand{\RealInf}{\mathbb{R}_{+}}

\newcommand{\Powerset}[2][]{\mathcal{P}_{#1}(#2)}
\newcommand{\indfun}[2]{{[{#2}\in{#1}]}}
\newcommand{\deq}{\mathrel{\smash[t]{\triangleq}}}

\begin{document}
\title{A Lambda-Calculus Foundation for\\Universal Probabilistic Programming\thanks{The first author is supported by the Swedish Research Council grant 2013-4853. The second author is partially supported by the ANR project 12IS02001
PACE and the ANR project 14CE250005 ELICA. The fourth author was supported by Microsoft Research through its PhD Scholarship Programme.}}
  \author{
      Johannes Borgstr\"om\footnote{Uppsala University} \and
      Ugo Dal Lago\footnote{University of Bologna \& INRIA} \and
      Andrew D. Gordon\footnote{Microsoft Research and University of Edinburgh} \and
      Marcin Szymczak\footnote{University of Edinburgh}}

\newtheorem{lemma}{Lemma}
\newtheorem{theorem}{Theorem}
\newtheorem{proposition}{Proposition}
\newtheorem{corollary}{Corollary}
\newtheorem{definition}{Definition}
\newenvironment{proof}{\begin{trivlist}
       \item[\hskip \labelsep {\sc Proof.}]}{\end{trivlist}}
%       \item[\hskip \labelsep {\bfseries Proof.}]}{\hfill $\Box$ \end{trivlist}}
\newcommand{\qed}{\hfill$\Box$}

\maketitle

\begin{abstract}
We develop the operational semantics of an untyped probabilistic
$\lambda$-calculus with continuous distributions, and both hard and soft constraints, 
as a foundation for universal probabilistic programming languages such as \Church,
\Anglican, and \Venture.
Our first contribution is to adapt the classic operational semantics
of~$\lambda$-calculus to a continuous setting via creating a measure
space on terms and defining step-indexed approximations. We prove
equivalence of big-step and small-step formulations of this
\emph{distribution-based semantics}.
To move closer to inference techniques, we also define the
\emph{sampling-based semantics} of a term as a function from a trace
of random samples to a value.
We show that the distribution induced by integration over the space of traces
equals the distribution-based semantics.
Our second contribution is to formalize the implementation technique
of trace \emph{Markov chain Monte Carlo} (MCMC) for our calculus and
to show its correctness.
A key step is defining sufficient conditions for the distribution
induced by trace MCMC to converge to the distribution-based semantics.
To the best of our knowledge, this is the first rigorous correctness
proof for trace MCMC for a higher-order functional language, 
or for a language with soft constraints.
\end{abstract}

\clearpage
\tableofcontents
\clearpage

%%%%%%%%%%%%%%%%%%%%%%
\section{Introduction} \label{sec:intro}
%%%%%%%%%%%%%%%%%%%%%%
In computer science, probability theory can be used for models that
enable system abstraction, and also as a way to compute in a setting
where having access to a source of randomness is essential to achieve
correctness, as in randomised computation or
cryptography~\cite{GoldwasserMicali}. Domains in which probabilistic
models play a key role include robotics~\cite{thrun2002robotic},
linguistics~\cite{manning1999foundations}, and especially machine
learning~\cite{pearl1988probabilistic}.  The wealth of applications
has stimulated the development of concrete and abstract programming
languages, that most often are extensions of their deterministic
ancestors. Among the many ways probabilistic choice can be captured in
programming, a simple one consists in endowing the language of
programs with an operator modelling the sampling from (one or many)
distributions. This renders program evaluation a probabilistic
process, and under mild assumptions the language becomes universal for
probabilistic computation. Particularly fruitful in this sense has
been the line of work in the functional paradigm.

In \emph{probabilistic programming}, programs become a way to specify
probabilistic models for observed data, on top of which one can later
do inference. This has been a source of inspiration for AI
researchers, and has recently been gathering interest in the
programming language community (see \citet{DBLP:conf/popl/Goodman13}, \citet{DBLP:conf/icse/GordonHNR14}, and 
\citet{DBLP:journals/cacm/Russell15}).
%%%%%%%%%%%%%%%%%%%%%%%%%%%%%%%%%%%%%%%%%%%%%%%%%%%%%%%%%%%
\subsection{Universal Probabilistic Programming in \Church}\label{sec:background}
%%%%%%%%%%%%%%%%%%%%%%%%%%%%%%%%%%%%%%%%%%%%%%%%%%%%%%%%%%%
\Church\ \cite{DBLP:conf/uai/GoodmanMRBT08} introduced \emph{universal
  probabilistic programming}, the idea of writing probabilistic models
for machine learning in a Turing-complete functional programming
language. \Church, and its descendants
\Venture\ \cite{DBLP:journals/corr/MansinghkaSP14},
\Anglican\ \cite{DBLP:conf/pkdd/TolpinMW15}, and
\WebChurch\ \cite{probmods} are dialects of \Scheme. Another example
of universal probabilistic programming is \WebPPL~\cite{dippl}, a
probabilistic interpretation of \JavaScript.

A probabilistic query in \Church\ has the form:
\begin{lstlisting}
  (query (define $x_1$ $e_1$)...(define $x_n$ $e_n$) $e_q$ $e_c$)
\end{lstlisting}
The query denotes the distribution given by the probabilistic
expression~$e_q$, given variables~$x_i$ defined by potentially
probabilistic expressions~$e_i$, constrained so that the boolean
predicate~$e_c$ is true.

Consider a coin with bias~$p$, that is,~$p$ is the probability of
heads. Recall that the \emph{geometric distribution} of the coin is
the distribution over the number of flips in a row before it comes up
heads. An example of a \Church\ query is as follows: it denotes the
geometric distribution for a fair coin, constrained to be greater than one.
\begin{lstlisting}
  (query
    (define flip (lambda (p) (< (rnd) p)))
    (define geometric (lambda (p)
      (if (flip p) 0 (+ 1 (geometric p))))
    (define n (geometric .5))
    n
    (> n 1))
\end{lstlisting}
The query defines three variables: (1) \lstinline{flip} is a function
that flips a coin with bias \lstinline{p}, by calling
\lstinline{(rnd)} to sample a probability from the uniform
distribution on the unit interval; (2)
\lstinline{geometric}\footnote{See
  \url{http://forestdb.org/models/geometric.html}.} is a function that
samples from the geometric distribution of a coin with bias
\lstinline{p}; and (3) \lstinline{n} denotes the geometric
distribution with bias \lstinline{0.5}.  Here are samples from this
query:
\begin{lstlisting}
  (5 5 5 4 2 2 2 2 2 3 3 2 2 7 2 2 3 4 2 3)
\end{lstlisting}
This example is a discrete distribution with unbounded support (any
integer greater than one may be sampled with some non-zero
probability), defined in terms of a continuous distribution (the
uniform distribution on the unit interval). Queries may also define
continuous distributions, such as regression parameters.

%For example, the following query defines a probability \lstinline{p}
%at random, defines a function to flip a coin with bias \lstinline{p},
%conditions the model on single observations of \lstinline{0} and
%\lstinline{1}, and returns a representation of the distribution of
%\lstinline{p}.
%\begin{lstlisting}
%  (query
%    (define p (rnd))
%    (define flip (lambda (x) (< (rnd) p)))
%    p
%    (condition
%      (and (eq (flip) 0) (eq (flip) 1)))))
%\end{lstlisting}

%For example, the following program defines a probability \lstinline{p}
%at random, defines a function to flip a coin with bias \lstinline{p},
%conditions the model on single observations of \lstinline{0} and
%\lstinline{1}, and returns \lstinline{p}. We use \lstinline{uniform()}
%to sample a probability from the uniform distribution on the unit
%interval.
%\begin{lstlisting}[language=caml]
%let p = uniform()
%let flip() = uniform() < p
%if (flip() = 0) and (flip() = 1) 
%then p
%else fail
%\end{lstlisting}

%%%%%%%%%%%%%%%%%%%%%%%%%%%%%%%%%%%%%%%%%%%%%%%%%%%%%%%%%%%%%%
\subsection{Problem 1: Semantics of \Church\ Queries}
%%%%%%%%%%%%%%%%%%%%%%%%%%%%%%%%%%%%%%%%%%%%%%%%%%%%%%%%%%%%%%
The first problem we address in this work is to provide a formal
semantics for universal probabilistic programming languages with
constraints.  Our example illustrates the common situation in machine
learning that models are based on continuous distributions (such as
\lstinline{(rnd)}) and use constraints, but previous works on formal
semantics for untyped probabilistic~$\lambda$-calculi do not rigorously treat
the combination of these features.

To address the problem we introduce a call-by-value~$\lambda$-calculus with
primitives for random draws from various continuous distributions,
and primitives for both hard and soft constraints.
We present an encoding of \Church\ into our calculus, and some nontrivial examples of
probabilistic models.

We consider two styles of operational semantics for our
$\lambda$-calculus, in which a term is interpreted in two ways, the
first closer to inference techniques, the second more extensional:
\begin{description}
\item[Sampling-Based:] A function from a trace to a value and weight.
\item[Distribution-Based:] A distribution over terms of our calculus.
\end{description}
To obtain a thorough understanding of the semantics of the calculus,
for each of these styles we present two inductive definitions of
operational semantics, in small-step and big-step style.

First, we consider the \emph{sampling-based semantics}: the two
inductive definitions have the forms shown below, where $M$ is a
closed term, $s$ is a finite \emph{trace} of random real numbers, $w >
0$ is a \emph{weight} (to impose soft constraints), and $G$ is a \emph{generalized value}
(either a value (constant or $\lambda$-abstraction) or the exception
\lstinline{fail}, used to model a failing hard constraint).

\begin{varitemize}
\item Figure~\ref{fig:small-step-sampling} defines small-step relation $(M,w,s) \rightarrow (M',w',s')$.
\item Figure~\ref{fig-sampling} defines the big-step relation $M \Downarrow^s_w G$.
\end{varitemize}

For example, if $M$ is the $\lambda$-term for our geometric distribution example
and we have $M \Downarrow^s_w G$ then there is $n \geq 0$ such that:
\begin{varitemize}
\item 
  the trace has the form $s=[q_1,\dots,q_{n+1}]$ where each $q_i$ is a probability, and $q_i < 0.5$
  if and only if $i=n+1$. (A sample $q_i \geq 0.5$ is tails; a sample $q_i < 0.5$ is heads.)
\item 
  the result takes the form $G=n$ if $n>1$, and otherwise
  $G=\mathtt{fail}$ (the failure of a hard constraint leads to $\mathtt{fail}$);
\item 
  and the weight is $w=1$ (the density of the uniform distribution on the unit interval).
\end{varitemize}

Our first result, Theorem~\ref{thm:sample-small-big-eq}, shows equivalence: that the big-step
and small-semantics of a term consume the same traces to produce the
same results with the same weights.

To interpret these semantics probabilistically, we describe a metric space of
$\lambda$-terms and let~$\distrone$ range over \emph{distributions},
that is, sub-probability Borel measures on terms of the
$\lambda$-calculus. We define~$\sbd{\termone}$ to be the distribution
induced by the sampling-based semantics of~$M$, 
by integrating the weight over the space of traces.

Second, we consider the \emph{distribution-based semantics}, that
directly associate distributions with terms, without needing to
integrate out traces. The two inductive definitions have the forms
shown below, where $n$ is a step-index:
\begin{varitemize}
\item 
  Figure~\ref{fig:smallstepsem} defines a family of small-step
  relations $\sisss{\termone}{n}{\distrone}$.
\item 
  Figure~\ref{fig:bigstepsem} defines a family of big-step relations
  $\sibss{\termone}{n}{\distrone}$.
\end{varitemize}
These step-indexed families are approximations to their suprema,
distributions written as $\sts{\termone}$ and $\bts{\termone}$. By
Theorem~\ref{thm:eq} we have $\sts{\termone} = \bts{\termone}$. The
proof of the theorem needs certain properties
(Lemmas~\ref{lemma:sstk}, \ref{lemma:derapp}, and \ref{lemma:invapp})
that build on compositionality results for sub-probability kernels
\cite{panangaden99markovkernels} from the measure theory
literature. We apply the distribution-based semantics in
Section~\ref{sec:app} to show an equation between hard and soft constraints.

Finally, we reconcile the two rather different styles of semantics:
Theorem~\ref{thm:sampling-distribution} establishes
that~$\sbd{\termone} = \sts{\termone}$.

%%%%%%%%%%%%%%%%%%%%%%%%%%%%%%%%%%%%%%%%%%%%%%%%%%%%
\subsection{Problem 2: Correctness of Trace MCMC}
%%%%%%%%%%%%%%%%%%%%%%%%%%%%%%%%%%%%%%%%%%%%%%%%%%%%

The second problem we address is implementation correctness. As recent
work shows \cite{corsamp,problems-kiselov16}, subtle errors in
inference algorithms for probabilistic languages are a motivation for
correctness proofs for probabilistic inference.

\emph{Markov chain Monte Carlo} (MCMC) is an important class of
inference methods, exemplified by the Metropolis-Hastings (MH)
algorithm \cite{metropolis53,hastings70:MH}, that accumulates samples
from a target distribution by exploring a Markov chain. The original
work on \Church\ introduced the implementation technique called
\emph{trace MCMC} \cite{DBLP:conf/uai/GoodmanMRBT08}. Given
a closed term~$M$, trace MCMC generates a Markov chain of
traces,~$s_0$, $s_1$, $s_2$, \dots. The MH algorithm is parametric in a 
\emph{proposal kernel}~$Q$: 
a function that maps a trace~$s$ of~$M$ to a probability distribution over traces, 
used to sample the next trace in the Markov chain.

Our final result, Theorem~\ref{thm:sampling-trace}, asserts that the
Markov chain generated by trace MCMC converges to a stationary
distribution, and that the induced distribution on values is equal to
the semantics~$\sts{\termone}$ conditional on success, that is, that
the computation terminates and yields a value (not
$\mathtt{fail}$). We formalize the algorithm rigorously, defining our
proposal kernel as a Lebesgue integral of a corresponding density
function, that we show to be measurable with respect to
the~$\sigma$-algebra on program traces. We
show that the resulting Markov chain satisfies standard criteria: 
\emph{aperiodicity} and \emph{irreducibility}. 
Hence, Theorem~\ref{thm:sampling-trace} follows from 
a classic result of \citet{tierney1994} 
together with Theorem~\ref{thm:sampling-distribution}.
%, and the sampling-based semantics, 
%which formalizes the traces used by trace MCMC.

%%%%%%%%%%%%%%%%%%%%%%%%%%%%%%%%%%%%%%%
\subsection{Contributions of the Paper}
%%%%%%%%%%%%%%%%%%%%%%%%%%%%%%%%%%%%%%%
We make the following original contributions:
\begin{varenumerate}
\item
  Definition of an untyped~$\lambda$-calculus with continuous
  distributions capable of encoding the core of \Church.
\item
  Development of both sampling-based and distribution-based semantics,
  shown equivalent (Theorems~\ref{thm:sample-small-big-eq},
  \ref{thm:eq}, and~\ref{thm:sampling-distribution}).
\item
  First proof of correctness of trace MCMC for a $\lambda$-calculus
  (Theorem~\ref{thm:sampling-trace}).
\end{varenumerate}

The only previous work on formal semantics of~$\lambda$-calculi with
constraints and continuous distributions is recent work by
\citet{DBLP:journals/corr/StatonYHKW16}.
Their main contribution is
an elegant denotational semantics for a simply typed
$\lambda$-calculus with continuous distributions and both hard and
soft constraints, but without recursion.  They do not consider MCMC
inference.  Their work does not apply to the recursive functions (such
as the geometric distribution in Section~\ref{sec:background}) or data
structures (such as lists) typically found in \Church\ programs.  For
our purpose of conferring formal semantics on \Church-family
languages, we consider it advantageous to rely on untyped techniques.

The only previous work on correctness of trace MCMC, and an important
influence on our work, is a recent paper by \citet{corsamp} which
proves correct an algorithm for computing an MH Markov chain.  Key
differences are that we work with higher-order languages and soft
constraints, and that we additionally give a proof that our Markov chain always converges, 
via the correctness criteria of \citet{tierney1994}.

%%%%%%%%%%%%%%%%%%%%%%%%%%%%%%%%%%%
\subsection{Structure of the Paper}
%%%%%%%%%%%%%%%%%%%%%%%%%%%%%%%%%%%
The rest of the paper is organized as follows.

Section~\ref{sec:syntax} defines the syntax of our probabilistic
$\lambda$-calculus with draws from continuous distributions, and
defines a deterministic \emph{sampling-based} operational semantics for our
calculus. The semantics is based on the explicit consumption of a
program trace~$s$ of random draws and production of an explicit weight
$w$ for each outcome. 

\ADG{why useful to have both big-step and small-step?}
Section~\ref{sect:samplingsemantics} is concerned
with a more in-depth treatment of sampling-based semantics,
given in two standard styles: big-step
semantics, $M \Downarrow^s_w G$, and small-step semantics, $(M,w,s)
\rightarrow (M',w',s')$, which are equivalent by Theorem~\ref{thm:sample-small-big-eq}.
We define a distribution $\tsq{\termone}$ on outcomes of $\termone$ by integrating
the weights with respect to a measure on traces and applying a measure transformation.

\ADG{ditto, why useful to have both big-step and small-step?}
Section~\ref{sec:operational-semantics}
defines our step-indexed \emph{distribution-based} operational
semantics, in both small-step ($\sisss{\termone}{n}{\distrone}$) and big-step
($\sibss{M}{n}{\distrone}$) styles, which by Theorem~\ref{thm:eq} are
equivalent, and define the meaning~$\tsq{\termone}$ of a term~$M$ to be the supremum of the step-indexed semantics.
%Section~\ref{sec:measure-space-on-program-traces} makes an important
%auxiliary construction, a measure space on the set of program traces,
%that is, sequences of reals of arbitrary length.
We end by linking the semantics of this section with those of Section~\ref{sect:samplingsemantics}:
Theorem~\ref{thm:sampling-distribution} establishes that~$\sbd{\termone} = \tsq{\termone}$

Section~\ref{sec:inference}
formalizes trace MCMC for our calculus, in the spirit of \citet{corsamp}.
Theorem~\ref{thm:sampling-trace} shows equivalence between
the distribution computed by the algorithm and the semantics of the
previous sections. Hence, Theorem~\ref{thm:sampling-trace} is the
first correctness theorem for trace MCMC for a
$\lambda$-calculus.

Section~\ref{sec:related} describes related work
and Section~\ref{sec:conc} concludes.
% removed for sake of anonymity \cite{mhlambda-arxiv} 

Appendix~\ref{section:proof-of-measurability} collects various proofs of measurability.

%%%%%%%%%%%%%%%%%%%%%%%%%%%%%%%%%%%%%%%%%%%%
\section{A Foundational Calculus for \Church}\label{sec:syntax}
%%%%%%%%%%%%%%%%%%%%%%%%%%%%%%%%%%%%%%%%%%%%
\newcommand{\ertone}{T}
\newcommand{\erttwo}{R}

In this section, we describe the syntax of our calculus and equip it with
an intuitive semantics relating program outcomes to the sequences
of random choices made during evaluation. By translating \Church\ constructs
to this calculus, we show that it serves as a foundation for Turing-complete
probabilistic languages.

To simplify the presentation, we do not include primitives for
discrete distributions in the calculus, as they can be encoded 
using the uniform distribution on the unit interval and inverse 
mass functions. However, it would be easy to extend the calculus with
primitive discrete random distributions, represented by their probability 
mass functions rather than densities. In Section~\ref{sec:measure-space-on-program-traces}, we 
explain how the semantics could be adapted in this case.

%\ADG{need to define ${\cal M}$}
%%%%%%%%%%%%%%%%%%%%%%%%%%%%%%%%%%%
\subsection{Syntax of the Calculus}
%%%%%%%%%%%%%%%%%%%%%%%%%%%%%%%%%%%
We represent scalar data as real numbers~$c \in \mathbb{R}$. We use
$0$ and~$1$ to represent~$\tfalse$ and~$\ttrue$, respectively.
%Let~$\expset$ be a set of exception identifiers, ranged over by
%metavariables like~$\expone$ and~$\exptwo$. We assume that
%$\fail,\terror\in\expset$.  The exception~$\terror$ represents a
%run-time error (such as applying a constant as if it were a function)
%while the exception~$\fail$ represents conditioning of the
%distribution defined by a term.  \ADG{what would go wrong formally if
%  we had~$\fail=\terror$?}
Let~$\distset$ be a countable set of \emph{distribution identifiers}
(or simply \emph{distributions}). Metavariables for distributions are
$\distone,\disttwo$. Each distribution identifier~$\distone$ has an
integer \emph{arity}~$\Abs{\distone} \ge 0$, and defines a density function
$\operatorname{pdf}_{\distone}:\Real^{\Abs{\distone}+1}\to[0,\infty)$
of a sub-probability kernel.
%$\mu_{\distone(\vec{c})}:\restr{\measterms}{\RR}\to\RRp$
For example, a draw (\textsf{rnd}()) from the uniform distribution on the unit interval has density
$\operatorname{pdf}_{\textsf{rnd}}(c)=1$ if~$c\in[0,1]$ and otherwise 0,
while a draw \textsf{(Gaussian}$(m,v)$) 
from the Gaussian distribution with mean~$m$ and variance~$v$
has density $\operatorname{pdf}_{\textsf{Gaussian}}(m,v,c)=1/(e^{\frac{(c-m)^{2}}{2v}}\sqrt{2v\pi})$ 
if~$v>0$ and otherwise 0.

Let~$g$ be a metavariable ranging over a countable set of
\emph{function identifiers} each with an integer \emph{arity}~$|g| >
0$ and with an interpretation as a total measurable function
$\intfun{g}: \mathbb{R}^{|g|} \to \mathbb{R}$.  Examples of function
identifiers include addition~$+$, comparison~$>$, and equality~$=$; 
they are often written in infix notation.
We define the \emph{values}~$V$ and \emph{terms}~$M$ as follows, where
$\varone$ ranges over a denumerable set of variables~$\varset$.
  \begin{align*}
    V &\bnf c \midd x \midd  \lambda x.M \\
    M, N &\bnf V \midd M\ N \midd \distone(V_1,\dots,V_{|\distone|}) \midd g(V_1,\dots,V_{|g|})\\
    & \midd \ite{V}{M}{N} \midd \score{V} \midd \fail
  \end{align*}
The term~$\fail$ acts as an exception and models a failed hard constraint.
The term~$\score{c}$ models a soft constraint, 
and is parametrized on a positive probability~$c\in(0,1]$.
As usual, free occurrences of~$\varone$ inside~$\termone$ are bound by
$\abstr{\varone}{\termone}$. Terms are taken modulo renaming of bound
variables. Substitution of all free occurrences of~$\varone$ by a
value~$\valone$ in~$\termone$ is defined as usual, and denoted
$\sbst{\termone}{\varone}{\valone}$. This can be easily generalized to
$\sbst{\termone}{\vec{\varone}}{\vec{\valone}}$, where~$\vec{\varone}$
is a sequence of variables and~$\vec{\valone}$ is a sequence of values
(of the same length). Let $\lamterms$ denote the set of all terms,
and~$\cterms$ the set of \emph{closed} terms.  The set of all
\emph{closed values} is $\valset$, and we write $\valset_\lambda$ for
$\valset\setminus\RR$. \emph{Generalized values}~$\gvalone$,
$\gvaltwo$ are elements of the set~$\gvalset=\valset\cup\{\fail \}$,
\ie, generalized values are either values or~$\fail$. Finally,
\emph{erroneous redexes}, ranged over by metavariables like
$\ertone,\erttwo$, are closed terms in one of the following five
forms:
\begin{varitemize}
\item
  $c\ \termone$.
\item
  $\distone(V_1,\dots,V_{|\distone|})$ where at least one of the $V_i$ is
  a $\lambda$-abstraction.
\item
  $g(V_1,\dots,V_{|g|})$ where at least one of the $V_i$ is
  a $\lambda$-abstraction.
\item
  $\ite{V}{M}{N}$, where $V$ is neither $\ttrue$ nor 
  $\tfalse$.
\item
  $\mathtt{score}(V)$, where $V \notin (0,1]$.
\end{varitemize}
%Erroneous redexes are.

\subsection{Big-step Sampling-based Semantics}\label{sect:bssbs}
In defining the first semantics of the calculus, we use the classical
observation~\cite{DBLP:conf/focs/Kozen79} that a probabilistic program
can be interpreted as a deterministic program parametrized by the
sequence of random draws made during the evaluation. We write $M
\Downarrow^s_w V$ to mean that evaluating~$M$ with the outcomes of
random draws as listed in the sequence~$s$ yields the value~$V$,
together with the \emph{weight}~$w$ that expresses how likely this
sequence of random draws would be if the program was just evaluated
randomly. Because our language has continuous distributions,~$w$ is a
probability \emph{density} rather than a probability mass.  Similarly,
$M \Downarrow^s_w \fail$ means that evaluation of~$M$ with the random
sequence~$s$ fails. In either case, the finite trace $s$ consists of
exactly the random choices made during evaluation, with no unused
choices permitted.

\newcommand{\tr}[1]{\textsc{(#1)}}
Formally, we define \emph{program traces}~$s$, $t$ to be finite
sequences $[c_1,\dots,c_n]$ of reals of arbitrary length. 
We let~$M \Downarrow^s_w G$ be the least relation closed under the rules in
Figure~\ref{fig-sampling}.  
The \tr{Eval Random} rule replaces a random draw from a distribution~$\distone$ 
parametrized by a vector~$\vec{c}$ with the first (and only) element~$c$ of the trace,
presumed to be the outcome of the random draw, 
and sets the weight to the value of the density of~$\distone(\vec{c})$ at~$c$. 
\tr{Eval Random Fail} throws an exception if~$c$ is outside the support of
the corresponding distribution. 
Meanwhile, \tr{Eval Score}, applied to~$\mathtt{score}(c)$, 
sets the weight to~$c$ and returns a dummy value.  
The applications of soft constraints using $\mathtt{score}$ 
are described in Section~\ref{sec:soft-conditioning}.
%\MS{Describe why weights are multiplied?}
\begin{SHORT}
\begin{figure*}
\begin{center}
\fbox{
\begin{minipage}{.97\textwidth}
\vspace{5pt}
\[
\AxiomC{$\gvalone \in \gvalset$}
\UnaryInfC{$\gvalone \Downarrow^{[]}_1 \gvalone$}
\DisplayProof{\;\tr{Eval Val}}
\qquad
\AxiomC{$w = \pdf{\distone}(\vec{c}, c)$}
\AxiomC{$w > 0$}
\BinaryInfC{$\distone (\vec{c}) \Downarrow^{[c]}_{w} c$}
\DisplayProof{\;\tr{Eval Random}}
\]
\[
\AxiomC{$\pdf{\distone}(\vec{c}, c)=0$}
\UnaryInfC{$\distone (\vec{c}) \Downarrow^{[c]}_{0} \fail$}
\DisplayProof{\;\tr{Eval Random Fail}}
\qquad
\AxiomC{}
\UnaryInfC{$g(\vec{c}) \Downarrow^{[]}_{1} \intfun{g}(\vec{c})$}
\DisplayProof{\;\tr{Eval Prim}}
\]
\[
\AxiomC{$M \Downarrow^{s_1}_{w_1} \lambda x . P$}
\AxiomC{$N \Downarrow^{s_2}_{w_2} V$}
\AxiomC{$P[V/x] \Downarrow^{s_3}_{w_3} G$}
\TrinaryInfC{$M\ N \Downarrow^{s_1@s_2@s_3}_{w_1\cdot w_2\cdot w_3} G$}
\DisplayProof{\;\tr{Eval Appl}}
\qquad
\AxiomC{$M \Downarrow^{s}_{w} \fail$}
\UnaryInfC{$M\ N \Downarrow^s_{w} \fail$}
\DisplayProof{\;\tr{Eval Appl Raise1}}
\]
\[
\AxiomC{$M \Downarrow^{s}_{w} c$}
\UnaryInfC{$M\ N \Downarrow^s_{w} \fail$}
\DisplayProof{\;\tr{Eval Appl Raise2}}
\qquad
\AxiomC{$M \Downarrow^{s_1}_{w_1} \lambda x . P$}
\AxiomC{$N \Downarrow^{s_2}_{w_2} \fail$}
\BinaryInfC{$M\ N \Downarrow^{s_1@s_2}_{w_1\cdot w_2} \fail$}
\DisplayProof{\;\tr{Eval Appl Raise3}}
\]
\[
\AxiomC{$M\Downarrow^{s}_{w} G$}
\UnaryInfC{$\ite{\ttrue}{M}{N} \Downarrow^{s}_{w} G$}
\DisplayProof{\;\tr{Eval If True}}
\qquad
\AxiomC{$N\Downarrow^{s}_{w} G$}
\UnaryInfC{$\ite{\tfalse}{M}{N} \Downarrow^{s}_{w} G$}
\DisplayProof{\;\tr{Eval If False}}
\]
\[
  \AxiomC{$c \in (0, 1]$}
  \UnaryInfC{$\mathtt{score}(c) \Downarrow^{[]}_c \ttrue$}
  \DisplayProof{\;\tr{Eval Score}}
\qquad
  \AxiomC{\mbox{$T$ is an erroneous redex}}
  \UnaryInfC{$T \Downarrow^{[]}_1 \fail$}
  \DisplayProof{\;\tr{Eval Fail}}
\]
\vspace{5pt}
\end{minipage}}
\condnocr
\caption{Sampling-Based Big Step Semantics}\label{fig-sampling}
\end{center}
\end{figure*}
\end{SHORT}

\begin{LONG}
\begin{figure*}
\begin{center}
\fbox{
\begin{minipage}{.97\textwidth}
\vspace{5pt}
\[
\AxiomC{$\gvalone \in \gvalset$}
\UnaryInfC{$\gvalone \Downarrow^{[]}_1 \gvalone$}
\DisplayProof{\;\tr{Eval Val}}
\qquad
\AxiomC{$w = \pdf{\distone}(\vec{c}, c)$}
\AxiomC{$w > 0$}
\BinaryInfC{$\distone (\vec{c}) \Downarrow^{[c]}_{w} c$}
\DisplayProof{\;\tr{Eval Random}}
\]
\[
\AxiomC{$\pdf{\distone}(\vec{c}, c)=0$}
\UnaryInfC{$\distone (\vec{c}) \Downarrow^{[c]}_{0} \fail$}
\DisplayProof{\;\tr{Eval Random Fail}}
\qquad
\AxiomC{}
\UnaryInfC{$g(\vec{c}) \Downarrow^{[]}_{1} \intfun{g}(\vec{c})$}
\DisplayProof{\;\tr{Eval Prim}}
\]
\[
\AxiomC{$M \Downarrow^{s_1}_{w_1} \lambda x . P$}
\AxiomC{$N \Downarrow^{s_2}_{w_2} V$}
\AxiomC{$P[V/x] \Downarrow^{s_3}_{w_3} G$}
\TrinaryInfC{$M\ N \Downarrow^{s_1@s_2@s_3}_{w_1\cdot w_2\cdot w_3} G$}
\DisplayProof{\;\tr{Eval Appl}}
\]
\[
\AxiomC{$M \Downarrow^{s}_{w} \fail$}
\UnaryInfC{$M\ N \Downarrow^s_{w} \fail$}
\DisplayProof{\;\tr{Eval Appl Raise1}}
\qquad
\AxiomC{$M \Downarrow^{s}_{w} c$}
\UnaryInfC{$M\ N \Downarrow^s_{w} \fail$}
\DisplayProof{\;\tr{Eval Appl Raise2}}
\]
\[
\AxiomC{$M \Downarrow^{s_1}_{w_1} \lambda x . P$}
\AxiomC{$N \Downarrow^{s_2}_{w_2} \fail$}
\BinaryInfC{$M\ N \Downarrow^{s_1@s_2}_{w_1\cdot w_2} \fail$}
\DisplayProof{\;\tr{Eval Appl Raise3}}
\]
\[
\AxiomC{$M\Downarrow^{s}_{w} G$}
\UnaryInfC{$\ite{\ttrue}{M}{N} \Downarrow^{s}_{w} G$}
\DisplayProof{\;\tr{Eval If True}}
\]
\[
\AxiomC{$N\Downarrow^{s}_{w} G$}
\UnaryInfC{$\ite{\tfalse}{M}{N} \Downarrow^{s}_{w} G$}
\DisplayProof{\;\tr{Eval If False}}
\]
\[
  \AxiomC{$c \in (0, 1]$}
  \UnaryInfC{$\mathtt{score}(c) \Downarrow^{[]}_c \ttrue$}
  \DisplayProof{\;\tr{Eval Score}}
\qquad
  \AxiomC{\mbox{$T$ is an erroneous redex}}
  \UnaryInfC{$T \Downarrow^{[]}_1 \fail$}
  \DisplayProof{\;\tr{Eval Fail}}
\]
\vspace{5pt}
\end{minipage}}
\condnocr
\caption{Sampling-Based Big Step Semantics}\label{fig-sampling}
\end{center}
\end{figure*}
\end{LONG}

All the other rules are standard for a call-by-value lambda-calculus,
except that they allow the traces to be split between subcomputations
and they multiply the weights yielded by subcomputations to obtain the
overall weight.
%%%%%%%%%%%%%%%%%%%%%%%%%%%%
\subsection{Encoding \Church}
%%%%%%%%%%%%%%%%%%%%%%%%%%%%
We now demonstrate the usefulness and expressive power of the calculus
via a translation of \Church, an untyped higher-order functional
probabilistic language.

The syntax of \Church's \emph{expressions}, \emph{definitions} and 
\emph{queries} is described as follows:
\newcommand{\tw}[1]{\texttt{#1}}
\newcommand{\llet}[3]{\mathtt{let}\ #1 = #2\ \mathtt{in}\ #3}
\newcommand{\lfix}[4]{\mathtt{let}\ #1 = \fixop #1\ \lambda #2. #3\ \mathtt{in}\ #4}
\newcommand{\trce}[1]{\langle #1 \rangle_e}
\newcommand{\trc}[1]{\langle #1 \rangle}
\newcommand{\fv}[1]{\mathit{fv}(#1)}
\newcommand{\fixop}{\mathit{fix}}
\begin{align*}
e\;::=\;&
  c\midd
  x\midd
  \tw{(}g\ e_1\dots  e_n\tw{)} \midd
  \tw{(}\distone\ e_1\dots e_n\tw{)}\midd
  \tw{(if}\ e_1\ e_2\ e_3\tw{)}
\\
%\midd \tw{(}\distone\ e_1\dots e_n\ e_{obs}\tw{)}\\
  \midd& \tw{(lambda}\ (x_1 \dots x_n)\ e\tw{)} 
  \midd \tw{(}e_1\ e_2\dots e_n\tw{)}
  \\
\tw{d}\;::=\;&\tw{(define}\ x\ e\tw{)}\\
\tw{q}\;::=\;&\tw{(query}\ d_1 \dots d_n\ e\  e_{cond}\tw{)}
\end{align*}
To make the translation more intuitive, it is convenient to add to
the target language a let-expression of the form~$\llet{x}{M}{N}$,
that can be interpreted as syntactic sugar for~$(\lambda x . N)\ M$,
and sequencing $M;N$ that stands for $\lambda\!\star\!.N\ M$ where $\star$
as usual stands for a variable that does not appear free in any of the
terms under consideration.

\begin{SHORT}
\begin{figure}
\begin{center}
\fbox{
\begin{minipage}{.442\textwidth}
  \vspace{5pt}
  \footnotesize
\[
\begin{prog}
   \trce{c} = c \\
   \trce{x} = x \\
   \trce{g\ e_1, \dots, e_n} =\\
   \qquad  \llet{x_1}{e_1}{\dots \llet{x_n}{e_n}{g(x_1, \dots, x_n)} }\\
   \qquad \text{where}\ x_1, \dots, x_n \notin \fv{e_1} \cup \dots \cup \fv{e_n}\\
   \trce{\distone\ e_1, \dots e_n} = \\
   \qquad \llet{x_1}{e_1}{\dots \llet{x_n}{e_n}{\distone(x_1, \dots, x_n)} } \\
   \qquad \text{where}\ x_1, \dots, x_n \notin \fv{e_1} \cup \dots \cup \fv{e_n}\\
%   \trce{\distone\ e_1, \dots e_n\ e_{obs}} = \\
%   \qquad \llet{x_1}{e_1}{\dots \llet{x_n}{e_n}{\llet{y}{e_{obs}}{}}\\
%   \qquad \quad {\mathtt{score}(\pdf{\distone}(x_1, \dots, x_n, y))}} \\
%   \qquad \text{where}\ x_1, \dots, x_n \notin \fv{e_1} \cup \dots \cup \fv{e_n}\\
   \trce{\mathtt{lambda}\ ()\ e} = \lambda x . \trce{e} 
     \quad\text{where}\ x \notin \fv{e} \\
   \trce{\mathtt{lambda}\ x\ e} = \lambda x . \trce{e} \\
   \trce{\mathtt{lambda}\ (x_1\ \dots\ x_n)\ e} = \lambda x_1 . \trce{\mathtt{lambda}\ (x_2\ \dots\ x_n)\ e } \\
   \trce{e_1\ e_2} = \trce{e_1}\ \trce{e_2} \\
   \trce{e_1\ e_2\ \dots\ e_n} = \trce{(e_1\ e_2)\ \dots\ e_n} \\
   \trce{\mathtt{if}\ e_1\ e_2\ e_3} = \llet{x}{e_1}{(\ite{x}{\trce{e_2}}{\trce{e_3}})} \\
  \qquad \text{where}\ x \notin \fv{e_2} \cup \fv{e_3} \\
  \\
  \trc{\mathtt{query}\ (\mathtt{define}\ x_1\ e_1) \dots 
  (\mathtt{define}\ x_n\ e_n)\ e_{out}\ e_{cond}}=\\
   \qquad \llet{x_1}{(\fixop\ x_1 . \trce{e_1}) }{ }\\
   \qquad \dots \\
   \qquad \llet{x_n}{(\fixop\ x_n . \trce{e_n}) }{ } \\
   \qquad \llet{b}{e_{cond}}{ }\\
   \qquad \ite{b}{e_{out}}{\fail}
\end{prog}
\]
\vspace{5pt}
\end{minipage}
}
\condnocr
\caption{Translation of \Church}\label{fig:translation}
\end{center}
\end{figure}
\end{SHORT}

\begin{LONG}
\begin{figure}
\begin{center}
\fbox{
\begin{minipage}{.520\textwidth}
  \vspace{5pt}
  \footnotesize
\[
\begin{prog}
   \trce{c} = c \\
   \trce{x} = x \\
   \trce{g\ e_1, \dots, e_n} =\\
   \qquad  \llet{x_1}{e_1}{\dots \llet{x_n}{e_n}{g(x_1, \dots, x_n)} }\\
   \qquad \text{where}\ x_1, \dots, x_n \notin \fv{e_1} \cup \dots \cup \fv{e_n}\\
   \trce{\distone\ e_1, \dots e_n} = \\
   \qquad \llet{x_1}{e_1}{\dots \llet{x_n}{e_n}{\distone(x_1, \dots, x_n)} } \\
   \qquad \text{where}\ x_1, \dots, x_n \notin \fv{e_1} \cup \dots \cup \fv{e_n}\\
%   \trce{\distone\ e_1, \dots e_n\ e_{obs}} = \\
%   \qquad \llet{x_1}{e_1}{\dots \llet{x_n}{e_n}{\llet{y}{e_{obs}}{}}\\
%   \qquad \quad {\mathtt{score}(\pdf{\distone}(x_1, \dots, x_n, y))}} \\
%   \qquad \text{where}\ x_1, \dots, x_n \notin \fv{e_1} \cup \dots \cup \fv{e_n}\\
   \trce{\mathtt{lambda}\ ()\ e} = \lambda x . \trce{e} 
     \quad\text{where}\ x \notin \fv{e} \\
   \trce{\mathtt{lambda}\ x\ e} = \lambda x . \trce{e} \\
   \trce{\mathtt{lambda}\ (x_1\ \dots\ x_n)\ e} = \lambda x_1 . \trce{\mathtt{lambda}\ (x_2\ \dots\ x_n)\ e } \\
   \trce{e_1\ e_2} = \trce{e_1}\ \trce{e_2} \\
   \trce{e_1\ e_2\ \dots\ e_n} = \trce{(e_1\ e_2)\ \dots\ e_n} \\
   \trce{\mathtt{if}\ e_1\ e_2\ e_3} = \llet{x}{e_1}{(\ite{x}{\trce{e_2}}{\trce{e_3}})} \\
  \qquad \text{where}\ x \notin \fv{e_2} \cup \fv{e_3} \\
  \\
  \trc{\mathtt{query}\ (\mathtt{define}\ x_1\ e_1) \dots 
  (\mathtt{define}\ x_n\ e_n)\ e_{out}\ e_{cond}}=\\
   \qquad \llet{x_1}{(\fixop\ x_1 . \trce{e_1}) }{ }\\
   \qquad \dots \\
   \qquad \llet{x_n}{(\fixop\ x_n . \trce{e_n}) }{ } \\
   \qquad \llet{b}{e_{cond}}{ }\\
   \qquad \ite{b}{e_{out}}{\fail}
\end{prog}
\]
\vspace{5pt}
\end{minipage}
}
\condnocr
\caption{Translation of \Church}\label{fig:translation}
\end{center}
\end{figure}
\end{LONG}

The rules for translating \Church\ expressions to the calculus are
shown in Figure \ref{fig:translation}, where~$\fv{e}$ denotes the set
of free variables in expression~$e$ and $\fixop\ x.\termone$ is a call-by-value fixpoint combinator
$\lambda y.\termtwo_\fixop \termtwo_\fixop (\lambda x.\termone) y$
where $\termtwo_\fixop$ is $\lambda z.\lambda w.w(\lambda y.((zz)w)y)$.
Observe that $(\fixop\ x.\termone)\valone$ evaluates
to $\termone\{(\fixop\ x.\termone)/x\}\valone$ deterministically.
We assume that for each distribution identifier $\distone$ of arity $k$, there
is a deterministic function $\pdf{\distone}$ of arity $k+1$ that calculates the
corresponding density at the given point.
%
%  First some syntactic sugar
%  \begin{display}[0.5]{Syntactic sugar in the calculus}
%    \clause{\llet{x}{M}{N} \deq (\lambda x . N)\ M}
%  \end{display}
  
%  
%  Note that all the variables introduced on the right-hand side are assumed to be fresh.
%  In the implementation the problem of having to generate unique variable names can be
%  circumvented by using de Bruijn indices.

%The example above only has discrete draws.
%Here is an example that learns the distribution of a continuous parameter.
%\begin{lstlisting}
%(define (geometric p)
%  (if (flip p)
%      0
%      (+ 1 (geometric p))))
%
%(hist (mh-query 5000 10
%          (define p (beta 1 1))
%          (define n (geometric p))
%          p
%          (< n 2
%            )))
%\end{lstlisting}

%\subsection{Example 2: Discrete Conditioning}
%
%Two coins:
%\begin{lstlisting}
%(query (define c1 (flip)) (define c2 (flip) (pair c1 c2) (condition (or c1 c2))))
%\end{lstlisting}

In addition to expressions presented here, \Church\ also supports 
\emph{stochastic memoization} \cite{DBLP:conf/uai/GoodmanMRBT08} by means of
a~$\tw{mem}$ function, which, applied to any given function,
produces a version of it that always returns the same value
when applied to the same arguments. This feature allows for
functions of integers to be treated as infinite lazy lists of random 
values, and is useful in defining some nonparametric models, such as the
Dirichlet Process. %It would be straightforward to 

It would be straightforward to add support for memoization in our encoding
by changing the translation to state-passing style,
but we omit this standard extension for the sake of brevity.
%%%%%%%%%%%%%%%%%%%%%%%%%%%%%%%%%%%%%%%%%%%%
\subsection{Example: Geometric Distribution} \label{subsec:geom}
%%%%%%%%%%%%%%%%%%%%%%%%%%%%%%%%%%%%%%%%%%%%
To illustrate the sampling-based semantics, recall the geometric
distribution example from Section~\ref{sec:intro}. 
It translates to
the following program in the core calculus:
\begin{center}
\begin{minipage}{.4\textwidth}
\begin{flushleft}
\begin{prog}
\llet{\textit{flip}}{\lambda x. (\textsf{rnd}() < x)}{}\\
\llet{\textit{geometric}}
{\\ \quad (\fixop\ g. \\ \qquad \quad \ \lambda p.\ 
(\llet{y}{\textsf{rnd}() < p}{\\ \qquad \qquad \ite{y}{0}{1 + (g\ p)}}) )}{} \\
\llet{n}{\fixop\ n'.\textit{geometric}\ 0.5}{}\\
\llet{b}{n > 1}{\\ \ite{b}{n}{\fail}}
\end{prog}
\end{flushleft}
\end{minipage}
\end{center}
Suppose we want to evaluate this program on the random trace $s =
[0.7, 0.8, 0.3]$. By \tr{Eval Appl}, we can substitute the definitions
of~$\textit{flip}$ and~$\textit{geometric}$ in the remainder of the
program, without consuming any elements of the trace nor changing the
weight of the sample.  Then we need to
evaluate~$\textit{geometric}\ 0.5$.

It can be shown (by repeatedly applying \tr{Eval Appl}) that for
any lambda-abstraction~$\lambda x . M$,
$\termone\{(\fixop\ x.\termone)/x\}\ \valone \Downarrow^s_w G $ 
if and only if $(\fixop\ x.\termone)\ \valone \Downarrow^s_w G$, which allows us
to unfold the recursion. Applying the unfolded definition of $\textit{geometric}$
to the argument $0.5$ yields an expression of the form
\begin{center}
\begin{minipage}{.4\textwidth}
\begin{flushleft}
\begin{prog}
\llet{y}{\textsf{rnd}() < 0.5}{\\ \quad \ite{y}{0}{1 + (\dots)}}.
\end{prog}
\end{flushleft}
\end{minipage}
\end{center}
For the first random draw, we have~\textsf{rnd}$()
\Downarrow^{[0.7]}_1 0.7$ by \tr{Eval Random} (because the density
of~\textsf{rnd} is~$1$ on the interval~$[0,1]$) and so \tr{Eval Prim}
gives~\textsf{rnd}$() < 0.5 \Downarrow^{[0.7]}_1~\tfalse$. After
unfolding the recursion two more times, evaluating the subsequent
``flips'' yields \textsf{rnd}$() < 0.5 \Downarrow^{[0.8]}_1~\tfalse$
and \textsf{rnd}$() < 0.5 \Downarrow^{[0.3]}_1~\ttrue$. By \tr{Eval If
  True}, the last if-statement evaluates to~$0$, terminating the
recursion. Combining the results by \tr{Eval Appl}, \tr{Eval If False}
and \tr{Eval Prim}, we arrive at $\textit{geometric}\ 0.5
\Downarrow^{[0.7, 0.8, 0.3]}_1~2$.

At this point, it is straightforward to see that the condition
in the if-statement on the final line is satisfied, and hence the program 
reduces with the given trace to the value~$2$ with weight~$1$.

This program actually yields weight~$1$ for every trace that returns an integer value.
This may seem counter-intuitive, because clearly not all outcomes have the same probability.
However, the probability of a given outcome is given by 
  an integral over the space of traces, 
  as described in Section~\ref{sec:from-trace-semantics}.
%%%%%%%%%%%%%%%%%%%%%%%%%%%%%%%
\subsection{Soft Constraints and \texttt{score}}\label{sec:soft-conditioning}
%%%%%%%%%%%%%%%%%%%%%%%%%%%%%%%
The geometric distribution example in Section \ref{subsec:geom} uses
a \emph{hard constraint}: program execution
fails and the value of $n$ is discarded whenever the Boolean
predicate $n > 1$ is not satisfied.
%}
In many machine learning applications we want to use a different kind
of constraint that models noisy data.  For instance, if
$c$ is the known output of a sensor that shows an approximate value
of some unknown quantity $x$ that is computed by the program, 
we want to assign higher probabilities to values of $x$ that are closer to $c$. 
This is sometimes known as a \emph{soft constraint}.

One naive way to implement a soft constraint is to use a hard constraint with a success probability based on $\lvert x - c\rvert$, for instance,
\[
\textit{condition}\ x\ c\ M := \ite{\textit{flip}(\textit{exp}({-(x-c)^2}))} M {\fail}.
\]
Then $\textit{condition}\ x\ c\ M$ has the effect of continuing as
$M$ with probability $\textit{exp}({-(x-c)^2})$, and
otherwise terminating execution.
In the context of a sampling-based semantics, it has the effect of 
adding a uniform sample from $[0,\textit{exp}({-(x-c)^2}))$ to any successful trace,
in addition to introducing more failing traces.

Instead, our calculus includes a primitive \texttt{score}, that avoids both 
adding dummy samples and introducing more failing traces.
It also admits the possibility of using efficient gradient-based methods of inference~(e.g.,~\citet{Homan2014NUTS}).   
Using \texttt{score}, the above conditioning operator can be redefined as
\[\textit{score-condition}\ x\ c\ M := \score{\textit{exp}({-(x-c)^2})};M\]

%
% A common choice of function to be used in the argument to \texttt{score} above
% is the density function of the Gaussian distribution, 
% as in $\score{\pdf{\textsf{Gaussian}}(m,v,c)}$.

% In the context of a probabilistic programming language, 
% this amounts to multiplying the probability of the current program trace by the value of 
% \ADG{where to put this?}
% There are many ways to do conditioning:
% \begin{enumerate}
% \item by rejection sampling - looping back to the start of the term
% \item hard conditioning using fail
% \item soft conditioning simulated using fail (and a nuisance parameter for the coin toss)
% \item soft conditioning using score as primitive (needs no nuisance parameter)
% \end{enumerate}
%%%%%%%%%%%%%%%%%%%%%%%%%%%%%%%%%%%%%%%%
\subsection{Example: Linear Regression}
\label{sec:exampl-line-regr}
%%%%%%%%%%%%%%%%%%%%%%%%%%%%%%%%%%%%%%%%
For an example of soft constraints, consider the ubiquitous linear
regression model $y = m\cdot x + b + \mathit{noise}$, where $x$ is
often a known feature and $y$ an observable outcome variable.  We can
model the noise as drawn from a Gaussian distribution with mean 0 and
variance 1/2 by letting the success probability be given by the
function \lstinline{squash} below.

The following
query\footnote{Cf.~\url{http://forestdb.org/models/linear-regression.html}.}
predicts the $y$-coordinate for $x=4$, given observations of four
points: $(0,0)$, $(1,1)$, $(2,4)$, and $(3,6)$.  (We use the
abbreviation \lstinline{(define ($f$ $x_1$ $\dots$ $x_n$) $e$)} for
\lstinline{(define $f$ (lambda ($x_1$ $\dots$ $x_n$) $e$)}, and use
\lstinline{and} for multiadic conjunction.)

\begin{lstlisting}
(query
   (define (sqr x) (* x x)))
   (define (squash x y) (exp(- (sqr(- x y)))))
   (define (flip p) (< (rnd) p))
   (define (softeq x y) (flip (squash x y))))

   (define m (gaussian 0 2))
   (define b (gaussian 0 2))
   (define (f x) (+ (* m x) b))
   
   (f 4)  ;; predict y for x=4

   (and (softeq (f 0) 0) (softeq (f 1) 1)
        (softeq (f 2) 4) (softeq (f 3) 6))
\end{lstlisting}
The model described above puts independent Gaussian priors on $m$ and $b$.
The condition of the query states that all observed $y$s are (soft) equal to $k\cdot x+m$.
Assuming that \lstinline{softeq} is used only to define constraints (i.e., positively),
we can avoid the nuisance parameter that arises from each \lstinline{flip} by
redefining \lstinline{softeq} as follows (given a \lstinline{score} primitive in \Church, mapped to $\score{-}$ in our $\lambda$-calculus):
\begin{lstlisting}
   (define (softeq x y) (score (squash x y)))
\end{lstlisting}

\section{Sampling-Based Operational Semantics}\label{sect:samplingsemantics}
%%%%%%%%%%%%%%%%%%%%%%%%%%%%%%%%%%%%%%%%%%%%%%
\newcommand{\pectxone}{A}
\newcommand{\pectxtwo}{B}
\newcommand{\rdxone}{R}
\newcommand{\app}{\mathit{app}}
\newcommand{\deapp}{\mathit{deapp}}
\newcommand{\distapp}{\mathit{distapp}}
\newcommand{\ectxs}[1]{\mathcal{C}}
\newcommand{\letbe}[3]{\mathtt{let}\;#2=#1\;\mathtt{in}\;#3}

\newcommand{\skset}{\mathsf{SK}}
\newcommand{\tmsofsk}[1]{\mathsf{TM}(#1)}

In this section, we further investigate sampling-based semantics for our calculus.
First, we introduce \emph{small}-step sampling-based
semantics and prove it equivalent to its big-step sibling as introduced
in Section \ref{sect:bssbs}. 
Then, we associate to any closed term $M$ two sub-probability distributions: 
one on the set of random traces, and the other on the set of return values. 
This requires some measure theory, recalled in Section \ref{sect:smtp}.
%%%%%%%%%%%%%%%%%%%%%%%%%%%%%%%%%%%%%%%%%%%%%%%
\subsection{Small-step Sampling-based Semantics}
%%%%%%%%%%%%%%%%%%%%%%%%%%%%%%%%%%%%%%%%%%%%%%%
We define small-step call-by-value evaluation. \emph{Evaluation
  contexts} are defined as follows:
\[
\ectxone\bnf\hole\midd\ectxone\termone\midd(\lambda\varone.\termone)\ectxone 
\]
We let $\ectxs{}$ be the set of all closed evaluation contexts, i.e.,
where every occurrence of a variable $x$ is as a subterm of $\lambda
x.M$. The term obtained by replacing the only occurrence of $\hole$ in
$\ectxone$ by $\termone$ is indicated as $\ct{\ectxone}{\termone}$.
\emph{Redexes} are generated by the following grammar:
\begin{align*}
  \rdxone\;\bnf\;&(\abstr{\varone}{\termone}) \valone%\midd c \termone
  \midd\distone(\vec{c})\midd g({\vec{c}})\midd \mathtt{score}(c)\\
  &\midd\fail\midd\ite{\ttrue}{\termone}{\termtwo}\\
  &\midd\ite{\tfalse}{\termone}{\termtwo}\midd\ertone
\end{align*}
\emph{Reducible terms} are those closed terms $\termone$ 
that can be written as $\ct{\ectxone}{\rdxone}$.
\begin{lemma}\label{lemma:unique}
  For every closed term $\termone$, either $\termone$ is a generalized
  value or there are unique $\ectxone,\rdxone$ such that
  $\termone=\ct{\ectxone}{\rdxone}$. 
  Moreover, if $\termone$ is not a generalized value and $\rdxone=\fail$,
  then $\ectxone$ is proper, that is, $\ectxone \neq \hole$.
\end{lemma}
\begin{proof}
  This is an easy induction on the structure of $\termone$.
  \qed
\end{proof}

\emph{Deterministic reduction} is the relation $\detred$ on closed terms
defined in Figure~\ref{fig:detred}.
\begin{figure}
\begin{center}
\fbox{
\begin{minipage}{.442\textwidth}
\begin{align*}
E[g(\vec{c})] &\detred E[\sigma_g(\vec{c})]\\
E[(\lambda x .  M)\ V] &\detred E[M\{V/x\}]\\
%E[c\ V] &\detred E[\fail]\\
E[\ite{1}{M_2}{M_3}] &\detred E[M_2]\\
E[\ite{0}{M_2}{M_3}] &\detred E[M_3]\\
E[T] &\detred E[\fail]\\
E[\fail] &\detred \mathtt \fail \quad \text{if $E$ is not $[\cdot]$}
\end{align*}
\end{minipage}}
\condnocr
\caption{Deterministic Reduction.}\label{fig:detred}
\end{center}
\end{figure}
\begin{lemma} \label{lemma:det-reduction-deterministic}
If $M \detred M'$ and $M \detred M''$ then $M' = M''$.
\end{lemma}
\begin{proof} 
    Since $M \detred M'$ implies that $M$ is not a generalized value,
    Lemma \ref{lemma:unique} states that $M = E[R]$ for some unique
    $E$, $R$. If $R = \fail$, then $E$ is proper and $E[R]$ can only
    reduce to $\fail$. Otherwise, it follows immediately by
    inspection of the reduction rules that $E[R] \detred E[N]$ for
    some $N$ that is uniquely determined by the redex $R$.
\qed
\end{proof}
Let us define composition of contexts $E \circ E'$ inductively as:
\begin{eqnarray*}
 [\ ] \circ E' &\deq& E'\\
 (E\ M) \circ E' &\deq& (E \circ E')\ M\\
 ((\lambda x . M)\ E) \circ E' &\deq& (\lambda x . M)\ (E \circ E')
\end{eqnarray*}

\begin{lemma} \label{lemma:context-composition-wdef}
$(E \circ E')[M] = E[E'[M]]$.
\end{lemma}
\begin{proof}
By induction on the structure of $E$.
\qed
\end{proof}
%\begin{lemma} \label{lemma:change-context-det}
%For any $E$ and $M$ such that $M \neq E'[\alpha]$, 
%\end{lemma}
\begin{lemma}\label{lemma:remove-context-det}
If $E[R] \detred E[N]$, then $R \detred N$.
\end{lemma}
\begin{proof}
By case analysis on the deterministic reduction rules.
\end{proof}
\begin{lemma} \label{lemma:change-context-det}
For any $E$ and $M$ such that $M \neq E'[\fail]$, if $M \detred M'$ then $E[M] \detred E[M']$.
\end{lemma}
\begin{proof}\mbox{ } Standard, using Lemmas \ref{lemma:unique} and \ref{lemma:det-reduction-deterministic}.%, and \ref{lemma:context-composition-wdef}.
%\begin{LONG}
%\begin{varitemize}
%\item[$\Rightarrow$:]

Since $M \detred M'$, $M$ is not a generalized value. 
By Lemma \ref{lemma:unique}, $M = E'[R]$ for some $E'$, $R$.

By assumption, $R \neq \fail$, so by inspection of the reduction rules
$E'[R] \detred E'[N]$ for some $N$. By Lemma 
 \ref{lemma:det-reduction-deterministic}, $E'[N] = M'$.
 By Lemma \ref{lemma:context-composition-wdef},
$E[M] = (E \circ E')[R]$ and $E[M'] = (E \circ E')[N]$.

Since Lemma \ref{lemma:remove-context-det} gives $R \detred N$,
by case analysis on the derivation of $R \detred N$ we can show 
that $(E \circ E')[R] \detred (E \circ E')[N]$, which implies
$E[M] \detred E[M']$.

%Since the context in the reduction rules is arbitrary, we can replace
%$E'$ with $(E \circ E')$ in $E'[R] \detred E'[N]$, obtaining 
%$(E \circ E')[R] \detred (E \circ E')[N]$, which implies
%$E[M] \detred E[M']$.
 %\item[$\Leftarrow$:]
%\end{varitemize}
%\end{LONG}
\qed
\end{proof}
\begin{figure*}
\begin{center}
\fbox{
\begin{minipage}{.97\textwidth}%{.64\textwidth}
\vspace{5pt}
\[
\AxiomC{$M \detred N$}
\UnaryInfC{$(M, w, s) \rightarrow (N, w, s)$}
\DisplayProof{\;\tr{Red Pure}}
\quad
  \AxiomC{$c \in (0,1]$}
  \UnaryInfC{$(E[\mathtt{score}(c)], w, s) \rightarrow (E[\ttrue], c \cdot w , s)$}
\DisplayProof{\;\tr{Red Score}}
\]
\[
\AxiomC{$w' = \pdf{\distone }(\vec{c}, c)$}
\AxiomC{$w' > 0$}
\BinaryInfC{$(E[\distone (\vec{c})], w,  c \mathrel{::} s) \rightarrow (E[c], w \cdot w', s)$}
\DisplayProof{\;\tr{Red Random}}
\]
\[
\AxiomC{$\pdf{\distone }(\vec{c}, c)=0$}
\UnaryInfC{$(E[\distone (\vec{c})], w,  c \mathrel{::} s) \rightarrow (E[\fail], 0, s)$}
\DisplayProof{\;\tr{Red Random Fail}}
\]
\vspace{5pt}
\end{minipage}}
\condnocr
\caption{Small-step sampling-based operational semantics}
\label{fig:small-step-sampling}
\end{center}
\end{figure*}
\begin{lemma} \label{lemma:small-step-determined}
If $(M,w,s) \rightarrow (M',w',s')$ and $(M,w,s) \rightarrow (M'',w'',s'')$,
then $M'=M''$, $w'=w''$ and $s''=s'$.
\end{lemma}
\begin{proof}
By case analysis. 
Since there is no rule that reduces generalized values,
$(M,w,s) \rightarrow (M',w',s')$ implies that $M \notin \gvalset$, so
by Lemma~\ref{lemma:unique}, $M = E[R]$ for some unique $E$, $R$.
\begin{varitemize}
   \item If $(M,w,s) \rightarrow (M',w',s')$ was derived with \tr{Red Pure},
    then $M=E[R]$, where $R \neq \distone(\underline{c})$ and $R \neq \mathtt{score}(c)$,
    which implies that $(M,w,s) \rightarrow (M'',w'',s'')$ must also have been derived
    with \tr{Red Pure}. Hence, we have $w'' = w' = w$, $s'' = s' = s$, $M \detred M'$
    and $M \detred M''$. By Lemma \ref{lemma:det-reduction-deterministic}, $M'' = M'$, as required.
   \item If $(M,w,s) \rightarrow (M',w',s')$ was derived with \tr{Red Random}, then
   $M = E[\distone(\vec{c})]$, $s = c \mathrel{::} s^*$ and $\pdf{\distone}(\vec{c},c) > 0$.
   Hence, $(M,w,s) \rightarrow (M'',w'',s'')$ must also have been derived with \tr{Red Random},
   and so $M''=M'=E[c]$, $s''=s'=s^*$ and $w''=w'=w\pdf{\distone}(\vec{c},c)$, as required.
   The \tr{Red Random Fail} case is analogous.
   \item If $(M,w,s) \rightarrow (M',w',s')$ was derived with \tr{Red Score},
   then $M=E[\mathtt{score}(c)]$ and $c \in (0,1]$, so $(M,w,s) \rightarrow (M'',w'',s'')$
   must also have been derived with \tr{Red Score}. Hence $M''=M'=E[\ttrue]$,
   $w''=w'=c\cdot w$ and $s''=s'=s$. 
   %The \tr{Red Score Fail} case is analogous.
\end{varitemize}
\qed
\end{proof}
%\begin{proof} By inversion, using Lemmas~\ref{lemma:unique} and \ref{lemma:det-reduction-deterministic}.
%  \begin{LONG}
%    If $M \neq E[\distone(\vec{c})]$, then both
%    $(M,w,s) \rightarrow (M',w',s')$ and
%    $(M,w,s) \rightarrow (M'',w'',s')$ must have been derived with
%    \tr{Red Pure}, so $w''=w'=w$, $M \detred M'$ and $M \detred M''$.
%    By Lemma \ref{lemma:det-reduction-deterministic}, $M' = M''$.
%
%    If $M = E[\distone(\vec{c})]$, then
%    $(M,w,s) \rightarrow (M',w',s')$ and
%    $(M,w,s) \rightarrow (M'',w'',s')$ must have been derived with
%    \tr{Red Random}. Hence, $s' = s @ [c]$ for some $c$, so
%    $w' = w'' = w \cdot\pdf{\distone }(\vec{c}, c)$ and $M' = M'' = E[c]$.
%    % \MS{Update this} Reduction rules are deterministic and at most
%    % one rule applies to every expression $M$.
%  \end{LONG}
%\end{proof}
Rules of small-step reduction are given in Figure \ref{fig:small-step-sampling}.
We let \emph{multi-step reduction} be the inductively defined relation 
$(M,w,s) \Rightarrow (M',w',s')$ if and only if $(M,w,s) = (M',w',s')$ or
$(M,w,s) \to (M'',w'',s'') \Rightarrow (M',w',s')$ for some $M'',w'',s''$.
%\begin{lemma} \label{lemma:left-trace-prefix}
%If $(M,w,s) \Rightarrow (M',w',s')$, then $s = s'' @ s'$ for some $s''$.
%\end{lemma}
%\begin{proof}
%By induction on the derivation of $(M,w,s) \Rightarrow (M',w',s')$.
%\end{proof}
As can be easily verified, the multi-step reduction of a term to a generalized value is
deterministic once the underlying trace is kept fixed:
\begin{lemma}\label{lemma:small-steps-determined}
If both $(M,w,s) \Rightarrow (G',w',s')$ and $(M,w,s) \Rightarrow (G'',w'',s'')$,
then $G'=G''$, $w'=w''$ and $s'=s''$.
\end{lemma}
\begin{proof}
By induction on the derivation of $(M,w,s) \Rightarrow (G',w',s')$,
with appeal to Lemma \ref{lemma:small-step-determined}.
\begin{varitemize}
\item Base case: $(M,w,s) = (G',w',s')$. Generalized values do not reduce, so $G''=G'=G$, $w''=w'=w$
and $s''=s'=s$.
\item Induction step:  $(M,w,s) \rightarrow (\hat{M}, \hat{w}, \hat{s}) \Rightarrow (G',w',s')$.
Since $M \neq G''$, we also have $(M,w,s) \rightarrow (M^*, w^*, s^*) \Rightarrow (G'',w'',s'')$.

By Lemma \ref{lemma:small-step-determined}, $(M^*, w^*, s^*) = (\hat{M}, \hat{w}, \hat{s})$,
and so by induction hypothesis, $(G'',w'',s'') = (G',w',s')$, as required.

%By inspection of the reduction rules, either $ \hat{s} = s^* = s$ or
%$\hat{s} = s @ [c]$ and $s^* = s @ [d]$ for some $c$, $d$. In the latter case,
%by Lemma \ref{lemma:left-trace-prefix}, $s' = s @ [c] @ s'' = s @ [d] @ s'''$
%for some $s''$, $s'''$ ,which implies $c=d$. Thus, $\hat{s} = s^*$.
%
%By Lemma \ref{lemma:small-step-determined}, $\hat{M} = M^*$ and $\hat{w} = w^*$.
%The induction hypothesis implies $G'' = G'$ and $w'' = w'$, as required.
\end{varitemize}
\qed
\end{proof}
\begin{lemma} \label{lemma:change-context}
For any $E$ and $M$ such that $M \neq E'[\fail]$,
if
$(M, w, s) \rightarrow (M', w', s')$ then
$(E[M], w, s) \rightarrow (E[M'], w', s')$
\end{lemma}
\begin{proof} By inversion of $\rightarrow$, using Lemma~\ref{lemma:change-context-det}.
  %\begin{LONG}
  \begin{varitemize}
  \item If $(M, w, s) \rightarrow (M', w', s')$ was derived with
  \tr{Red Pure}, then $M \detred M'$, so by Lemma \ref{lemma:change-context-det},
  $E[M] \detred E[M']$, and by \tr{Red Pure}, $(E[M],w,s) \red (E[M'],w',s')$.

%  $M = E[R]$ for some deterministic $R$, and
%  $R \neq \fail$ by assumption. Hence we have $E[R] \detred E[N']$,
%  where $E[N'] = N$ by Lemma \ref{lemma:det-reduction-deterministic}
  \item If $(M, w, s) \rightarrow (M', w', s')$ was derived with
  \tr{Red Random}, then $M = E'[\distone(\vec{c})]$, $M' = E'[c]$,
  $s = c \mathrel{::} s'$ and $w' = w \pdf{\distone}(\vec{c},c)$,
  where $\pdf{\distone}(\vec{c},c) > 0$. By \tr{Red Random} and Lemma
  \ref{lemma:context-composition-wdef}, we can derive
  $(E[M], w, s) \rightarrow (E[M'], w', s')$. Cases \tr{Red Random Fail}
  and \tr{Red Score}
   %and \tr{Red Score Fail} 
  are anaologous.
  \end{varitemize}
%    If $M\neq E[\distone (\vec{c})]$, then
%    $(E[M], w, s) \rightarrow (E[M'], w', s')$ was derived with
%    \tr{Red Pure}, so the result follows by
%    Lemma~\ref{lemma:change-context-det}.
%
%    If $M = E[\distone (\vec{c})]$, then
%    $(E[M], w, s) \rightarrow (E[M'], w', s')$ was derived with
%    \tr{Red Random}, so $w' = w w''$ and $s' = s @ [c]$, where
%    $w'' = \pdf{\distone}(\vec{c}, c)$.  From the assumptions
%    $w'' = \pdf{\distone}(\vec{c}, c)$ and $w'' > 0$, we can derive
%    $(E[M], w, s) \rightarrow (E[M'], w w'', s@[c])$ by \tr{Red
%      Random} for any $E$.
%  \end{LONG}
\qed
\end{proof}
\begin{lemma} \label{lemma:remove-context}
If $(E[R], w, s) \rightarrow (E[N], w', s')$ then $(R, w, s) \rightarrow (N, w', s')$
\end{lemma}
\begin{proof} By case analysis.

  \begin{varitemize}
  \item If $(E[R],w,s) \red (E[N],w',s')$ was derived with \tr{Red Pure}, then $E[R] \detred E[N]$,
  so by Lemma \ref{lemma:change-context-det}, $R \detred N$, which implies
  $(M, w, s) \rightarrow (M', w', s')$.
  \item If $(E[R], w, s) \rightarrow (E[N], w', s')$ was derived with \tr{Red Random},
  then $R = \distone(\vec{c})$, $N= c$, $s = c \mathrel{::} s'$ and $w' = w \pdf{\distone}(\vec{c},c)$,
  where $\pdf{\distone}(\vec{c},c) > 0$.%We have $M = E^*[D(\vec{c})]$ and $M' = E^*[c]$ for some $E^*$ such
%  that $E \circ E^* = E'$. \MS{Make it more precise?}
  
  Hence, with \tr{Red Random}, we can derive
  $(\distone(\vec{c}), w, s) \rightarrow (c, w', s')$
  %$(E^*[\distone(\vec{c})], w, s) \rightarrow (E^*[c], w', s')$, which implies
  %$(M, w, s) \rightarrow (M', w', s')$.
  
  Cases \tr{Red Random Fail}
  and \tr{Red Score}
   %and \tr{Red Score Fail} 
  are anaologous.
  
\end{varitemize}
\end{proof}
Reduction can take place in any evaluation context, provided the result is not
a failure. Moreover, multi-step reduction is a transitive relation. This is
captured by the following lemmas.
\begin{lemma} \label{lemma:congr-closure}
For any $E$, if
$(M,w,s) \Rightarrow (M',w',s')$ and $M' \neq \fail$, then we have $(E[M],w,s) \Rightarrow (E[M'],w',s')$.
\end{lemma}
\begin{proof}
%By induction on the number of steps in the derivation of 
%$(M,w,s) \Rightarrow (M',w',s')$.
By induction on the number of steps in the derivation of 
$(M,w,s) \Rightarrow (M',w',s')$, with appeal to Lemma \ref{lemma:change-context}.

Since $M' \neq \fail$, no expression in the 
derivation chain (other than the last one) can be of the form $E'[\fail]$.
\qed
\end{proof}
\begin{lemma} \label{lemma:congr-closure-exc}
For any $E$, if
$(M,w,s) \Rightarrow (\fail,w',s')$ then \\$(E[M],w,s) \Rightarrow (\fail,w',s')$.
\end{lemma}
\begin{proof} By induction on the number of steps in the derivation, 
  using Lemmas~\ref{lemma:change-context} and \ref{lemma:congr-closure}.
  If $E = [\ ]$, the result holds trivially, so let us assume $E \neq [\ ]$.
%  \begin{LONG}
    If $(M,w,s)
    \Rightarrow (\fail,w',s')$ was derived in $0$ steps, then $M =
    \fail$, $w' = w'$ and $s' =
    s$, so by \tr{Red Pure}, $(E[\fail],w,s) \red
    (\fail,w,s)$, as required.

    If $(M,w,s)
    \Rightarrow (\fail,w',s')$ was derived in
    $1$
    or more steps, then:
    \begin{varitemize}
    \item If $M = E'[\fail]$ and $E' \neq [\ ]$, then 
    $((E \circ E')[\fail],w,s) \red (\fail,w',s')$ by \tr{Red Pure}.
    \item
      Otherwise, there exist $\hat{M}$,
    $\hat{w}$,
    $\hat{s}$
    such that $(M,w,s)
    \red (\hat{M}, \hat{w}, \hat{s}) \Rightarrow
    (\fail,w',s')$, where $M \notin \gvalset$.
    By induction hypothesis, $(E[\hat{M}], \hat{w}, \hat{s}) \Rightarrow
    (\fail,w',s')$ for any $E$, and by Lemma \ref{lemma:change-context},
    $(E[M], w, s) \red (E[\hat{M}], \hat{w}, \hat{s})$.
    \end{varitemize}
%    where $\hat{M} \notin \gvalset$. Because $\alpha
%    \neq E'[c]$ for any $E'$, $c$, $(\hat{M}, \hat{w}, \hat{s})
%    \rightarrow
%    (\alpha,w',s')$ must have been derived with \tr{Red Pure}, which
%    implies $\hat{w} = w'$ and $\hat{s} = s'$.
%
%    By Lemma \ref{lemma:congr-closure},
%    $(E[M],w,s) \Rightarrow (E[\hat{M}], w', s')$.
%
%    If $\hat{M} = E'[\beta]$ for some $E'$, $\beta$, then
%    $\beta = \alpha$ and by \tr{Red Pure},
%    $((E \circ E')[\alpha], w', s') \rightarrow (\alpha, w', s')$.
%    Thus, $(E[M],w,s) \Rightarrow (\alpha,w',s')$, as required.
%
%    If $\hat{M} \neq E'[\beta]$, then by Lemma
%    \ref{lemma:change-context},
%    $(E[\hat{M}], w', s') \rightarrow (E[\alpha], w', s')$.  By
%    \tr{Red Pure},
%    $(E[\alpha], w', s') \rightarrow (\alpha, w', s')$. Thus,
%    $(E[M],w,s) \Rightarrow (\alpha,w',s')$.
%  \end{LONG}
  \qed
\end{proof}

\begin{lemma} \label{lemma:w-greater-zero}
If $(M, w, s) \Rightarrow (M',w',s')$ and $w \geq 0$, then $w' \geq 0$.
\end{lemma}
\begin{proof}
By induction on the number of steps in the derivation.% of $(M, w, s) \Rightarrow (M',w',s')$:
\begin{varitemize}
\item If $(M, w, s) \Rightarrow (M',w',s')$ was derived in $0$ steps, then $w' = w$, so $w' \geq 0$.
\item If $(M, w, s) \Rightarrow (M',w',s')$ was derived in $1$ or more steps, then
$(M, w, s) \red (M^*, w^*, s^*) \Rightarrow (M',w',s')$. 

If $(M, w, s) \rightarrow (M^*,w^*,s^*)$ was derived with \tr{Red Pure}, 
then $w^* = w \geq 0$.

If $(M, w, s) \rightarrow (M^*,w^*,s^*)$ was derived with \tr{Red Random},
then $w^* = w\cdot w''$ for some $w''>0$, so $w^* \geq 0$.

If $(M, w, s) \rightarrow (M^*,w^*,s^*)$ was derived with \tr{Red Score},
then $w^* = w\cdot c$ for some $c>0$, so $w' \geq 0$.

If $(M, w, s) \rightarrow (M^*,w^*,s^*)$ was derived with \tr{Red Random Fail},
then $w^* = 0$.

In either case, $w^* \geq 0$, so  by induction hypothesis, $w' \geq0$.
\end{varitemize}
\qed
\end{proof}
\begin{lemma} \label{lemma:w-greater-zero-strict}
If $(M, w, s) \red (M',w',s')$ was not derived with \tr{Red Random Fail} and $w > 0$, then $w' > 0$.
\end{lemma}
\begin{proof}
By inspection (similar to the inductive step in the proof of Lemma \ref{lemma:w-greater-zero}).
\end{proof}
\begin{lemma} \label{lemma:multiply}
If $(M,w,s) \rightarrow (M',w',s')$, then for any $w^* \geq 0$,
$(M,w w^*,s) \rightarrow (M',w' w^*,s')$
\end{lemma}
\begin{proof}
By case analysis.
\qed
\end{proof}
\begin{lemma} \label{lemma:add-suffix}
If $(M,w,s) \rightarrow (M',w',s')$, then for any $s^*$,
$(M,w,s @ s^*) \rightarrow (M',w',s' @ s^*)$
\end{lemma}
\begin{proof}
By case analysis.
\qed
\end{proof}

\begin{lemma} \label{lemma:remove-suffix}
If $(M,w,s) \rightarrow (M',w',s')$, then
there is $s^*$ such that $s = s^* @ s'$
and $(M,w,s^*) \rightarrow (M',w',[])$
\end{lemma}
\begin{proof}
By case analysis.
\qed
\end{proof}

\begin{lemma} \label{lemma:multiply-closure}
If $(M,w,s) \rightarrow^k (M',w',s')$, then for any $w^* \geq 0$,
$(M,w w^*,s) \rightarrow^k (M',w' w^*,s')$
\end{lemma}
\begin{proof}
By induction on $k$, with appeal to Lemma \ref{lemma:multiply}.
\qed
\end{proof}

\begin{lemma} \label{lemma:add-suffix-closure}
If $(M,w,s) \rightarrow^k (M',w',s^*)$, then for any $s'$,
$(M,w,s @ s') \rightarrow^k (M',w',s^* @ s')$
\end{lemma}
\begin{proof}
By induction on $k$, with appeal to Lemma \ref{lemma:add-suffix}.
\end{proof}

\begin{lemma} \label{lemma:small-step-comp}
If both $(M,1,s) \Rightarrow (M',w',[])$ and $(M',1,s') \Rightarrow (M'',w'',[])$,
then $(M,1,s@s') \Rightarrow (M'', w'\cdot w'', [])$.
\end{lemma}
\begin{proof}
By Lemma \ref{lemma:add-suffix-closure}, $(M,1,s @ s') \Rightarrow (M',w',s')$
and by lemma \ref{lemma:w-greater-zero}, $w' \geq 0$.
Hence, by Lemma \ref{lemma:multiply-closure}, $(M',w',s') \Rightarrow (M'', w' w'', [])$,
which gives $(M,1,s@s') \Rightarrow (M'', w' w'', [])$.
%By induction on the number of steps in derivation of 
%$(M,1,s) \Rightarrow (M',w',[])$.
%\begin{varitemize}
%\item Base case: If $(M,1,s) \red^0 (M',w',[])$, then result holds trivially.
%\item Induction step: assume $(M,1,s) \red (M^*,w^*,s^*) \red^k (M',w',[])$.
%By Lemma \ref{lemma:add-suffix-closure}, we have $(M,1,s@s') \red (M^*,w^*,s^* @ s') \red^k (M',w',s')$.
%
%By Lemma \ref{lemma:w-greater-zero}, $w^* > 0$, so by Lemma \ref{lemma:multiply-closure}
%$(M^*,1,s^*) \red^k (M',w' / w^*,[])$. By induction hypothesis, 
%$(M^*,1,s^* @ s' ) \Rightarrow (M'', w' w'' / w^*, [])$. Now,
%by Lemma \ref{lemma:add-suffix}, $(M,1,s@s') \red (M^*,w^*,s^*@s')$
%and by Lemma \ref{lemma:multiply-closure}, 
%$(M^*,w^* ,s^* @ s' ) \Rightarrow (M'', w' w'', [])$, which implies
%$(M,1,s@s') \Rightarrow (M'', w' w'', [])$.
%\end{varitemize}
\qed
\end{proof}

\begin{lemma} \label{lemma:eval-fail}
For any $E$, $E[\fail] \Downarrow^{[]}_1 \fail$.
\end{lemma}
\begin{proof}
By induction on the structure of $E$.

\begin{varitemize}
\item Base case: $E = [\ ]$, the result follows by \tr{Eval Val}.
\item Induction step:
\begin{varitemize}
\item Case $E = (\lambda x .L)\ E'$:
By induction hypothesis, $E'[\fail] \Downarrow^{[]}_1 \fail$,
and by \tr{Eval Appl Raise2}, $(\lambda x .L)\ E'[\fail] \Downarrow^{[]}_1 \fail$,
as required.

\item Case $E = E'\ L$:
By induction hypothesis, $E'[\fail] \Downarrow^{[]}_1 \fail$,
so by \tr{Eval Appl Raise1}, we get $E'[\fail]\ L \Downarrow^{[]}_1 \fail$.
\end{varitemize}
\end{varitemize}
\qed
\end{proof}
\begin{lemma}\label{lemma:eval-random-fail}
For any $E$, if $\pdf{D}(\vec{c},c) = 0$, then $E[D(\vec{c})] \Downarrow^{[c]}_0 \fail$.
\end{lemma}
\begin{proof}
By induction on the structure of $E$.

\begin{varitemize}
\item Base case: $E = [\ ]$,  the result follows by \tr{Eval Random Fail}.
\item Induction step:
\begin{varitemize}
\item Case $E = (\lambda x .L)\ E'$:
By induction hypothesis, $E'[D(\vec{c})] \Downarrow^{[c]}_0 \fail$,
and by \tr{Eval Appl Raise2}, $(\lambda x .L)\ E'[D(\vec{c})] \Downarrow^{[]}_0 \fail$,
as required.

\item Case $E = E'\ L$:
By induction hypothesis, $E'[D(\vec{c})] \Downarrow^{[c]}_0 \fail$,
so by \tr{Eval Appl Raise1}, we get $E'[D(\vec{c})]\ L \Downarrow^{[]}_0 \fail$.
\end{varitemize}
\end{varitemize}
\qed
\end{proof}
The following directly relates the small-step and big-step semantics, saying
that the latter is invariant on the former:
\begin{lemma} \label{lemma:split}
If $(M,1,s) \rightarrow (M', w, [])$ and $M' \Downarrow^{s'}_{w'}G$,
then $M \Downarrow^{s@s'}_{w\cdot w'}G$.
\end{lemma}
\begin{proof}
By induction on the structure of $M$.

If $M = E[\fail]$ for some $E \neq [\ ]$, the result follows immediately by Lemma 
\ref{lemma:eval-fail}. Now, let us assume that $M \neq E[\fail]$.
\begin{varitemize}
\item Base case: $M = R$:
\begin{varitemize}
\item If $M = g(\vec{c})$ or  $M = c\ V$
%or $M = \alpha\ V$ or $M = (\lambda x.N)\ \alpha$
or $M = T$, then  
$M$ reduces to a generalized value in 1 step, so the result holds trivially (by one of the evaluation rules).
\item Case $M = \ite{\ttrue}{M_2}{M_3}$: We have 
$(\ite{\ttrue}{M_2}{M_3}, 1, []) \red (M_2, 1, [])$.
By assumption, $M_2 \Downarrow^{s'}_{w'} G$. Thus, the desired result holds by \tr{Eval If True}.
\item Case $M = \ite{\tfalse}{M_2}{M_3}$: analogous to the previous case.
\item Case $M = (\lambda x . N_1)\ V$: We have 
$((\lambda x . N_1)\ V, 1, []) \red (N_1 \{V/x\}, 1, [])$.
Since $ (\lambda x . N_1)$ and $V$ are already values and $N_1 \{V/x\} \Downarrow^{s'}_{w'} G$
by assumption, \tr{Eval Appl} yields $ (\lambda x . N_1)\ V  \Downarrow^{s'}_{w'} G$.
\item Case $M = \distone(\vec{c})$: $(M,1,s) \rightarrow (M', w, [])$ must have been
derived with \tr{Red Random} or \tr{Red Random Fail}. In the former case,
$s = [c]$, $M' = c$, and $w = \pdf{D}(\vec{c},c)$, where $c > 0$.
The second assumption then takes the form $c \Downarrow^{[]}_1 c$, so the
required result follows from \tr{Eval Random}. The \tr{Red Random Fail} case is
similar, with the result following from \tr{Eval Random Fail}.
\item Case $M = \mathtt{score}(c)$, $c\in (0,1]$:
 $(M,1,s) \rightarrow (M', w, [])$ must have been
derived with \tr{Red Score}.% or \tr{Red Score Fail}.
so
$M' = \ttrue$,
$w = c$ and $s = []$. Thus, the result then follows from \tr{Eval Score}.
%The second case is similar- the result follows from \tr{Eval Score Fail}.

\end{varitemize}
\item Induction step:
$M = E[R]$, $E \neq [\ ]$, $R \neq \fail$:
\begin{varitemize}
\item Case $E = (\lambda x . L)\ E'$: $M = (\lambda x . L)\ E'[R]$.
%
%If $R = \fail$ and $E' = , then $(E[\fail], 1, []) \rightarrow (\fail, 1, [])$ and
%$\fail \Downarrow^{[]}_1 \fail$.

We have
$((\lambda x . L)\ E'[R], 1, s) \red ((\lambda x . L)\ E'[N], w, [])$ for some $N$,
so by lemmas \ref{lemma:change-context} and \ref{lemma:remove-context},
$(E'[R], 1, s) \red (E'[N], w, [])$. By assumption, $(\lambda x . L)\ E'[N] \Downarrow^{s'}_{w'} G$.

\begin{varitemize}
\item If $(\lambda x . L)\ E'[N] \Downarrow^{s'}_{w'} G$ was derived with \tr{Eval Appl}, then 
$E'[N] \Downarrow^{s_1}_{w_1} V$ and $(\lambda x . L)\ V \Downarrow^{s_2}_{w_2} G$,
where $ w' = w_1 w_2$ and $s' = s_1 @ s_2$. By induction hypothesis, 
$E'[R] \Downarrow^{s@s_1}_{w w_1} V$, so \tr{Eval Appl} gives $(\lambda x . L)\ E'[R] \Downarrow^{s@s'}_{w w'} G$,
as required.
\item If $(\lambda x . L)\ E'[N] \Downarrow^{s'}_{w'} G$ was derived with \tr{Eval Appl Raise3},
then $G = \fail$ and $E'[N] \Downarrow^{s'}_{w'} \fail$. By induction hypothesis,
$E'[R] \Downarrow^{s@s'}_{w w'} \fail$, so by \tr{Eval Appl Raise3}, 
$(\lambda x . L)\ E'[R] \Downarrow^{s@s'}_{w w'} \fail$
\end{varitemize}

\item Case $E = E'\ L$: $M = E'[M^*]\ L$:

We have
$(E'[R]\ L, 1, s) \red (E'[N]\ L, w, [])$ for some $N$,
so by lemmas \ref{lemma:change-context} and \ref{lemma:remove-context},
$(E'[R], 1, s) \red (E'[N], w, [])$. By assumption, $E'[N]\ L \Downarrow^{s'}_{w'} G$.

\begin{varitemize}
\item If $E'[N]\ L \Downarrow^{s'}_{w'} G$ was derived with \tr{Eval Appl}, then 
$E'[N] \Downarrow^{w_1}_{s_1} (\lambda x . N')$,
$L~\Downarrow^{w_2}_{s_2}~V$ and $N'[V/x] \Downarrow^{w_3}_{s_3} G$,
where $ w' = w_1 w_2 w_3$ and $s' = s_1 @ s_2 @ s_3$. 
By induction hypothesis, 
$E'[R] \Downarrow^{s@s_1}_{w w_1} (\lambda x . N')$, so \tr{Eval Appl} gives $E'[R]\ L \Downarrow^{s@s'}_{w w'} G$,
as required.

\item If $E'[N]\ L \Downarrow^{s'}_{w'} G$ was derived with \tr{Eval Appl Raise1},
then $G = \fail$ and $E'[N] \Downarrow^{s'}_{w'} \fail$. By induction hypothesis,
$E'[R] \Downarrow^{s @ s'}_{w w'} \fail$, so by  \tr{Eval Appl Raise1}, $E'[R]\ L \Downarrow^{s@s'}_{w w'} \fail$

\item If $E'[N]\ L \Downarrow^{s'}_{w'} G$ was derived with \tr{Eval Appl Raise3}, then 
$E'[N] \Downarrow^{s_1}_{w_1} (\lambda x . N')$ and
$L \Downarrow^{s_2}_{w_2} \fail$,
where $ w' = w_1 w_2$ and $s' = s_1 @ s_2$. 
By induction hypothesis, 
$E'[R] \Downarrow^{s@ s_1}_{w w_1} (\lambda x . N')$, so \tr{Eval Appl Raise3} 
gives $E'[R]\ L \Downarrow^{s@s'}_{w w'} \fail$, as required.

\item If $E'[N]\ L \Downarrow^{s'}_{w'} G$ was derived with \tr{Eval Appl Raise1},
then $G = \terror$ and $N_1' \Downarrow^{s'}_{w'} c$. By induction hypothesis,
$E'[R] \Downarrow^{s@s'}_{w w'} c$, so by  \tr{Eval Appl Raise1},
$E'[R]\ L \Downarrow^{s@s'}_{w w'} \terror$.
%\end{varitemize}
\end{varitemize}
\end{varitemize}

%\begin{\varitemize}

\end{varitemize}
\qed
\end{proof}

%\begin{proof} By inversion of $\rightarrow$, using Lemmas~\ref{lemma:confluence-det} and~\ref{lemma:confluence-rnd}.
%  \begin{LONG}
%
%    If $(M,1,[]) \rightarrow (M', w_1, s_1)$ was derived with \tr{Red
%      Pure}, then $M \detred M'$, $w_1 = 1$ and $s_1 = []$, so the
%    result follows directly from Lemma \ref{lemma:confluence-det}.
%
%    If $(M,1,[]) \rightarrow (M', w_1, s_1)$ was derived with \tr{Red
%      Random}, then $M = E[\distone(\vec{c})]$, so the result follows
%    by Lemma \ref{lemma:confluence-rnd}.
%  \end{LONG}
%\end{proof}

Finally, we have all the ingredients to show that the small-step and
the big-step sampling-based semantics both compute the same traces with the
same weights.
\begin{theorem} \label{thm:sample-small-big-eq}
$M \Downarrow^s_w G$ if and only if $(M,1,s) \Rightarrow (G, w, [])$.
\end{theorem}
\begin{proof}\mbox{ }
%\begin{varitemize}
%\item
The left to right implication is an induction on the derivation of $M
\Downarrow^s_w G$.
  The most interesting case is definitely the following:
  \begin{center}
    \restaterule{Eval Appl}
                {M \Downarrow^{s_1}_{w_1} \lambda x . M' \qquad
                  N \Downarrow^{s_2}_{w_2} V\qquad
                  M'[V/x] \Downarrow^{s_3}_{w_3} G}
                {M\ N \Downarrow^{s_1@s_2@s_3}_{w_1\cdot w_2\cdot w_3} G}
  \end{center}
  By induction hypothesis, $(M,1,s_1) \Rightarrow (\lambda x .  M',
  w_1, [])$, $(N,1,s_2) \Rightarrow (V, w_2, [])$ and
  $(M'[V/x],1,s_3) \Rightarrow (G,w_3, [])$.  By Lemma
  \ref{lemma:congr-closure} (for $E = [\ ]\ N$), $(M\ N,1,s_1)
  \Rightarrow ((\lambda x .  M')\ N, w_1, [])$.  By Lemma
  \ref{lemma:congr-closure} again (for $E = (\lambda x .
  M')\ [\ ]$), $((\lambda x .  M')\ N,1,s_2) \Rightarrow ((\lambda x
  .  M')\ V, w_2, [])$.  By Lemma \ref{lemma:small-step-comp},
  $(M\ N,1,s_1@s_2) \Rightarrow ((\lambda x .  M')\ V, w_1 w_2, [])$
  By \tr{Red Pure}, $((\lambda x .  M')\ V, w_1\cdot w_2, [])
  \rightarrow (M'[V/x],w_1\cdot w_2,[])$, which implies
  $(M\ N,1,s_1@s_2) \Rightarrow ((\lambda x .  M')\ V, w_1 w_2, [])$
  Thus, the desired result follows by Lemma \ref{lemma:small-step-comp}.
  \begin{varitemize}
  \item Case:
    \restaterule{Eval Val}
                {G \in \gvalset}
                {G \Downarrow^{[]}_1 G }
                
                Here, $M = V$, $w = 1$ and $s=[]$. so $(M,w_0,s_0)$ reduces to 
                $(V, w_0, s_0)$ in 0 steps by the small-step semantics.
  \item Case:
    \restaterule{Eval Random}
                            {w = \pdf{\distone }(\vec{c}, c) \\ w > 0}
                            {\distone (\vec{c}) \Downarrow^{[c]}_{w} c}
                            
                  By \tr{Red Random} (taking $E = [\ ]$), $(\distone (\vec{c}), 1, [c]) \rightarrow (c, w, [])$.
  \item Case:
      \restaterule{Eval Random Fail}
                  {\pdf{\distone }(\vec{c}, c) = 0}
                  {\distone (\vec{c}) \Downarrow^{[c]}_{0} \fail}
                  
                  By \tr{Red Random Fail} (taking $E = [\ ]$), $(\distone (\vec{c}), 1, [c]) \rightarrow (\fail, 0, [])$.
                  
    \item Case:
      \restaterule{Eval Prim}
                  {}
                  {g(\vec{c}) \Downarrow^{[]}_{1} \intfun{g}(\vec{c})}
                  
                  By \tr{Red Pure} (taking $E = [\ ]$), $(g(\vec{c}), 1, [])  \rightarrow (\intfun{g}(\vec{c}), 1, [])$.
                  
    \item Case:
      \restaterule{Eval Score}
                  {c \in (0,1]}
                    {\mathtt{score}(c) \Downarrow^{[]}_c \ttrue}
                    
                    By \tr{Red Score} (taking $E = [\ ]$), $(\distone (\vec{c}), 1, []) \rightarrow (c, w, [])$.
                    
                    %\item Case:
                    %\restaterule{Eval Score Fail}
                    %{c \notin (0,1]}
                    %{\mathtt{score}(c) \Downarrow^{[]}_c \fail}
                    %
                    %By \tr{Red Score Fail} (taking $E = [\ ]$), $(\distone (\vec{c}), 1, []) \rightarrow (\fail, w, [])$.
    \item Case:
      \restaterule{Eval Appl}
                  {M \Downarrow^{s_1}_{w_1} \lambda x . M' \\
                    N \Downarrow^{s_2}_{w_2} V\\
                    M'[V/x] \Downarrow^{s_3}_{w_3} G}
                  {M\ N \Downarrow^{s_1@s_2@s_3}_{w_1\cdot w_2\cdot w_3} G}
                  
                  By induction hypothesis, 
                  $(M,1,s_1) \Rightarrow (\lambda x .  M', w_1, [])$,
                  $(N,1,s_2) \Rightarrow (V, w_2, [])$
                  and $(M'[V/x],1,s_3) \Rightarrow (G,w_3, [])$.
                  
                  By Lemma \ref{lemma:congr-closure} (for $E = [\ ]\ N$), 
                  $(M\ N,1,s_1) \Rightarrow ((\lambda x .  M')\ N, w_1, [])$.
                  
                  By Lemma \ref{lemma:congr-closure} again (for $E = (\lambda x .  M')\ [\ ]$), 
                  $((\lambda x .  M')\ N,1,s_2) \Rightarrow ((\lambda x .  M')\ V, w_2, [])$.
                  
                  By Lemma \ref{lemma:small-step-comp},
                  $(M\ N,1,s_1@s_2) \Rightarrow ((\lambda x .  M')\ V, w_1 w_2, [])$
                  
                  By \tr{Red Pure}, $((\lambda x .  M')\ V, w_1\cdot w_2, []) \rightarrow (M'[V/x],w_1\cdot w_2,[])$,
                  which implies $(M\ N,1,s_1@s_2) \Rightarrow ((\lambda x .  M')\ V, w_1 w_2, [])$
                  
                  Thus, the desired result follows by Lemma \ref{lemma:small-step-comp}.
                  
    \item Case:
      \restaterule{Eval Appl Raise1}
                  {M \Downarrow^{s}_{w} \fail }
                  {M\ N \Downarrow^s_{w} \fail}
                  
                  By induction hypothesis, $(M, 1, s) \Rightarrow (\fail, w, [])$.
                  
                  %If $M = E[\alpha]$ for some $E$, then by Lemma \ref{lemma:eval-alpha} $w=1$ and $s=[]$.
                  %By \tr{Red Pure}, we get $(E[\alpha], 1, []) \rightarrow (\alpha, 1, [])$.
                  %
                  %If $M \neq E[\alpha]$, 
                  
                  By Lemma \ref{lemma:congr-closure-exc} 
                  (with $E = [\ ]\ N) $), $(M\ N, 1, s) \Rightarrow (\fail, w, [])$.
                  
                  %By \tr{Red Pure},
                  %$(\alpha\ N, w, s) \rightarrow (\alpha, w, s)$.
                  
                  %Thus, $(M\ N, 1, []) \Rightarrow (\alpha, w, s)$.
                  
    \item Case:
      \restaterule{Eval Appl Raise2}
                  {M \Downarrow^{s}_{w} c }
                  {M\ N \Downarrow^s_{w} \fail}
                  
                  By induction hypothesis, $(M, 1, s) \Rightarrow (c, w, [])$.
                  By Lemma \ref{lemma:congr-closure} (with $E = [\ ]\ N) $), $(M\ N, 1, s) \Rightarrow (c\ N, w, [])$.
                  
                  By \tr{Red Pure},
                  $(c\ N, w, []) \rightarrow (\fail, w, [])$.
                  
                  Thus, $(M\ N, 1, s) \Rightarrow (\fail, w, [])$.
                  
    \item Case:
      \restaterule{Eval Appl Raise3}
                  {M \Downarrow^{s_1}_{w_1} \lambda x . M' \\
                    N \Downarrow^{s_2}_{w_2} \fail}
                  {M\ N \Downarrow^{s_1@s_2}_{w_1\cdot w_2} \fail}
                  
                  By induction hypothesis, 
                  $(M,1,s_1) \Rightarrow (\lambda x .  M', w_1, [])$,
                  and $(N,1,s_2) \Rightarrow (\fail, w_2, [])$.
                  
                  By Lemma \ref{lemma:congr-closure},
                  $(M\ N,1,s_1) \Rightarrow ((\lambda x . M')\ N, w_1, [])$.
                  
                  By Lemma \ref{lemma:congr-closure-exc},
                  $((\lambda x.M')\ N,1, s_2) \Rightarrow (\fail, w_2, [])$.
                  
                  Thus, by Lemma \ref{lemma:small-step-comp},
                  $(M\ N,1 ,s_1 @ s_2) \Rightarrow (\fail, w_1\cdot w_2, [])$.

                  %
                  %If $N = E[\alpha]$ for some $E$, then by Lemma \ref{lemma:eval-alpha}, $w_2=1$ and $s_2=[]$.
                  %By \tr{Red Pure} and Lemma \ref{lemma:init-weight-trace-const},
                  % $((\lambda x . M')\ E[\alpha], w_1, s_1) \rightarrow (\alpha, w_1, s_1)$.
                  % Hence, $(M,1,[]) \Rightarrow (\alpha, w_1, s_1)$.
                  %
                  %If $N \neq E[\alpha]$, by Lemma \ref{lemma:congr-closure} (for $E = [\ ]\ N$), 
                  %$(M\ N,1,[]) \Rightarrow ((\lambda x .  M')\ N, w_1, s_1)$
                  %
                  %By Lemma \ref{lemma:congr-closure} again (for $E = (\lambda x .  M')\ [\ ]$), 
                  %$((\lambda x .  M')\ N,w_1,s_1) \Rightarrow ((\lambda x .  M')\ \alpha, w_1 w_2, s_1 \mathrel{@} s_2)$.
                  %
                  %By \tr{Red Pure}, $((\lambda x .  M')\ \alpha, w_1 w_2, s_1 \mathrel{@} s_2) \rightarrow (\alpha,w_1 w_2,s_1 @ s_2)$.
                  
   \item Case:
     \restaterule{Eval If True}
                 { M_2 \Downarrow^{s}_{w} G}
                 {\ite{\ttrue}{M_2}{M_3} \Downarrow^{s}_{w} G}
                 
                 By \tr{Red Pure} (taking $E = [\ ]$), $(\ite{\ttrue}{M_2}{M_3}, 1, s) \rightarrow (M_2, 1, s)$.
                 By induction hypothesis, $(M_2, 1, s) \Rightarrow (G, w, [])$. 
                 
                 Hence $(\ite{1}{M_2}{M_3}, 1, s) \Rightarrow (G, w, [])$
               \item Case \tr{Eval If False}: analogous to \tr{Eval If True}
                 
               \item Case:
                 \restaterule{Eval Fail}
                             {}
                             {T \Downarrow^{[]}_{1} \fail}
                             
                             By \tr{Red Pure}, $(T, 1, []) \rightarrow (\fail, 1, [])$.
  \end{varitemize}
%\item
    The right to left implication can be proved by an induction on the
    length of the derivation of $(M,1,s) \Rightarrow (G,w, [])$.
    \begin{varitemize}
    \item Base case: If $(M,1,s) = (G,w, [])$, then $M \Downarrow^w_s G$ by \tr{Eval Val}.
    \item Induction step:
      assume $(M,1,s) \rightarrow (M', w', s')\rightarrow^{n} (G,w, [])$. 
      % By Lemma \ref{lemma:init-weight-trace-const}.\ref{item:iwtc.mult}
      % and Lemma \ref{lemma:small-steps-determined}
      % we have $w = w_1 w_2$ and $s = s_1 @ s_2$ for some $w_{2},s_{2}$.
      %
      If $(M,1,s) \rightarrow (M', w', s')$ was derived with \tr{Red Random Fail}, then
      $M = E[D(\vec{c})]$, $n=1$, $s = [c]$, $G = \fail$ and $w = w' = \pdf{D}(\vec{c},c) = 0$.
      By Lemma \ref{lemma:eval-random-fail}, we have $M \Downarrow^[c]_0 \fail$, as required.
      
      Otherwise, by Lemma \ref{lemma:w-greater-zero-strict}, $w' > 0$, so
      by Lemma \ref{lemma:multiply-closure}, $(M', 1, s') \rightarrow^n (G, w / w', [])$.
      %problem with universal quantification
      By induction hypothesis, $M' \Downarrow^{s'}_{w / w'} G$.
      By Lemma \ref{lemma:remove-suffix}, $(M,1,s^*) \rightarrow (M', w', [])$,
      where $s = s^* @ s'$.
      
      Therefore, by Lemma \ref{lemma:split}, $M \Downarrow^{s^* @ s'}_{w} G$, and so
      $M \Downarrow^{s}_{w} G$.\qed
    \end{varitemize}
%\end{varitemize}
%\qed
\end{proof}

As a corollary of Theorem~\ref{thm:sample-small-big-eq} and Lemma~\ref{lemma:small-steps-determined} we obtain:
\begin{lemma} \label{lemma:big-step-determined}
If $M \Downarrow^s_w G$ and $M \Downarrow^{s}_{w'} G'$ then $w=w'$ and $G=G'$.
\end{lemma}
%\begin{proof}
%Corollary of 
%%By induction on the derivation of $M \Downarrow^s_w G$,
%%Using the fact that if $M \Downarrow^s_w G$, then $M$ does not
%%reduce to a generalized value under any prefixs or extension of $s$.
%%-- Marcin to check
%\end{proof}

At this point, we have defined intuitive operational semantics based
on the consumption of an explicit trace of randomness, but we have defined
no distributions.  In the rest of this section we show
that this semantics indeed associates a sub-probability distribution
with each term. Before proceeding, however, we need some measure theory.

%%%%%%%%%%%%%%%%%%%%%%%%%%%%%%%%%%%%%%%%%%%%%%%%
\subsection{Some Measure-Theoretic Preliminaries.}\label{sect:smtp}
%%%%%%%%%%%%%%%%%%%%%%%%%%%%%%%%%%%%%%%%%%%%%%%%%
We begin by recapitulating some standard definitions for
sub-probability distributions and kernels over metric spaces.
For a more complete, tutorial-style introduction to measure theory,
see \citet{billingsley95}, \citet{panangaden2009labelled}, or another standard textbook or
lecture notes.

A $\sigma$-algebra (over a set $X$) is a set $\mAlg$ of subsets of $X$ that contains $\0$, and is closed
under complement and countable union (and hence is closed under
countable intersection). Let the $\sigma$-algebra \emph{generated} by
$S$, written $\sigma(S)$, be the set $\sigma(S)$ %, when
  $S\subseteq\Powerset{X}$, that is the least $\sigma$-algebra
over $\cup S$ that is a superset of $S$.
  In other words, $\sigma(S)$ is the least set such that:
  \begin{enumerate}
  \item we have $S \subseteq \sigma(S)$ and $\0 \in \sigma(S)$;
    and % and $\cup S \in \sigma(S)$;
  \item $((\cup S) \setminus A) \in \sigma(S)$ if $A \in \sigma(S)$;
    and
  \item $\cup_{i \in \Nat} A_i \in \sigma(S)$ if each
    $A_i \in \sigma(S)$.
  \end{enumerate}
  An equivalent definition is that
  $\sigma(S) \deq \bigcap \Set{\Sigma \mid \mbox{$S \subseteq \Sigma$
      and $\Sigma$ is a $\sigma$-algebra}}$.

We write $\RealInf$ for $[0,\infty]$ and $\RRp$ for the interval
$[0,1]$.
A \emph{metric space} is a set $\setone$ with a symmetric
\emph{distance function} $\delta:\setone\times\setone\to\RealInf$ that
satisfies the triangle inequality $\delta(x,z)\le\delta(x,y) +
\delta(y,z)$ and the axiom $\delta(x,x)=0$.  We write
$\mathbf{B}(x,r)\deq\Set{y\mid\delta(x,y)< r}$ for the open ball
around~$x$ of radius $r$. We equip $\RealInf$ and $\RRp$ with the
standard metric $\delta(x,y)=\Abs{x-y}$, and products of metric spaces
with the Manhattan metric (e.g., $\delta((x_{1},x_{2}), (y_{1},y_{2}))
= \delta (x_{1},y_{1}) + \delta(x_{2},y_{2})$).
The \emph{Borel} $\sigma$-\emph{algebra} on a metric space
$(\setone,\delta)$ is $\Borel(\setone,
\delta)\deq\sigma(\Set{\mathbf{B}(x,r)\mid x\in\setone\land r>0})$.
We often omit the arguments to $\Borel$ when they are clear from the
context.

A \emph{measurable space} is a pair $(X,\mAlg)$ where $X$ is a set of
possible outcomes, and $\mAlg \subseteq \Powerset{X}$ is a
$\sigma$-algebra of \emph{measurable sets}. As an example, consider the extended positive
real numbers $\RealInf$ equipped with the Borel $\sigma$-algebra
$\mathcal{R}$, i.e.~the set
$\sigma(\Set{(a,b)\mid a,b\ge 0})$ which is the smallest
$\sigma$-algebra containing all open (and closed) intervals. We can
create finite products of measurable spaces by iterating the construction
$(X, \mAlg)\times(X', \mAlg') = (X\times X', \sigma(A\times B\mid A\in
\mAlg\land B\in\mAlg'))$. If $(X, \mAlg)$ and $(X', \mAlg')$ are measurable spaces, then the
function $f : X \to X'$ is \emph{measurable} if and only if for all
$A\in\mAlg'$, $f^{-1}(A)\in\mAlg$, where the \emph{inverse image}
$f^{-1} : \Powerset{X'} \to \Powerset{X}$ is given by $f^{-1}(A) \deq
\{x \in X \mid f(x) \in A\}$.

A \emph{measure} $\mu$ on $(X,\mAlg)$ is a function from $\Sigma$ to
$\RealInf$, that is (1) zero on the empty set, that is,
$\mu(\emptyset)=0$, and (2) countably additive, that is,
$\mu(\cup_iA_i)=\Sigma_i\mu(A_i)$ if $A_1,A_2,\dots$ are pair-wise
disjoint.  The measure $\mu$ is called a \emph{(sub-probability) distribution} if $\mu(X) \le 1$
%, a \emph{probability distribution} if $\mu(X) = 1$, 
and \emph{finite} if $\mu(X)\neq\infty$.
If $\mu,\nu$ are finite measures and $c\ge0$, 
we write $c\cdot\mu$ for the finite measure $A\mapsto c\cdot(\mu(A))$ 
and $\mu + \nu$ for the finite measure $A\mapsto \mu(A)+\nu(A)$.  
 We write $\emdistr$ for the zero measure $A \mapsto 0$.
For any element $x$ of $X$, the Dirac measure $\dirac x$ is defined
as follows:
$$
\dirac x(A)=
\left\{
\begin{array}{ll}
  1 & \mbox{if $x\in A$;}\\
  0 & \mbox{otherwise.}
\end{array}
\right.
$$

A \emph{measure space} is a triple $\msone=(\setone,\mAlg,\mu)$ where
$\mu$ is a measure on the measurable space $(X,\mAlg)$. Given a
measurable function $\funone : \setone \to \RealInf$, the
\emph{integral} of $\funone$ over $\msone$ can be defined following
Lebesgue's theory and denoted as either of
$$
\int\!\funone\;d\mu = \int\!\funone(x)\;\mu(dx)\in\RealInf.
$$
The Iverson brackets $[P]$ are 1 if predicate $P$ is true, and 0 otherwise.
We then write
$$
\int_A\funone\;d\mu\deq
\int\! \funone(x)\cdot\indfun Ax \,\mu(dx).
$$

We equip some measurable spaces $(\setone,\algone)$ with a \emph{stock
  measure}~$\mu$.  We then
write $\int\!f(s)\,ds$ (or shorter, $\int\!f$) for $\int\!f\,d\mu$ when $f$ is measurable $f:\setone\to\RealInf$.
In particular, we let the stock measure on $(\Real^{n},\Borel)$ be the Lebesgue measure $\lambda_n$.

A function $f$ is a \emph{density} of a measure $\nu$ (with respect
to the measure $\mu$) if $\nu(A) = \int_A f\,d\mu$ for all measurable
$A$.

Given a measurable set $\tsetone$ from $(\setone,\algone)$,
we write $\restr{\algone}{\tsetone}$ for the \emph{restriction} of
$\algone$ to elements \emph{in} $\tsetone$, \ie,
$\restr{\algone}{\tsetone}=\{\tsettwo\cap\tsetone
\mid\tsettwo\in\algone\}$. Then
$(\tsetone,\restr{\algone}{\tsetone})$ is a measurable space. 
% \longv{Dually, the restriction of $\distrone$ to elements \emph{not in}
% $\tsetone$ is $\negrestr{\distrone}{\tsetone}$.}
% Then, an \emph{$\tsetone$-distribution} is a distribution on
% $(\tsetone,\restr{\measterms}{\tsetone})$\longv{, that is, a measure
% $\distrone:\restr{\measterms}{\tsetone}\rightarrow\RRp$ such that
% $\distrone(\tsetone)\leq 1$}.  
%We call a  $\gvalset$-distribution 
Any  distribution $\mu$ on $(\setone,\algone)$ trivially yields a distribution $\restr{\mu}{\tsetone}$ on $(\tsetone,\restr{\algone}{\tsetone})$ by 
$\restr{\mu}{\tsetone}(\tsettwo)= \mu(\tsettwo)$.  
% This construction yields a measure space
% $\mst{\tsetone}{\distrone}=(\tsetone,\restr{\measterms}{\tsetone},\restr{\distrone}{\tsetone})$.

%%%%%%%%%%%%%%%%%%%%%%%%%%%%%%%%%%%%%%%%%%%%%%%%%%
\subsection{Measure Space of Program Traces}\label{sec:measure-space-on-program-traces}
%%%%%%%%%%%%%%%%%%%%%%%%%%%%%%%%%%%%%%%%%%%%%%%%%%
In this section, we construct a measure space on the set $\sampseq$
of program traces: (1) we define a measurable space $(\sampseq, {\cal S})$
and (2) we equip it with a stock measure $\mu$ to obtain our measure
space $(\sampseq, {\cal S}, \mu)$.

%%%%%%%%%%%%%%%%%%%%%%%%%%%%%%%%%%%%%%%%%%%%%%%%%%
\paragraph{The Measurable Space of Program Traces}
%%%%%%%%%%%%%%%%%%%%%%%%%%%%%%%%%%%%%%%%%%%%%%%%%%
To define the semantics of a program as a measure on the space of random choices, 
we first need to define a measurable space of program traces. 
Since a program trace is a sequence of real numbers
of an arbitrary length (possibly 0), the set of all program traces is
%biguplus
$\sampseq = \biguplus_{n \in \mathbb{N}} \mathbb{R}^n$. Now, let us define the $\sigma$-algebra
$\mathcal{S}$ on $\sampseq$ as follows:
let $\Borel_n$ be the Borel $\sigma$-algebra on $\mathbb{R}^n$ (we take $\Borel_0$ to be $\{\{[] \}, \{\} \}$).
Consider the class of sets $\mathcal{S}$ of the form:
\[
  A = \biguplus_{n \in \mathbb{N}} H_n
\]
where $H_n \in \Borel_n$ for all $n$. Then $\mathcal{S}$ is a $\sigma$-algebra, and so $(\sampseq,\mathcal{S})$ is a measurable space.

%  Now, let
%\[
%\mathcal{S} = \sigma(\bigcup_{n \in \mathbb{N}}\mathcal{R}^n)
%\]
%
%where $\sigma_(A)$ is the smallest $\sigma$-algebra on $E$ generated by set $A$.

\begin{lemma}
$\mathcal{S}$ is a 
%countably-generated 
$\sigma$-algebra on $\sampseq$.
\end{lemma}
  \begin{proof}
    We have $\sampseq = \biguplus_{n \in \mathbb{N}}\mathbb{R}^n$ and
    $\mathbb{R}^n \in \mathcal{R}^n$ for all $n$, so
    $\sampseq \in \mathcal{S}$.

    If $A$ is defined as above, then
    $\sampseq - A = \biguplus_{n \in \mathbb{N}} (\mathbb{R}^n -
    H_n)$,
    where $\mathbb{R}^n - H_n \in \mathcal{R}^n$ for all $n$, so
    $\sampseq - A \in \mathcal{S}$.

    For closure under countable union, take
    $A_i = \biguplus_{n \in \mathbb{N}}H_{in}$ for all
    $i \in \mathbb{N}$, where $H_{in} \in \mathcal{R}^n$ for all $i$,
    $n$.  Then
    $\bigcup_{i \in \mathbb{N}} A_i = \bigcup_{i \in \mathbb{N}}
    \biguplus_{n \in \mathbb{N}} H_{in} = \biguplus_{n \in \mathbb{N}}
    (\bigcup_{i \in \mathbb{N}} H_{in}) \in \mathcal{S}$,
    because $\bigcup_{i \in \mathbb{N}} H_{in} \in \mathcal{R}^n$.

    Thus, $\mathcal{S}$ is a $\sigma$-algebra on $E$.
    % $\mathcal{S}$ is a $\sigma$-algebra by construction. Since each
    % Borel set $\mathcal{B}^n$ is countably generated, $\mathcal{S}$
    % is generated by a countable union of countable sets, and so it
    % is countably generated
    \qed
  \end{proof}
\paragraph{Stock Measure on Program Traces}
%%%%%%%%%%%%%%%%%%%%%%%%%%%%%%%%%%%%%%%%%%%
Since each primitive distribution $\distone$ has a density, 
the probability of each random value (and thus of each trace of random values) is zero. 
Instead, we define the trace and transition probabilities in terms of densities,
with respect to the stock measure $\mu$ on $(\sampseq, \mathcal{S})$ defined as follows:
%\[
%\mu(A) = \sum_{n=1}^{\infty} \lambda_n(A_n)
%\]
\[
\mu\left(\biguplus_{n \in \mathbb{N}} H_n\right) = \sum_{n\in\Nat} \lambda_n(H_n)
\]
\noindent
where $\lambda_0 = \delta({[]})$ and $\lambda_n$ is the Lebesgue measure on $\mathbb{R}^n$ for $n>0$. 
\begin{lemma}
$\mu$ is a measure on $(\sampseq, \mathcal{S})$.
\end{lemma}
  \begin{proof}
    We check the three properties:
    \begin{enumerate}
    \item Since for all $n \in \mathbb{N}$ and
      $H_n \in \mathcal{R}^n$, we have
      $\lambda_n(H_n) \in [0, \infty]$, obviously
      $\mu(\biguplus_{n \in \mathbb{N}} H_n) = \sum_{n=1}^{\infty}
      \lambda_n(H_n) \in [0, \infty]$
    \item If $H= \biguplus_{n \in \mathbb{N}} H_n =\emptyset$, then
      $H_n=\emptyset$ for all $n$, so
      $\mu(H) = \sum_{n=1}^{\infty} \lambda_n(\emptyset) = 0$.
    \item Countable additivity: if
      $H_1= \biguplus_{n \in \mathbb{N}} H_{1 n} , H_2 =\biguplus_{n
        \in \mathbb{N}} H_{2n},\dots$
      is a sequence of disjoint sets in $\mathcal{S}$, then:

\begin{eqnarray*}
  \mu(\biguplus_{m =1}^{\infty}H_{m})
  &=& \mu(\biguplus_{m =1}^{\infty} \biguplus_{n =0}^{\infty} H_{m n})\\ 
  &=& \mu(\biguplus_{n =0}^{\infty} \biguplus_{m =1}^{\infty} H_{m n})\\ 
  &=& \sum_{n =0}^{\infty}  \lambda_n(\biguplus_{m =1}^{\infty}H_{m n})\\
  &=& \sum_{n =0}^{\infty} \sum_{m=1}^{\infty} \lambda_n(H_{m n})\\
  &=& \sum_{m =1}^{\infty} \sum_{n=0}^{\infty} \lambda_n(H_{m n})\\
  &=& \sum_{m =1}^{\infty} \mu(H_m)
\end{eqnarray*}

where the equality
$\sum_{n =0}^{\infty} \sum_{m=1}^{\infty} \lambda_n(H_{m n}) = \sum_{m
  =1}^{\infty} \sum_{n=0}^{\infty} \lambda_n(H_{m n})$
follows from Tonelli's theorem for
series (see \cite{tao2011measure}).

% \begin{eqnarray*}
%\mu(\bigcup_{m =1}^{\infty}A^m)
%&=& \sum_{n =1}^{\infty}  \lambda_n((\bigcup_{m =1}^{\infty}A^m) \cap \mathbb{R}^n)\\
%&=& \sum_{n =1}^{\infty}  \lambda_n(\bigcup_{m =1}^{\infty}A^m_n)\\
%&=& \sum_{n =1}^{\infty} \sum_{m=1}^{\infty} \lambda_n(A^m_n)\\
%&=& \sum_{m =1}^{\infty} \sum_{n=1}^{\infty} \lambda_n(A^m_n)\\
%&=& \sum_{m =1}^{\infty} \mu(A^m)
%\end{eqnarray*}
\end{enumerate}
\qed
\end{proof}

A measure $\mu$ on $(X, \Sigma)$ is $\sigma$-finite if
$X = \bigcup_{i} A_i$ for some countable (finite or infinite) sequence
of sets $A_i \in \Sigma$ such that $\mu(A_i) < \infty$. If $\mu$ is a
$\sigma$-finite measure on $(X, \Sigma)$, the measure space
$(X, \Sigma, \mu)$ is also called $\sigma$-finite.
 $\sigma$-finite measure spaces behave better with respect to integration than those who are not.\MS{P1 substantiate}

In the following, let $[a,b]^n = \{(x_1, \dots, x_n)\ |\ 
x_i \in [a,b]\quad \forall i \in 1..n  \}$.

\begin{lemma}
The measure $\mu$ on $(\sampseq, \mathcal{S})$ is $\sigma$-finite.
\end{lemma}
\begin{proof}
For every $n \in \mathbb{N}$, we have that $\mathbb{R}^n = \bigcup_{k \in 
\mathbb{N}} [-k, k]^n$.
Hence, $\sampseq = \biguplus_{n \in \mathbb{N}} \mathbb{R}^n
= \biguplus_{n \in \mathbb{N}}\bigcup_{k \in \mathbb{N}} [-k, k]^n$ is 
a countable union of sets in $\mathcal{S}$ of the form  $[-k, k]^n$.
Finally,  for all $k, n \in \mathbb{N}$ we have $\mu([-k, k]^n) = \lambda_n([-k, k]^n) = (2k)^{n}< \infty$.
\qed
\end{proof}

It follows that $(\sampseq, \mathcal{S},\mu)$ is a $\sigma$-finite measure space.

%%%%%%%%%%%%%%%%%%%%%%%%%%%%%%%%%%%%%
\paragraph{Discrete Random Variables}
%%%%%%%%%%%%%%%%%%%%%%%%%%%%%%%%%%%%%
We have taken the sample space to be the real numbers, but any
complete separable metric space will do.  For example, in order to add
discrete distributions to the language we can change $\mathbb{S}$ to
$\biguplus_{n \in \mathbb{N}} (\mathbb{R} \uplus \mathbb{N}) ^n$.
  The measurable sets $H_n$ become
  $(\mathcal{R} \uplus \mathcal{P}(\mathbb{N}))^n$.  The stock measure
  $\mu$ would become $\sum_{n=0}^{\infty}(\lambda,\mu_{\#})^n (H_n)$,
  where $\mu_{\#}$ is the counting measure on $\mathbb{N}$
  (that is, $\mu_{\#}(A) = |A|$ if $A$ is finite, otherwise $\mu_{\#}(A) = \infty$).

Discrete distributions have probability mass functions,
that is, densities with respect to the counting measure on $\mathbb{N}$, 
which are trivially zero-extended to densities with respect to $\mu$.
Given a measurable injection function $\Inject{\cdot}:\mathbb{R} \uplus \mathbb{N}\to\valset$ (e.g., mapping natural numbers to the corresponding reals, or to Church numerals),
it is easy to update the different semantics to a more general sample space, e.g., 
\[
\AxiomC{$w = \pdf{\distone}(\vec{c}, z)$}
\AxiomC{$w > 0$}
\BinaryInfC{$\distone (\vec{c}) \Downarrow^{[z]}_{w} \Inject{z}$}
\DisplayProof{\;\tr{Eval Random}}
\]
\subsection{Distributions $\Tracedist{\termone}$ and $\sbd{\termone}$ Given by Sampling-Based Semantics}
\label{sec:from-trace-semantics}
%%%%%%%%%%%%%%%%%%%%%%%%%%%%%%%%%%%%%%%%%%%%%%%%%%%%%%%%%%%%%%%%%%%%%%%%%%%%%%
The result of a closed term $M$ on a given trace is 
\begin{align*}
  \obssmp{\termone}{s}&=\left\{
  \begin{array}{ll}
    \gvalone& \mbox{if $\termone\Downarrow^s_w\gvalone$ for some $w\in\RealInf$}\\
    \fail & \mbox{otherwise.}
  \end{array}
  \right.
\end{align*}
 The density of termination of a closed term $M$ on a given trace is defined as follows.
\begin{align*}
  \termsmp{\termone}{s}&=\left\{
  \begin{array}{ll}
    w & \mbox{if $\termone\Downarrow^s_w \gvalone$ for some $\gvalone\in\gvalset$}\\
    0 & \mbox{otherwise}
  \end{array}
  \right.
\end{align*}
This density function induces a distribution $\Tracedist{\termone}$ on traces defined as
$\Tracedist{\termone}(\tsetone):=\int_{\tsetone}\mathbf{P}_{\termone}$.

By inverting the result function $\Obssmp\termone$, 
we also obtain a distribution $\sbd{\termone}$ over
generalised values (also called a \emph{result distribution}).
It can be computed by integrating the density of termination over all traces that yield the generalised values of interest.
$$
\sbd{\termone}(\tsetone):=\Tracedist M(\mathbf{O}^{-1}_M(A))=\int\termsmp{\termone}{s}\cdot\indfun{\tsetone}{\obssmp{\termone}{s}}\,ds.
$$

As an example, for the geometric distribution example of Section~\ref{subsec:geom} we have 
$\obssmp{\mathit{geometric}\ 0.5}s=n$ if $s\in[0.5,1]^n[0,0.5)$, and otherwise 
$\obssmp{\mathit{geometric}\ 0.5}s=\mathtt{fail}$. Similarly, we have 
$\termsmp{\mathit{geometric}\ 0.5}{s}=1$ if $s\in[0.5,1]^{n}[0,0.5)$ for some $n$, and otherwise~$0$.
We then obtain 
\begin{align*}
\Tracedist{\mathit{geometric}\ 0.5}(A)={}&\sum_{n\in\mathbb{N}}\lambda_{n+1}(A\cap\{[0.5,1]^{n}[0,0.5)\})
\text{~and~}\\
\sbd{\mathit{geometric}\ 0.5}(\{n\})={}&\int\![s\in\{[0.5,1]^{n}[0,0.5)\}]\,ds=\frac 1{2^{n+1}}.
\end{align*}

As seen above, we use the exception $\fail$ to model the failure of a hard constraint.  
To restrict attention to normal termination, we modify~$\mathbf{P}_{M}$ as follows.
  \begin{align*}
        \termsmv{\termone}{s}&=\left\{
      \begin{array}{ll}
        w & \mbox{if $\termone\Downarrow^s_w\valone$ for some $\valone\in\valset$}\\
        0 & \mbox{otherwise.}
      \end{array}
      \right.
  \end{align*}
As above, this density function generates distributions over
traces and values as, respectively
\[
  \Tracevaldist{\termone}(A) := \int_\tsetone \mathbf{P}^{\valset}_{\termone} = \Tracedist{\termone}(\tsetone \cap \mathbf{O}_\termone^{-1}(\valset) )
\]
\[
  \restr{(\sbd{\termone})}{\valset}(\tsetone)%:={}&\sbd{\termone}(\tsetone\cap\valset)
:= \Tracevaldist M(\mathbf{O}^{-1}_M(A)) = \int\termsmv{\termone}{s}\cdot[\obssmp{\termone}{s}\in\tsetone]\,ds
\]

\begin{figure}
\begin{center}
\fbox{
\begin{minipage}{.442\textwidth}
\begin{center}
\begin{align*}
  d(\varone,\varone) &{}=0\\
  d(c,d) &{}=\Abs{c-d}\\
  d(\termone \termtwo,\termthree\termfour) &{}=d({\termone},{\termthree}) + d({\termtwo},{\termfour}) \\
%  d((M,M'), (N,N')) &{} = d(M,N) + d(M',N')\\
  d(g(V_1,\dots,V_n), &g(W_1,\dots,W_n))\\
   &{} = d(V_1, W_1) + \dots + d(V_n, W_n) \\
%  d(\trycatch{\expone}{\termone}{\termone'},\trycatch{\expone}{\termtwo}{\termtwo'}) &{}=d({\termone},{\termtwo}) + d({\termone'},{\termtwo'}) \\
  d(\abstr{\varone}{\termone},\abstr{\varone}{\termtwo})&{}=d({\termone},{\termtwo}) \\
%  d(\raiset{\expone}{\termone},\raiset{\expone}{\termtwo}) &{}=d({\termone},{\termtwo}) \\
  d(\distone(V_1,\dots,V_n),&\distone({W_1,\dots,W_n}))\\
    &{}=d(V_1, W_1) + \dots + d(V_n, W_n) \\
  d(\mathtt{score}(V),\mathtt{score}(W))
    &{}=d(V, W) \\
%  d(\Llet{\termone}{\termone'}{\termone''}, \Llet{\termtwo}{\termtwo'}{\termtwo''})
%  &{}= d(\termone,\termtwo) + d(\termone',\termtwo') +
%  d(\termone'',\termtwo'') \\
  d(\ite{\valone}{\termone}{\termtwo},&\ite{\valtwo}{\termthree}{\termfour})\\
   &{}= d(\valone,\valtwo) + d(\termone,\termthree)+d(\termtwo,\termfour) \\
%  d(M.1, N.1) &{} = d(M,N) \\
%  d(M.2, N.2) &{} = d(M,N) \\
  d(\fail, \fail) &{} = 0 \\
d({\termone},{\termtwo}) &{}=\infty \text{ otherwise}
\end{align*}
\end{center}
\end{minipage}}
\condnocr
\caption{Metric $d$ on terms.}\label{fig:metric}
\end{center}
\end{figure}
To show that the above definitions make sense measure-theo\-retically, 
we first define the measurable space of terms $(\lamterms,\measterms)$,
where $\measterms$ is the set of
Borel-measurable sets of terms with respect to the recursively defined
metric $d$ in Figure~\ref{fig:metric}.

%however, some functions need to be measurable:
\begin{lemma} \label{lemma:pi-measurable}
  For any closed term $M$, the functions $\mathbf{P}_{M}$, $\mathbf{O}_M$ and~$\mathbf{P}_{M}^\valset$ are all measurable; 
$\Tracedist{\termone}$ and $\Tracevaldist\termone$ are measures on $(\sampseq, \mathcal{S})$; 
$\sbd{\termone}$ is a measure on $(\gvalset,\restr\measterms\gvalset)$; 
and $\restr{(\sbd{\termone})}\valset$ is a measure on $(\valset,\restr\measterms\valset)$.

\end{lemma}
  \begin{proof}
    See \Appref{section:proof-of-measurability}.
    \qed
  \end{proof}

\section{Distribution-Based Operational Semantics}\label{sec:operational-semantics}
%%%%%%%%%%%%%%%%%%%%%%%%%%%%%%%%%%%%
In this section we introduce small- and big-step distribution-based operational semantics, 
where the small-step semantics is a generalisation of \citet{jones90:PhD} to continuous distributions. 
We prove correspondence between the semantics using some non-obvious properties of kernels. 
Moreover, we will prove that the distribution-based semantics are equivalent to
the sampling-based semantics from Section \ref{sect:samplingsemantics}.
A term will correspond to a distribution over generalised values,
below called a result distribution.  A term $\termone$ is said to be
\emph{skeleton} iff no real number occurs in $\termone$, and each
variable occurs \emph{at most once} in $\termone$. The set of
skeletons is $\skset$. Any closed term $\termone$ can be written as
$\sbst{\termtwo}{\vec{\varone}}{\vec{c}}$, where $\termtwo$ is a
skeleton. The set of closed terms corresponding this way to a skeleton
$\termone\in\skset$ is denoted as $\tmsofsk{\termone}$.  If the
underlying term is a skeleton, substitution can be defined also when
the substituted terms are \emph{sets} of values rather than mere
values, because variables occurs at most once; in that case, we will
used the notation $\sbst{\termone}{\varone}{\setone}$, where $\setone$
is any set of values.

%%%%%%%%%%%%%%%%%%%%%%%%%%%%%%%%%%%
\subsection{Sub-Probability Kernels}
%%%%%%%%%%%%%%%%%%%%%%%%%%%%%%%%%%%
If $(\setone,\algone)$ and $(\settwo,\algtwo)$ are measurable spaces,
then a function $Q:\setone\times\algtwo\rightarrow \RRp$ is called a
\emph{(sub-probability) kernel} (from $(\setone,\algone)$ to
$(\settwo,\algtwo)$) if
\begin{varenumerate}
\item for every $x \in \setone$, $Q(x, \cdot)$ is a sub-probability distribution on $(\settwo,\algtwo)$; and
\item for every $A \in \algtwo$, $Q(\cdot,A)$ is a non-negative measurable function
  $\setone\to\RRp$.
\end{varenumerate}
%$Q$ is said to be a \emph{probability kernel} if $Q(x,\settwo)=1$ for all $x\in\setone$. 
The measurable function $q:\setone\times\settwo\to\RealInf$ is said to be a \emph{density} of
kernel $Q$ with respect to a measure $\mu$ on $(\settwo,\algtwo)$ 
if $Q(v,A) = \int_A q(v,y)\,\mu(dy)$ for all $v\in\setone$ and $A\in\algtwo$.
 When~$Q$ is a kernel, note that $\int\!f(y)\,Q(x,dy)$ denotes the
integral of~$f$ with respect to the measure $Q(x,\cdot)$.

Kernels can be composed in the following ways: If $Q_1$ is a kernel
from $(\setone_{1},\algone_{1})$ to $(\setone_{2},\algone_{2})$ and
$Q_2$ is a kernel from $(\setone_{2},\algone_{2})$ to
$(\setone_{3},\algone_{3})$, then $Q_{2}\circ Q_{1}:(x,A) \mapsto \int
Q_{2}(y,A)\,Q_{1}(x,dy)$ is a kernel from $(\setone_{1},\algone_{1})$
to $(\setone_{3},\algone_{3})$.  Moreover, if $Q_1$ is a kernel from
$(\setone_{1},\algone_{1})$ to $(\setone_{2},\algone_{2})$ and $Q_2$
is a kernel from $(\setone'_{1},\algone'_{1})$ to
$(\setone'_{2},\algone'_{2})$, then $Q_{1}\times Q_{2}:((x,y),(A\times
B)) \mapsto Q_{1}(x,A)\cdot Q_{2}(y,B)$ uniquely extends to a kernel
from $(\setone_{1},\algone_{1})\times (\setone'_{1},\algone'_{1})$ to
$(\setone_{2},\algone_{2})\times (\setone'_{2},\algone'_{2})$.

%%%%%%%%%%%%%%%%%%%%%%%%%%%%%%%%
\subsection{Approximation Small-Step Semantics}
%%%%%%%%%%%%%%%%%%%%%%%%%%%%%%%%
The first thing we need to do is to generalize deterministic reduction
into a relation between closed terms and term \emph{distributions}.
If $\mu$ is a measure on terms and $\ectxone$ is an evaluation
context, we let $\ctm{\ectxone}{\mu}$ be the push-forward measure
$A\mapsto \mu(\{\termone\mid\ct{\ectxone}{\termone}\in A\})$.

\emph{One-step evaluation} is a relation $M \red \distrone$ between closed terms $M$
and distributions $\distrone$ on terms, defined as follows:
\begin{align*}
  \ct{\ectxone}{\distone(\vec{c})}&\red\ctm{\ectxone}{\mu_{\distone(\vec{c})}}\\
  \ct{\ectxone}{M}&\red\dirac{\ct{\ectxone}{N}}\text{~if~}M\detred N\\
  \ct{\ectxone}{\score{c}}&\red c\cdot\dirac{\ct{\ectxone}{\texttt{true}}}\text{~if~}0<c\leq 1
  %\ct{\ectxone}{\score{c}}&\red \dirac{\fail}\text{~if~}c\leq 0\vee c>1
  % \ct{\ectxone}{(\abstr{\varone}{\termone}) \valone}&\red\dirac{\ct{\ectxone}{\sbst{\termone}{\varone}{\valone}}}\\
  % \ct{\ectxone}{c \termone}&\red\dirac{\ct{\ectxone}{\terror}}\\
  % \ct{\ectxone}{g({\vec{c}})}&\red\dirac{\ct{\ectxone}{\intfun{g}(\vec{c})}}\\
  % \ct{\ectxone}{\expone}&\red\dirac{\expone}\qquad\mbox{if $\ectxone$ is proper}\\
  % \ct{\ectxone}{\ite{\ttrue}{\termone}{\termtwo}}&\red\dirac{\ct{\ectxone}{\termone}}\\
  % \ct{\ectxone}{\ite{\tfalse}{\termone}{\termtwo}}&\red\dirac{\ct{\ectxone}{\termtwo}}\\
  % \ct{\ectxone}{\ertone}&\red\dirac{\ct{\ectxone}{\terror}}
\end{align*}
We first of all want to show that one-step reduction is essentially
deterministic, and that we have a form of deadlock-freedom.
\begin{lemma}\label{lemma:funct}
  For every closed term $\termone$, either $\termone$ is a generalized value
  or there is a unique $\distrone$ such that $\termone\red\distrone$.
\end{lemma}
\begin{proof}
  An easy consequence of Lemma~\ref{lemma:unique}.
  \qed
\end{proof}
We need to prove the just introduced notion of one-step reduction
to support composition. This is captured by the following result.
\begin{lemma}\label{lemma:redsk}
  $\red$ is a sub-probability kernel.
\end{lemma}
\newcommand{\OS}[2]{\mathit{OS}(#1,#2)}
Let $\redterms=\{E[R]\in\cterms\}$ be the set of all closed reducible terms.
  \begin{proof}
    Lemma~\ref{lemma:funct} already tells us that $\red$ can be seen
    as a function $\hat{\red}$ defined as follows:
    $$
    \hat{\red}(M,A)=\left\{
    \begin{array}{ll}
      \distrone(A) & \mbox{if $\termone\red\distrone$;}\\
      $0$ & \mbox{otherwise.}
    \end{array}
    \right.
    $$
    The fact that $\hat{\red}(\termone,\cdot)$ is a distribution is easily
    verified. On other hand, the fact that $\hat{\red}(\cdot,A)$ is measurable
    amounts to proving that $\OS{A}{B}=(\hat{\red}(\cdot,A))^{-1}(B)$ is a measurable set
    of terms whenever $B$ is a measurable set of real numbers. We will do that
    by showing that for every skeleton $\termtwo$, the set
    $\OS{A}{B}\cap\tmsofsk{\termtwo}$ is measurable. 
    The thesis then follows by observing that
    $$
    \OS{A}{B}=\bigcup_{\termtwo\in\skset}\OS{A}{B}\cap\tmsofsk{\termtwo}
    $$
    and that $\skset$ is countable. Now, let us observe that for every
    skeleton $\termtwo$, the nature of any term $\termthree$
    in $\tmsofsk{\termtwo}$
    as for if being a value, or containing a deterministic redex, or
    containing a sampling redex, only depends on $\termtwo$ and not on
    the term $\termthree$. As an example, terms in 
    $\tmsofsk{\varone\vartwo}$ are nothing but deterministic redexes
    (actually, all of them rewrites deterministically to 
    $\dirac{\terror}$. This allows us to proceed by distinguishing
    three cases:
    \begin{varitemize}
    \item
      If all terms in $\tmsofsk{\termtwo}$ are values, then
      it can be easily verified that
      $$
      \OS{A}{B}\cap\tmsofsk{\termtwo}=
      \left\{
      \begin{array}{ll}
         \tmsofsk{\termtwo} &\mbox{ if $0\in B$;}\\
         \emptyset &\mbox{ if $0\not\in B$.}
      \end{array}
      \right.
      $$
      Both when $0\in B$ and when $0\not\in B$, then, $\OS{A}{B}\cap\tmsofsk{\termtwo}$
      is indeed measurable.
    \item
      If all terms in $\tmsofsk{\termtwo}$ contain deterministic
      redexes, then
      $$
      \OS{A}{B}\cap\tmsofsk{\termtwo}=
      \left\{
      \begin{array}{ll}
        \rightarrow^{-1}(A)\cap\tmsofsk{\termtwo} &\mbox{ if $1\in B$}\\
        \emptyset &\mbox{ if $1\not\in B$.}
      \end{array}
      \right.
      $$
      Since deterministic reduction $\rightarrow$ is known to be measurable,
      then both when $1\in B$ and when $1\not\in B$, the set
      $\OS{A}{B}\cap\tmsofsk{\termtwo}$ is measurable.
    \item
      The hardest case is when $\termtwo$ is of the form
      $\ct{\ectxthree}{\distone(\vec{\varone})}$, where
      $\ectxthree$ is an evaluation context. In this case,
      however, we can proceed by decomposing the function we
      want to prove measurable into three measurable functions:
       \begin{varitemize}
       \item 
         The function
         $\app:\ectxs{}\times\cterms\rightarrow\cterms$,
         that given an evaluation context $\ectxone$ and a term $\termone$, 
         returns the term $\ct{\ectxone}{\termone}$. This is proved measurable
         in the Appendix.
       \item
         The function $\deapp:\redterms\rightarrow\ectxs{\emptyset}\times\cterms$
         that ``splits'' a term in $\redterms$ into an evaluation context $\ectxone$
         and a closed term $\termone$. This is proved measurable in the Appendix.
       \item
         For every distribution identifier $\distone$, the function $\distapp_\distone:\RR^n
         \rightarrow\cterms$ (where $n$ is the arity of $\distone$) that,
         given a tuple of real numbers $x$, returns the term $\distone(x)$. 
         This function is a continuous function between two metric spaces, so measurable.
       \item
         We know that for every distribution identifier $\distone$, there is a 
         kernel $\mu_\distone:\RR^n\times\Sigma_\RR\rightarrow\RRp$. Moreover, one
         can also consider the Dirac kernel on evaluation contexts, namely
         $I:\ectxs{\emptyset}\times\Sigma_{\ectxs{\emptyset}}\rightarrow\RRp$ where $I(E,A)=\indfun EA$. 
         Then, the product $\mu_\distone\times I$ is also a kernel, so measurable. 
       \end{varitemize}
    \end{varitemize}
    This concludes the proof.
    \qed
  \end{proof}

Given a family $\{\distrone_\termone\}_{\termone\in\tsetone}$ of
distributions indexed by terms in a measurable set $\tsetone$ of terms, 
and a measurable set $\tsettwo$, we often write, with an abuse of notation, 
$\distrone_\termone(\tsettwo)$ for the function that assigns to any term $\termone\in\tsetone$ 
the real number $\distrone_\termone(\tsettwo)$. 
The \emph{step-indexed approximation small-step semantics} is the family of $n$-indexed relations $\sisss{\termone}{n}{\distrone}$ between terms and result distributions
inductively defined in Figure \ref{fig:smallstepsem}.
\begin{SHORT}
\begin{figure}
  \fbox{
    \begin{minipage}{.442\textwidth}
      \[
      \AxiomC{$n>0$}
      \UnaryInfC{$\sisss{\gvalone}{n}{\dirac{\gvalone}}$}
      \DisplayProof{\;\tr{DRed Val}}
      \quad
      \AxiomC{}%$M\not\in\gvalset$}
      \UnaryInfC{$\sisss{\termone}{0}{\emdistr}$}
      \DisplayProof{\;\tr{DRed Empty}}
      \]
      \[
      \AxiomC{$\termone\red\distrone$}
      \AxiomC{$\{\sisss{\termtwo}{n}{\distrtwo_{\termtwo}}\}_{\termtwo\in\supp{\distrone}}$}
      \BinaryInfC{$\sisss{\termone}{n+1}{A\mapsto\int\distrtwo_{\termtwo}(A)\;\distrone(d\termtwo)}$}
      \DisplayProof{\;\tr{DRed Step}}
      \]
  \end{minipage}}
  \condnocr
  \caption{Step-Indexed Approximation Small-Step Semantics.}\label{fig:smallstepsem}
\end{figure}
\end{SHORT}
\begin{LONG}
\begin{figure}
\begin{center}
  \fbox{
    \begin{minipage}{.63\textwidth}
      \[
      \AxiomC{$n>0$}
      \UnaryInfC{$\sisss{\gvalone}{n}{\dirac{\gvalone}}$}
      \DisplayProof{\;\tr{DRed Val}}
      \quad
      \AxiomC{}%$M\not\in\gvalset$}
      \UnaryInfC{$\sisss{\termone}{0}{\emdistr}$}
      \DisplayProof{\;\tr{DRed Empty}}
      \]
      \[
      \AxiomC{$\termone\red\distrone$}
      \AxiomC{$\{\sisss{\termtwo}{n}{\distrtwo_{\termtwo}}\}_{\termtwo\in\supp{\distrone}}$}
      \BinaryInfC{$\sisss{\termone}{n+1}{A\mapsto\int\distrtwo_{\termtwo}(A)\;\distrone(d\termtwo)}$}
      \DisplayProof{\;\tr{DRed Step}}
      \]
  \end{minipage}}
  \condnocr
  \caption{Step-Indexed Approximation Small-Step Semantics.}\label{fig:smallstepsem}
\end{center}
\end{figure}
\end{LONG}
Since generalised values have no transitions (there is no $\distrone$ such that $G\to\distrone$), 
the rules above are disjoint and so there is at most one $\distrone$ such that $M\to_{n}\distrone$.
Compared to the discrete case \cite{jones90:PhD}, \JB{P0 supp is undefined.}
the step-index $n$ is needed to ensure that the integral in \textsc{(DRed Step)} is defined.
\begin{lemma}\label{lemma:sstk}
  For every $n\in\NN$, the function $\redss_n$ is a kernel.
\end{lemma}
  \begin{proof}
    By induction on $n$:
    \begin{varitemize}
    \item
      $\redss_0$ can be seen as the function $\hat{\red}_0$ that attributes $0$ to
      any pair $(\termone,A)$. This is clearly a kernel.
    \item
      $\redss_{n+1}$ can be seen as the function $\hat{\redss_{n+1}}$ defined as follows:
      $$
      \hat{\red}_{n+1}(\termone,A)=
      \left\{
        \begin{array}{ll}
          1 & \mbox{if $\termone\in\gvalset$ and $\termone\in A$};\\
          0 & \mbox{if $\termone\in\gvalset$ and $\termone\not\in A$};\\
          (\int\hat{\red}_n(\termtwo,A)\,\distrone(dN)) & \mbox{if $\termone\red\distrone$}.
        \end{array}
      \right.
      $$
      The fact that $\hat{\red}_{n+1}(\termone,\cdot)$ is a measure for every $\termone$
      is clear, and can be proved by case distinction on $\termone$. On the other hand,
      if $B$ is a measurable set of reals, then:
      \begin{align*}
      (\hat{\red}_{n+1}(\cdot,A))^{-1}(B)=&\;
        (\hat{\red}_{n+1}(\cdot,A))^{-1}(B)\cap\gvalset\;\cup\\
        &\;(\hat{\red}_{n+1}(\cdot,A))^{-1}(B)\cap(\cterms-\gvalset).
      \end{align*}
      Now, the fact that
      $(\hat{\red}_{n+1}(\cdot,A))^{-1}(B)\cap\gvalset$ is a measurable
      set of terms is clear: it is $A\cap\gvalset$ if $1\in B$ and
      $\emptyset$ otherwise.  But how about
      $(\hat{\red}_{n+1}(\cdot,A))^{-1}(B)\cap(\cterms-\gvalset)$?
      In that case, we just need to notice that
      $$
      (\hat{\red}_{n+1}(\cdot,A))^{-1}(B)\cap(\cterms-\gvalset)=(\hat{\red}_{n}\circ\red)^{-1}(B)
      $$
      where $\hat{\red}_n$ and $\red$ are kernels (the
      former by induction hypothesis, the latter by
      Lemma~\ref{lemma:redsk}). Since kernels compose, this concludes the proof.
    \end{varitemize}
    \qed
  \end{proof}
\begin{lemma}
  For every closed term $\termone$ and for every $n\in\NN$
  there is a unique distribution $\distrone$ such that
  $\termone\redss_n\distrone$.
\end{lemma}
\begin{proof}
  This is an easy consequence of Lemma~\ref{lemma:sstk}.\qed
\end{proof}
%%%%%%%%%%%%%%%%%%%%%%%%%%%%%%%%
\subsection{Approximation Big-Step Semantics}
%%%%%%%%%%%%%%%%%%%%%%%%%%%%%%%%
The \emph{step-indexed approximation big-step semantics}
$\sibss{M}{n}{\distrone}$ is the $n$-indexed family of relations
between terms and result distributions
inductively defined by the rules in Figure~\ref{fig:bigstepsem}.
\begin{SHORT}
\begin{center}
  \begin{figure*}
    \begin{center}
    \fbox{
      \begin{minipage}{.97\textwidth}
        \begin{center}
          \[
          \AxiomC{$n>0$}
          \UnaryInfC{$\sibss{\gvalone}{n}{\dirac{\gvalone}}$}
          \DisplayProof{\;\tr{DEval Val}}
          \qquad
          \AxiomC{}%$\termone\not\in\gvalset$}
          \UnaryInfC{$\sibss{\termone}{0}{\emdistr}$}
          \DisplayProof{\;\tr{DEval Empty}}
          \qquad
          \AxiomC{$n>0$}
          \UnaryInfC{$\sibss{\ertone}{n}{\dirac{\terror}}$}
          \DisplayProof{\;\tr{DEval Fail}}
          \]
          \[
          \AxiomC{$n>0$}
          \UnaryInfC{$\sibss{\distone(\vec{c})}{n}{\mu_{\distone(\vec{c})}}$}
          \DisplayProof{\;\tr{DEval Samp}}
          \qquad
          \AxiomC{$n>0$}
          \UnaryInfC{$\sibss{g(\vec{c})}{n}{\dirac{\intfun{g}(\vec{c})}}$}
          \DisplayProof{\;\tr{DEval Fun}}
          \qquad
          % \AxiomC{$c\leq 0\vee c>1$}
          % \AxiomC{$n>0$}
          % \BinaryInfC{$\sibss{\score{c}}{n}{\dirac{\fail}}$}
          % \DisplayProof          
          \AxiomC{$0<c\leq 1$}
          \AxiomC{$n>0$}
          \BinaryInfC{$\sibss{\score{c}}{n}{c\cdot\dirac{\texttt{true}}}$}
          \DisplayProof{\;\tr{DEval Score}}
          \]
          \[
          \AxiomC{$\sibss{\termone}{n}{\distrone}$}
          \UnaryInfC{$\sibss{\ite{\ttrue}{\termone}{\termtwo}}{n+1}{\distrone}$}
          \DisplayProof{\;\tr{DEval If True}}  
          \qquad
          \AxiomC{$\sibss{\termtwo}{n}{\distrone}$}
          \UnaryInfC{$\sibss{\ite{\tfalse}{\termone}{\termtwo}}{n+1}{\distrone}$}
          \DisplayProof{\;\tr{DEval If False}}
          \]
          \[
          \AxiomC{$\sibss{\termone}{n}{\distrone}$}
          \AxiomC{$\sibss{\termtwo}{n}{\distrtwo}$}
          \AxiomC{$\{\sibss{\sbst{\termthree}{\varone}{\valone}}{n}{\distrtwo_{\termthree,\valone}}\}_{(\abstr{\varone}{\termthree})\in\supp{\distrone},\valone\in\supp{\distrtwo}}$}
          \TrinaryInfC{$\sibss{\termone\termtwo}{n+1}{A\mapsto \restr{\distrone}{\expset}(A)+\distrone(\RR)\cdot\indfun{A}{\terror}+\distrone(\valset_\lambda)\cdot\restr{\distrtwo}{\expset}(A)+\iint\distrtwo_{\termthree,\valone}(A)\,\restr{\distrone}{\valset_\lambda}(\abstr{\varone}{d\termthree})\,\restr{\distrtwo}{\valset}(d\valone)}$}
          \DisplayProof{\;\tr{DEval Appl}}
          \]
        \end{center}
    \end{minipage}}
    \end{center}
    \condnocr
    \caption{Step Indexed Approximation Big-Step Semantics.}\label{fig:bigstepsem}
  \end{figure*}
  \end{center}
  \end{SHORT}
\begin{LONG}
\begin{center}
  \begin{figure*}
    \begin{center}
    \fbox{
      \begin{minipage}{.97\textwidth}
        \begin{center}
          \[
          \AxiomC{$n>0$}
          \UnaryInfC{$\sibss{\gvalone}{n}{\dirac{\gvalone}}$}
          \DisplayProof{\;\tr{DEval Val}}
          \qquad
          \AxiomC{}%$\termone\not\in\gvalset$}
          \UnaryInfC{$\sibss{\termone}{0}{\emdistr}$}
          \DisplayProof{\;\tr{DEval Empty}}
          \qquad
          \AxiomC{$n>0$}
          \UnaryInfC{$\sibss{\ertone}{n}{\dirac{\terror}}$}
          \DisplayProof{\;\tr{DEval Fail}}
          \]
          \[
          \AxiomC{$n>0$}
          \UnaryInfC{$\sibss{\distone(\vec{c})}{n}{\mu_{\distone(\vec{c})}}$}
          \DisplayProof{\;\tr{DEval Samp}}
          \qquad
          \AxiomC{$n>0$}
          \UnaryInfC{$\sibss{g(\vec{c})}{n}{\dirac{\intfun{g}(\vec{c})}}$}
          \DisplayProof{\;\tr{DEval Fun}}
          \]
          \[
          % \AxiomC{$c\leq 0\vee c>1$}
          % \AxiomC{$n>0$}
          % \BinaryInfC{$\sibss{\score{c}}{n}{\dirac{\fail}}$}
          % \DisplayProof          
          \AxiomC{$0<c\leq 1$}
          \AxiomC{$n>0$}
          \BinaryInfC{$\sibss{\score{c}}{n}{c\cdot\dirac{\texttt{true}}}$}
          \DisplayProof{\;\tr{DEval Score}}
          \]
          \[
          \AxiomC{$\sibss{\termone}{n}{\distrone}$}
          \UnaryInfC{$\sibss{\ite{\ttrue}{\termone}{\termtwo}}{n+1}{\distrone}$}
          \DisplayProof{\;\tr{DEval If True}}  
          \]
          \[
          \AxiomC{$\sibss{\termtwo}{n}{\distrone}$}
          \UnaryInfC{$\sibss{\ite{\tfalse}{\termone}{\termtwo}}{n+1}{\distrone}$}
          \DisplayProof{\;\tr{DEval If False}}
          \]
          \[
          \AxiomC{$\sibss{\termone}{n}{\distrone}$}
          \AxiomC{$\sibss{\termtwo}{n}{\distrtwo}$}
          \AxiomC{$\{\sibss{\sbst{\termthree}{\varone}{\valone}}{n}{\distrtwo_{\termthree,\valone}}\}_{(\abstr{\varone}{\termthree})\in\supp{\distrone},\valone\in\supp{\distrtwo}}$}
          \TrinaryInfC{$\sibss{\termone\termtwo}{n+1}{\begin{prog} A\mapsto \restr{\distrone}{\expset}(A)+\distrone(\RR)\cdot\indfun{A}{\terror}+\distrone(\valset_\lambda)\cdot\restr{\distrtwo}{\expset}(A)\\
          \qquad \qquad +\iint\distrtwo_{\termthree,\valone}(A)\,\restr{\distrone}{\valset_\lambda}(\abstr{\varone}{d\termthree})\,\restr{\distrtwo}{\valset}(d\valone) \end{prog}}$}
          \DisplayProof{\;\tr{DEval Appl}}
          \]
        \end{center}
    \end{minipage}}
    \end{center}
    \condnocr
    \caption{Step Indexed Approximation Big-Step Semantics.}\label{fig:bigstepsem}
  \end{figure*}
  \end{center}
  \end{LONG}
Above, the rule for applications is the most complex, with the
resulting distribution consisting of three exceptional terms in
addition to the normal case.  To better understand this rule, one
can study what happens if we replace general applications with a let
construct plus application of values to values. Then we would end up
having the following three rules, instead of the rule for
application above:
\[
\AxiomC{$\sibss{\termone}{n}{\distrone}$}
\AxiomC{$\{\sibss{\sbst{\termtwo}{\varone}{\valone}}{n}{\distrtwo_{\valone}}\}_{\valone\in\supp{\distrone}}$}
\BinaryInfC{$\sibss{\letbe{\termone}{\varone}{\termtwo}}{n+1}{A\mapsto 
    \left(
    \begin{array}{l}
      \restr{\distrone}{\expset}(A)+\distrone(\RR)\cdot\indfun{A}{\terror}\\
      +\int\distrtwo_{\valone}(A)\,\restr{\distrone}{\valset}(d\valone)}
    \end{array}
    \right)$}
\DisplayProof
\]
\[
\AxiomC{$\sibss{\sbst{\termone}{\varone}{\valone}}{n}{\distrtwo}$}
\UnaryInfC{$\sibss{(\abstr{\varone}{\termone})\valone}{n+1}{\distrtwo}$}
\DisplayProof
\qquad
\AxiomC{$n>0$}
\UnaryInfC{$\sibss{c\;\valone}{n}{\dirac{\terror}}$}
\DisplayProof  
\]
The existence of the integral in rule $\tr{DEval Appl}$ is
guaranteed by a lemma analogous to
Lemma \ref{lemma:sstk}.
\begin{lemma}\label{lemma:bstk}
  For every $n\in\NN$, the function $\Downarrow_n$ is a kernel.
\end{lemma}
This can be proved by induction on $n$, with the most difficult case
being precisely the one of applications. Composition and product 
properties of kernels are the key ingredients there.
%%%%%%%%%%%%%%%%%%%%%%%%%%%%%%%%%%
\subsection{Beyond Approximations}
%%%%%%%%%%%%%%%%%%%%%%%%%%%%%%%%%%
The set of result distributions with the pointwise order forms an
$\oCPO$, and thus any denumerable, directed set of result
distributions has a least upper bound.  One can define the
\emph{small-step semantics} and the \emph{big-step semantics} as,
respectively, the two distributions
\begin{align*}
  \sts{\termone}&=\sup \{\distrone\mid{\termone}\to_n{\distrone}\}\\
  \qquad\qquad
  \bts{\termone}&=\sup \{\distrone\mid\sibss{\termone}{n}{\distrone}\}
\end{align*}
It would be quite disappointing if the two object above were
different. Indeed, this section is devoted to proving the following
theorem:
\begin{theorem}\label{thm:eq}
  For every term $\termone$, $\sts{\termone}=\bts{\termone}$.
\end{theorem}
The following is a fact which will be quite useful in the following:
\begin{lemma}[Monotonicity]
  If $\sisss{\termone}{n}{\distrone}$, $m\geq n$ and $\sisss{\termone}{m}{\distrtwo}$,
  then $\distrtwo\geq\distrone$.
\end{lemma}
Theorem \ref{thm:eq} can be proved by showing that any big-step approximation
can itself over-approximated with small-step, and vice versa. Let us start
by showing that, essentially, the big-step rule for applications is
small-step-admissible:
\begin{lemma}\label{lemma:derapp}
  If $\sisss{\termone}{n}{\distrone}$,
  $\sisss{\termtwo}{m}{\distrtwo}$, and for all $\termthree$ and $\valone$,
  $\sisss{\sbst{\termthree}{\varone}{\valone}}{p}{\distrtwo_{\termthree,\valone}}$,
  then
  $\sisss{\termone\termtwo}{n+m+p}{\distrthree}$
  such that for all $A$
  \begin{align*}
    \distrthree(A)\geq{}&
    \restr{\distrone}{\expset}(A)+\distrone(\RR)\cdot\indfun{A}{\terror}+\distrone(\valset_\lambda)\cdot\restr{\distrtwo}{\expset}(A)\\
    &+\iint\distrtwo_{\termthree,\valone}(A)\,\restr{\distrone}{\valset_\lambda}(\abstr{\varone}{d\termthree})\,\restr{\distrtwo}{\valset}(d\valone).
  \end{align*}
\end{lemma}
\begin{proof}
First of all, one can prove that if $\sisss{\termtwo}{n}{\distrone}$
and
$\sisss{\sbst{\termthree}{\varone}{\valone}}{m}{\distrtwo_{\valone}}$ for all $\valone$
then $\sisss{(\abstr{\varone}{\termthree})\termtwo}{n+m}{\distrthree}$
where
$\distrthree(A)\geq\restr{\distrone}{\expset}(A)+\int\distrtwo_{\valone}(A)\,\restr{\distrone}{\valset}(d\valone)$ for all $A$.  This is an induction on $n$.
\begin{varitemize}
\item
  If $n=0$, then $\distrone$ is necessarily the zero distribution $A\mapsto 0$.  
  Then 
  $\distrthree(A)\ge 0=\restr{\distrone}{\expset}(A)+\int\distrtwo_{\valone}(A)\,\restr{\distrone}{\valset}(d\valone))$.
\item
  Suppose the thesis holds for $n$, and let's try to prove the thesis
  for $n+1$. We proceed by further distinguishing some subcases:
  \begin{varitemize}
  \item
    If $\termtwo$ is a value $\valtwo$, then
    $\distrone=\dirac{\valtwo}$, $\restr{\distrone}{\expset}$ is the
    zero distribution and thus
    $$
    (\abstr{\varone}{\termthree})\termtwo\redss_{m+1}(A\mapsto
    \restr{\distrone}{\expset}(A)+\int\distrtwo_{\valone}(A)\,\restr{\distrone}{\valset}(d\valone)).
    $$
    The thesis follows by monotonicity.
  \item
    If $\termtwo$ is an exception $\fail$, then $\distrone=\dirac{\fail}$,
    and since $(\abstr{\varone}{\termthree})\fail\red\fail$, we can conclude
    that, since $\restr{\distrone}{\valset}$ is the zero distribution,
    $$
    (\abstr{\varone}{\termthree})\termtwo\redss_{2}\dirac{\fail}=
    (A\mapsto\restr{\distrone}{\expset}(A)+\int\distrtwo_{\valone}(A)\,\restr{\distrone}{\valset}(d\valone)).
    $$
    The thesis again follows by monotonicity.
  \item
    If $\termtwo$ is not a generalized value, then, necessarily
    $\distrone(A)=\int\distrfour_\termfour(A)\,\distrfive(d\termfour)$,
    where $\termtwo\red\distrfive$ and $\termfour\redss_{n}\distrfour_\termfour$
    for every $\termfour$. By induction hypothesis, there are distributions
    $\distrsix_\termfour$ such that 
    $(\abstr{\varone}{\termthree})\termfour\redss_{n+m}\distrsix_\termfour$,
    and, for all $A$,
    $$
    \distrsix_\termfour(A)\geq \restr{\distrfour_\termfour}{\expset}+
      \int \distrtwo_\valone(A)\,\restr{\distrfour_\termfour}{\valset}(d\valone)
    $$
    Let now $\ectxone$ be the evaluation context $(\abstr{\varone}{\termthree})\hole$.
    Then, it holds that $(\abstr{\varone}{\termthree})\termtwo\red\ctm{\ectxone}{\distrfive}$
    and thus: 
    $$
    (\abstr{\varone}{\termthree})\termtwo\redss_{n+m+1}(A\mapsto\int\distrsix_\termfour(A)\,(\ctm{\ectxone}{\distrfive}((\abstr{\varone}{\termthree})d\termfour))).
    $$
    We can now observe that:
    \begin{align*}
      \int\distrsix_\termfour(A)&\,(\ctm{\ectxone}{\distrfive}((\abstr{\varone}{\termthree})d\termfour))=
         \int\distrsix_\termfour(A)\,\distrfive(d\termfour)\\
         &\geq\int\restr{\distrfour_\termfour}{\expset}(A)\,\distrfive(d\termfour)+
         \iint\distrtwo_\valone(A)\,\restr{\distrfour_\termfour}{\valset}(d\valone)\,\distrfive(d\termfour)\\
         &=\restr{\distrone}{\expset}(A)+
         \iint\distrtwo_\valone(A)\,\restr{\distrfour_\termfour}{\valset}(d\valone)\,\distrfive(d\termfour)\\
         &=\restr{\distrone}{\expset}(A)+
         \int\distrtwo_\valone(A)\,\restr{\distrone}{\valset}(d\valone).
    \end{align*}
  \end{varitemize}
\end{varitemize}
Then one can prove the statement of the lemma, again by induction on $n$, following
the same strategy as above.
\qed
\end{proof}
\begin{lemma}\label{lemma:sshibs}
  \mbox{If $\sibss{\termone}{n}{\distrone}$ there is $\distrtwo$ s.t.~$\sisss{\termone}{3^n}{\distrtwo}$
  and $\distrtwo\geq\distrone$.}
\end{lemma}
\begin{proof}
  By induction on $n$, we can prove that if $\sibss{\termone}{n}{\distrone}$, then
  $\sisss{\termone}{\distrtwo}$ where $\distrtwo\geq\distrone$.
  The only interesting case is when $\termone$ is an application,
  and there we simply use Lemma \ref{lemma:derapp}.\qed
\end{proof}
At this point, we already know that $\sts{\termone}\geq\bts{\termone}$.
The symmetric inequality can be proved by showing that the big-step rule for
applications can be \emph{inverted} in the small-step:
\begin{lemma}\label{lemma:invapp}
  If $\sisss{\termone\termtwo}{n+1}{\distrone}$, then
  $\sisss{\termone}{n}{\distrtwo}$, $\sisss{\termtwo}{n}{\distrthree}$ and 
  for all $\termfour$ and $\valone$, 
  $\sisss{\sbst{\termfour}{\varone}{\valone}}{n}{\distrfour_{\termfour,\valone}}$
  such that for all $A$,
  \begin{align*}
    \distrone(A)\leq{}&\restr{\distrtwo}{\expset}(A)+
    \distrtwo(\RR)\cdot\indfun{A}{\terror}+
    \distrtwo(\valset_\lambda)\cdot\restr{\distrthree}{\expset}(A)\\
    &+\iint\distrfour_{\termfour,\valone}(A)
    \,\restr{\distrtwo}{\valset_\lambda}(\abstr{\varone}{d\termfour})
    \,\restr{\distrthree}{\valset}(d\valone).
  \end{align*}
\end{lemma}
  \begin{proof}
    By induction on $n$.
    \begin{varitemize}
    \item
      If $n=0$, then $\distrone$ is the zero distribution, and so are $\distrtwo,\distrthree$ and all $\distrfour_{\termfour,\valone}$.
    \item
      Suppose the thesis holds for every natural number smaller than
      $n$ and prove it for $n$. Let us distinguish a few cases, and
      examine the most relevant ones:
      \begin{varitemize}
      \item
        If $\termone$ is an abstraction $\abstr{\varone}{\termthree}$
        and $\termtwo$ is a value $\valtwo$, then 
        $\termone\red\dirac{\sbst{\termthree}{\varone}{\valtwo}}$ and
        $\sbst{\termthree}{\varone}{\valtwo}\redss_n\distrone$.
        We can then observe that
        \begin{align*}
          \distrtwo&=\dirac{\abstr{\varone}{\termthree}}\\
          \distrthree&=\dirac{\valtwo}\\
          \distrfour_{\termfour,\valone}&=\distrone\mbox{ whenever }\termfour=\termthree\mbox{ and }\valone=\valtwo
        \end{align*}
        Just observe that
        $$
        \distrone(A)=\iint\distrfour_{\termfour,\valone}(A)\,\distrtwo(\abstr{\varone}{d\termfour})\,\distrthree(d\valone)
        $$
        and that $\restr{\distrtwo}{\expset}=\restr{\distrtwo}{\RR}=\restr{\distrthree}{\expset}=\emdistr$.
      \item
        If none of $\termone$ and $\termtwo$ are values, then
        $\termone\red\distrseven$ and thus $\termone\termtwo\red\ctm{\ectxone}{\distrseven}$
        where $\ectxone=\hole\termtwo$. Moreover,
        $\termthree\termtwo\redss_n\distrfive_\termthree$,
        where
        $$
        \distrone(A)=\int\distrfive_\termthree(A)\,\ctm{\ectxone}{\distrseven}((d\termthree)\termtwo)
           =\int\distrfive_\termthree(A)\,\distrseven(d\termthree)
        $$ 
        We apply the induction hypothesis (and monotonicity) to each of the
        $\termthree\termtwo\redss_n\distrfive_\termthree$, and we obtain
        that $\termthree\redss_{n-1}\distrsix_\termthree$,
        $\termtwo\redss_{n}\distrthree$ and 
        $\sisss{\sbst{\termfour}{\varone}{\valone}}{n}{\distrfour_{\termfour,\valone}}$,
        where
        \begin{align*}
        \distrfive_\termthree(A)\leq{}&\restr{\distrsix_\termthree}{\expset}(A)+
        \distrsix_\termthree(\RR)\cdot\indfun{A}{\terror}+
        \distrsix_\termthree(\valset_\lambda)\cdot\restr{\distrthree}{\expset}(A)\\ 
        &{}+ \iint\distrfour_{\termfour,\valone}(A)
        \,\restr{\distrsix_\termthree}{\valset_\lambda}(\abstr{\varone}{d\termfour})
        \,\restr{\distrthree}{\valset}(d\valone)
        \end{align*}
        Now let $\distrtwo$ be the distribution
        $$
        A\mapsto\int\distrsix_\termthree(A)\,\distrseven(d\termthree).
        $$
        Clearly, $\termone\redss_{n}\distrtwo$. Moreover,
        \begin{align*}
          \distrone(A)&=\int\distrfive_\termthree(A)\,\distrseven(d\termthree)\\
          &\leq\int\restr{\distrsix_\termthree}{\expset}(A)\,\distrseven(d\termthree)
             +\int\distrsix_\termthree(\RR)\cdot\indfun{A}{\terror}\,\distrseven(d\termthree)\\
          &\qquad+\int\distrsix_\termthree(\valset_\lambda)\cdot\restr{\distrthree}{\expset}(A)\,\distrseven(d\termthree)\\
          &\qquad+\iiint\distrfour_{\termfour,\valone}(A)
             \,\restr{\distrsix_\termthree}{\valset_\lambda}(\abstr{\varone}{d\termfour})
             \,\restr{\distrthree}{\valset}(d\valone)\,\distrseven(d\termthree)\\
          &=\restr{\distrtwo}{\expset}(A)+
            \distrtwo(\RR)\cdot\indfun{A}{\terror}+
            \distrtwo(\valset_\lambda)\cdot\restr{\distrthree}{\expset}(A)\\
          &\qquad+\iint\distrfour_{\termfour,\valone}(A)
             \,\restr{\distrtwo}{\valset_\lambda}(\abstr{\varone}{d\termfour})
             \,\restr{\distrthree}{\valset}(d\valone)
        \end{align*}
      \end{varitemize}
    \end{varitemize}
    \qed
  \end{proof}
%\item
\begin{lemma}\label{lemma:bshiss}
  If $\sisss{\termone}{n}{\distrone}$, then there is $\distrtwo$ such that $\sibss{\termone}{n}{\distrtwo}$
  and $\distrtwo\geq\distrone$.
\end{lemma}
\begin{proof}
  Again, this is an induction on $n$ that makes essential use, this time,
  of Lemma \ref{lemma:invapp}.\qed
\end{proof}
\begin{restate}{Theorem~\ref{thm:eq}}
  For all $\termone$, $\sts{\termone}=\bts{\termone}$.
\end{restate}
\begin{proof}
This is a consequence of Lemma~\ref{lemma:sshibs} and Lemma~\ref{lemma:bshiss}.\qed
\end{proof}

In subsequent sections we let $\tsq{\termone}$ stand for
$\sts{\termone}$ or $\bts{\termone}$.
%%%%%%%%%%%%%%%%%%%%%%%
\subsection{Geometric Distribution, Revisited}
%%%%%%%%%%%%%%%%%%%%%%%
Let's consider again the geometric distribution of Section \ref{subsec:geom}.
There is a monotonically increasing
map $f:\NN\rightarrow\NN$ such that for every $n$, it holds that
$$
\sisss{(\mathit{geometric}\;0.5)}{f(n)}{\sum_{i=0}^n\frac{1}{2^{i+1}}\dirac{i}}
$$
As a consequence, $\tsq{\mathit{geometric}\;0.5}=\sum_{i=0}^\infty\frac{1}{2^{i+1}}\dirac{i}$.
%%%%%%%%%%%%%%%%%%%%%%%%%%%%%%%%%%%%%%%%%%%%%%%%%%%%%%%%%%%%%%%%%%%%%%%%%%%
\subsection{Distribution-based and Sampling-based Semantics are Equivalent}
%%%%%%%%%%%%%%%%%%%%%%%%%%%%%%%%%%%%%%%%%%%%%%%%%%%%%%%%%%%%%%%%%%%%%%%%%%%
This section is a proof of the following theorem.
\begin{theorem}\label{thm:sampling-distribution}
  For every term $\termone$, $\sbd{\termone} = \tsq{\termone}$.
\end{theorem}
The way to prove Theorem \ref{thm:sampling-distribution} is by looking
at traces of \emph{bounded} length. For every $n\in\NN$, let
$\sampseqp{n}$ be the set of sample traces of length at most $n$,
which itself has the structure of a measure space with measurable sets
$\restr{\cal S}{\sampseqp{n}}$.  We define the result distribution $\sbdp{\termone}{n}$ as follows:
$$
\sbdp{\termone}{n}(\tsetone) = \int_{\sampseqp{n}}\termsmp{\termone}{s}\cdot\indfun{\tsetone}{\obssmp{\termone}{s}}\,ds
$$ 
%If $\dom(f)=\sampseq$ we write $f\rvert_{n}$ for the restriction of $f$ to $\sampseqp{n}$. 
The integral over \emph{all} traces can be seen
as the limit of all integrals over bounded-length traces:
\begin{lemma}\label{lemma:sup-sampseqn}
  If $f:\sampseq\to\RealInf$ is measurable then $\int\! f=\sup_{n}\int_{\sampseqp{n}}\! f$.
\end{lemma}
\begin{proof}
  Let $g_{n}(s)=f(s)\cdot[s\in \sampseqp{n}]$, so $\int_{\sampseqp{n}}\! f = \int\!g_{n}$ by definition. 
  Since the $g_n$ are converging to $f$ pointwise from below, we have 
  $\int\! f=\sup_{n}\int\! g_{n}$ by the monotone convergence theorem.\qed
\end{proof}
A corollary is that $\sbd{\termone}=\sup_{n\in\NN}\sbdp{\termone}{n}$.
\begin{lemma}
  If $\termone\red\distrone$, then
  $\tsq{\termone}=A\mapsto\int\tsq{\termtwo}(A)\;\distrone(d\termtwo)$
\end{lemma}
\begin{proof}
  For every $n$ and for every term $\termtwo$, let
  $\distrtwo^n_\termtwo$ be the unique value distribution such that
  $\termtwo\redss_n\distrtwo^n_\termtwo$. By definition, we have that
  $$
  \distrtwo^{n+1}_\termone(A)=\int\distrtwo^{n}_\termtwo(A)\;\distrone(d\termtwo).
  $$
  By the monotone convergence theorem, then,
  \begin{align*}
    \tsq{\termone}(A)&=\sup_n\distrtwo^{n+1}_\termone(A)=\sup_n\int\distrtwo^{n}_\termtwo(A)\;\distrone(d\termtwo)\\
       &=\int(\sup_n\distrtwo^{n}_\termtwo(A))\;\distrone(d\termtwo)=\int\tsq{\termtwo}(A)\;\distrone(d\termtwo).
  \end{align*}
  \qed
\end{proof}
The following is a useful technical lemma.
\begin{lemma}\label{lemma:commutetrace}
  $\sbdp{E[\distone(\vec{c})]}{n+1}(A)=\int\sbdp{\termtwo}{n}(A)\;\ctm{E}{\mu_\distone(\vec{c})}(d\termtwo)$.
\end{lemma}
\begin{lemma}
  If $\termone\red\distrone$, then
  $\sbd{\termone}=A\mapsto\int\sbd{\termtwo}(A)\;\distrone(d\termtwo)$
\end{lemma}
A program $\termone$ is said to \emph{deterministically diverge} iff
$(\termone,1,s)\Rightarrow(\termtwo,w,[])$ implies that $w=1$, $s=[]$, and
$\termtwo$ is \emph{not} a generalized value.   Terms that deterministically
diverge have very predictable semantics, both distribution- and
sampling-based.  
\begin{lemma}\label{lemma:detdiverg}
  If $\termone$ deterministically diverges then
  $\tsq{\termone}=\sbd{\termone}=\emdistr$.
\end{lemma}
\begin{proof}
  One can easily prove, by induction on $n$, that if $\termone$
  deterministically diverges, then $\sisss{\termone}{n}{\emdistr}$:
  \begin{varitemize}
  \item
    If $n=0$, then $\sisss{\termone}{n}{\emdistr}$ by definition.
  \item
    About the inductive case, since $\termone$ cannot be a generalized value, it must
    be that $\termone\red\dirac{\termtwo}$ (where $\termtwo$
    deterministically diverges) and that
    $\sisss{\termone}{n+1}{\distrone}$, where
    $\sisss{\termtwo}{n}{\distrone}$.
    By induction hypothesis, $\distrone=\emdistr$.
  \end{varitemize}
  The fact that $\sbd{\termone}=\emdistr$ is even
  simpler to prove, since if $\termone$ deterministically diverges,
  then there cannot be any $s,w,\valone$ such that
  $\termone\Downarrow^s_w\valone$, and thus
  $\termsmp{\termone}{s}$ is necessarily $0$.
  \qed
\end{proof}

A program $\termone$ is
said to \emph{deterministically converge to a program~$\termtwo$} iff
$(\termone,1,[])\Rightarrow(\termtwo,1,[])$. 
Any term that deterministically converges
to another term has the same semantics as the latter.
\begin{lemma}\label{lemma:detconv}
  Let $\termone$ deterministically converge to $\termtwo$.
  Then:
  \begin{varitemize}
  \item
    $\distrone\leq\distrtwo$ whenever $\termone\redss_{n}\distrone$;
    and $\termtwo\redss_n\distrtwo$;
  \item
    $\sbdp{\termone}{n}=\sbdp{\termtwo}{n}$;
  \item
    $\tsq{\termone}=\tsq{\termtwo}$ and $\sbd{\termone}=\sbd{\termtwo}$
  \end{varitemize}
\end{lemma}
\begin{proof}
  The first point is an induction on the structure of the
  proof that $\termone$ deterministically converge to $\termtwo$.
  Let us consider the second and third points.
  Since equality is transitive, we can assume, without
  losing any generality, that $(\termone,1,[])\rightarrow(\termtwo,1,[])$,
  namely that $\termone\detred\termtwo$. With the latter hypothesis,
  it is easy to realize that $\termone\Downarrow^w_s\valone$ iff
  $\termtwo\Downarrow^w_s\valone$ and that
  $\termone\redss_{n+1}\distrone$ iff $\termtwo\redss_n\distrone$.
  The thesis easily follows.
  \qed
\end{proof}
\begin{lemma}\label{lemma:values}
  For every generalized value $\gvalone$, it holds that
  $\tsq{\gvalone}=\sbd{\gvalone}=\dirac{\gvalone}$.
\end{lemma}
If a term does not diverge deterministically, then 
it converges either to a generalized value or to a term
that performs a sampling.
\begin{lemma}\label{lemma:tracecases}
  For every program $\termone$, exactly one of the following
  conditions holds:
  \begin{varitemize}
  \item
    $\termone$ deterministically diverges;
  \item
    There is generalized value $\gvalone$ such that
    $\termone$ deterministically converges to $\gvalone$
  \item
    There are $\ectxone,\distone,c_1,\ldots,c_{|\distone|}$ such that
    $\termone$ deterministically converges to $E[\distone(c_1,\ldots,c_{|\distone|})]$.
  \end{varitemize}
\end{lemma}
\begin{proof}
  Easy.
  \qed
\end{proof}
We are finally ready to give the two main lemmas that lead
to a proof of Theorem \ref{thm:sampling-distribution}. The first
one tells us that any distribution-based approximation is smaller
than the sampling based semantics:
\begin{lemma} \label{lemma:dist-leq-sample}
  If $\termone\redss_{n}\distrone$, then $\distrone\leq\sbd{\termone}$.
\end{lemma}
\begin{proof}
  By induction on $n$:
  \begin{varitemize}
  \item
    If $n=0$, then $\distrone$ is necessarily $\emdistr$, and we are done.
  \item
    About the inductive case, let's distinguish three cases depending on
    the three cases of Lemma \ref{lemma:tracecases}, applied to $\termone$:
    \begin{varitemize}
    \item
      If $\termone$ deterministically diverges, then by Lemma \ref{lemma:detdiverg},
      $\distrone\leq\tsq{\termone}=\sbd{\termone}$.
    \item
      If $\termone$ deterministically converges to a generalized value $\gvalone$, then
      by Lemma \ref{lemma:detconv} and Lemma \ref{lemma:values}, it holds that
      $$
      \distrone\leq\tsq{\termone}=\tsq{\gvalone}=\dirac{\gvalone}=\sbd{\gvalone}=\sbd{\termone}.
      $$
    \item
      If $\termone$ deterministically converges to $E[\distone(\vec{c})]$,
      let $\distrtwo$ be such that $E[\distone(\vec{c})]\redss_{n+1}\distrtwo$.
      By Lemma \ref{lemma:detconv} and Lemma \ref{lemma:commutetrace} we have, by
      induction hypothesis, that
      \begin{align*}
        \distrone(A)\leq\distrtwo(A)&=\int\distrthree_\termtwo(A)\;\ctm{E}{\mu_\distone(\vec{c})}(d\termtwo)\\
        &\leq\int\sbd{\termtwo}(A)\;\ctm{E}{\mu_\distone(\vec{c})}(d\termtwo)\\
        &=\sbd{\termone}
      \end{align*}
      where $\termtwo\redss_n\distrthree_\termtwo$.\qed
    \end{varitemize}
  \end{varitemize}
\end{proof}
The second main lemma tells us that if we limit our attention to traces
of length at most $n$, then we stay below distribution-based semantics:
\begin{lemma}\label{lemma:sample-leq-dist}
  For every $n\in\NN$, $\sbdp{\termone}{n}\leq\tsq{\termone}$.
\end{lemma}
\begin{proof}
  By induction on $n$:
  \begin{varitemize}
  \item
    In the base case, then let us distinguish three cases depending on
    the three cases of Lemma \ref{lemma:tracecases}, applied to $\termone$:
    \begin{varitemize}
    \item
      If $\termone$ deterministically diverges, then by Lemma \ref{lemma:detdiverg},
      $\sbdp{\termone}{0}\leq\sbd{\termone}=\tsq{\termone}$.
    \item
      If $\termone$ deterministically converges to a generalized value $\gvalone$, then
      by Lemma \ref{lemma:detconv} and Lemma \ref{lemma:values}, it holds that
      $$
      \sbdp{\termone}{0}=\sbdp{\gvalone}{0}\leq\sbd{\gvalone}=\dirac{\gvalone}=\tsq{\gvalone}=\tsq{\termone}.
      $$
    \item
      If $\termone$ deterministically converges to $E[\distone(\vec{c})]$,
      then $\sbdp{\termone}{0}=\sbdp{E[\distone(\vec{c})]}{0}=\emdistr\leq\tsq{\termone}$.
    \end{varitemize}
  \item
    About the inductive case, let us again distinguish three cases depending on
    the three cases of Lemma \ref{lemma:tracecases}, applied to $\termone$:
     \begin{varitemize}
    \item
      If $\termone$ deterministically diverges, then by Lemma \ref{lemma:detdiverg},
      $\sbdp{\termone}{n+1}\leq\sbd{\termone}=\tsq{\termone}$.
    \item
      If $\termone$ deterministically converges to a generalized value $\gvalone$, then
      by Lemma \ref{lemma:detconv} and Lemma \ref{lemma:values}, it holds that
      $$
      \sbdp{\termone}{n+1}=\sbdp{\gvalone}{n+1}\leq\sbd{\gvalone}=\dirac{\gvalone}=\tsq{\gvalone}=\tsq{\termone}.
      $$      
    \item
      If $\termone$ deterministically converges to $E[\distone(\vec{c})]$,      
      by Lemma \ref{lemma:detconv} and Lemma \ref{lemma:commutetrace} we have, by
      induction hypothesis, that
      \begin{align*}
        \sbdp{\termone}{n+1}(A)&=\sbdp{E[\distone(\vec{c})]}{n+1}(A)\\
        &=\int\sbdp{\termtwo}{n}(A)\;\ctm{E}{\mu_\distone(\vec{c})}(d\termtwo)\\
        &\leq\int\tsq{\termtwo}(A)\;\ctm{E}{\mu_\distone(\vec{c})}(d\termtwo)\\
        &=\tsq{\termone}.
      \end{align*}
     \end{varitemize}
  \end{varitemize}\qed
 \end{proof}

\begin{comment}
  \begin{lemma} \label{lemma:traces-values-equiv} For any program $M$,
    $\sbd{M} = \Termsmp M \mathbf{O}^{-1}_M$
  \end{lemma}
  \begin{proof}
    \begin{eqnarray*}
      \Termsmp M \mathbf{O}^{-1}_M (A) &=& \int \Termsmp{M}(s) \mathbf{I}_{\mathbf{O}^{-1}_M(A)}(s)\ \mu(ds) \\
      &=& \int \Termsmp{M}(s) \mathbf{I}_A (\mathbf{O}_M(s))\ \mu(ds) \\
      &=& \sbd{M}
    \end{eqnarray*}
    \qed
  \end{proof}
\end{comment}

\begin{restate}{Theorem~\ref{thm:sampling-distribution}}
  $\sbd{\termone}=\tsq{\termone}$.
\end{restate}
\begin{proof}\belowdisplayskip=-9pt
  \begin{align*}
    \sts{\termone} 
    &= \sup_{n\in\NN}\{\distrone\mid M\to_n\distrone\} &
    \text{(by definition)}\\
    & \leq  \sbd{\termone}&
    \text{(by Lemma~\ref{lemma:dist-leq-sample})} \\
    & = \sup_{n\in\NN}\sbdp{\termone}{n}& 
    \text{(by Lemma~\ref{lemma:sup-sampseqn})} \\ 
& \leq  \tsq{\termone} &
    \text{(by Lemma~\ref{lemma:sample-leq-dist})}\\
&=    \sts{\termone} & \text{(by Theorem~\ref{thm:eq})}  
  \end{align*} \qed
\end{proof}
\begin{corollary}\label{cor:tracedist-sub-probability}
The measures $\Tracedist{\termone}$ and $\Tracevaldist\termone$ are sub-probability distributions.
\end{corollary}
%%%%%%%%%%%%%%%%%%%%%%%%%%%%%%%%%%%%%%%%%%%%%%%%%%%%%%%%%%%
\subsection{An Application of the Distribution-Based Semantics}\label{sec:app}
%%%%%%%%%%%%%%%%%%%%%%%%%%%%%%%%%%%%%%%%%%%%%%%%%%%%%%%%%%%
It is routine to show that for all values $\valone$,
it holds that \linebreak $\restr{\tsq{\score{\valone}}}{\valset}=\restr{\tsq{\termone\valone}}{\valset}$,
where $\termone$ is the term
\begin{align*}
\lambda x.&\mathtt{if}\;(0<x)\land (x\leq 1)\\
  &\mathtt{then}\;(\ite{\mathit{flip}(x)}{\texttt{true}}{\fail})\\
  &\mathtt{else}\;\fail
\end{align*}
This shows that even though $\score{\valone}$ and $\termone\valone$
\emph{do not} have the same sampling-based semantics, they can be used
interchangeably whenever only their extensional, distribution-based
behaviour on values is important. We use the equation to our advantage by
encoding soft constraints with \texttt{score} instead of \texttt{flip}
(as discussed in Section~\ref{sec:soft-conditioning}), as the fewer the
nuisance parameters the better for inference.

Whenever the stronger equation 
$\tsq{\score{\valone}}=\tsq{\termone\valone}$ is needed, we could
replace $\termone$ by the following one:
\begin{align*}
\lambda x.&\mathtt{if}\;(0<x)\land (x\leq 1)\\
  &\mathtt{then}\;(\ite{\mathit{flip}(x)}{\texttt{true}}{\Omega})\\
  &\mathtt{else}\;\fail
\end{align*}
%%%%%%%%%%%%%%%%%%%%%%%%%%%%%%%%%%%%%%%%%%%%%%%%%%%%%%%%%%%
\subsection{Rejection Sampling} \label{sec:rejection-sampling}
%%%%%%%%%%%%%%%%%%%%%%%%%%%%%%%%%%%%%%%%%%%%%%%%%%%%%%%%%%%
%
The same (normalized) distribution on successful runs would be
obtained by re-evaluating the entire program from the beginning
whenever the Boolean predicate fails, as in:
\begin{center}
\begin{minipage}{.4\textwidth}
\begin{flushleft}
\begin{prog}\fixop\ f. \lambda \!\star\!. \mathtt{let} \dots \mathtt{in}\\
 \quad \ite{b}{n}{f\ 0}\\
f\ 0
\end{prog}
\end{flushleft}
\end{minipage}
\end{center}
This corresponds to a basic inference algorithm known as
\emph{rejection sampling}.

%%%%%%%%%%%%%%%%%%%%%%%%%%%%%%%%%%%%%%%%%%%%%%%%%%%%%%%%%%%
\subsection{Motivation for 1-Bounded Scores}\label{sec:motivation}
%%%%%%%%%%%%%%%%%%%%%%%%%%%%%%%%%%%%%%%%%%%%%%%%%%%%%%%%%%%
Recall that we only consider $\score{c}$ for $c\in(0,1]$.  Admitting $\score{2}$ (say), 
we exhibit an anomaly by constructing a recursive program that intuitively terminates with probability 1, 
but where the expected value of its score is infinite. Let 
\begin{SHORT}
\begin{multline*}
\mathit{inflate} := {}\fixop\ f\ \lambda x. \\ 
                    \ite{\mathit{flip}(0.5)}{\score 2;(f\ x)}{x}.
      \end{multline*}
\end{SHORT}
\begin{LONG}
\[
\mathit{inflate} := {}\fixop\ f\ \lambda x. \ite{\mathit{flip}(0.5)}{\score 2;(f\ x)}{x}.
\]
\end{LONG}
% $\mathit{inflate} := \fixop\ f\ \lambda x. M_{\texttt{score}} $ where 
% \[M_{\texttt{score}} := \ite{\mathit{flip}(0.5)}{\score 2;(f\ x)}{x}.\]
Since $\tsq{\score2}=\tsq{\fail}$ we have $\tsq{\mathit{inflate}\ V}=0.5\cdot\dirac{V} + 0.5\cdot\dirac{\mathtt{fail}}$ for any $V$.
However, evaluating $\mathit{inflate}\ V$ in a version of our trace semantics where the argument to \texttt{score} may be 2 yields \[\Tracedist{\mathit{inflate}\ V}(\sampseq_{n})=\sum_{k=1}^{n}1/2=n/2\]
and so there $\sbd{\mathit{inflate}\ V}(A)=\infty$ if $V\in A$, otherwise 0.

More strikingly, in the modified semantics we would also have
$\sbd{\mathit{inflate}\ \textsf{Gaussian}(0,1)}([q,r])=\infty$ for all
real numbers $q<r$. These examples show that even statically bounded
scores in combination with recursion may yield return value measures
that are not even $\sigma$-finite, causing many standard results in
measure theory not to apply.  For this reason, we restrict attention
to positive scores bounded by one.
An alternative approach would be to admit unbounded scores,
and restrict attention to those programs for which $\sbd{\termone}(\valset)<\infty$.
% The results on the trace semantics, as well as the inference algorithm described in Section~\ref{sec:inference} below, generalise nicely to this setting, 
% but our proofs for the distributional semantics do not.
%\end{LONG}

%%%%%%%%%%%%%%%%%%%
\section{Inference}\label{sec:inference}
%%%%%%%%%%%%%%%%%%%
In this section, we present a variant of the Metropolis-Hastings (MH) 
algorithm~\cite{metropolis53,hastings70:MH} for sampling the return 
values of a particular closed term $M \in \cterms$. 
This algorithm yields consecutive samples from a Markov
chain over $\sampseq$, such that the density of the samples $s$ converges
to $\Termsmv M (s)$ up to normalization. We can then apply the function
$\mathbf{O}_M$ to obtain the return value of $M$ for a given trace.

We prove correctness of this algorithm by showing that as the
number of samples goes to infinity, the distribution of the samples
approaches the distributional semantics of the program.
%%%%%%%%%%%%%%%%%%%%%%%%%%%%%%%%%%%%%%%%%%%%%%%%%%%%%
\subsection{A Metropolis-Hastings Sampling Algorithm} 
%%%%%%%%%%%%%%%%%%%%%%%%%%%%%%%%%%%%%%%%%%%%%%%%%%%%%
We begin by outlining a generic Metropolis-Hastings algorithm for
probabilistic programs, parametric in a proposal density function
$q(s,t)$. The algorithm consists of three steps:
\begin{varenumerate}
\item Pick an initial state $s$ with $\Termsmv M(s)\neq 0$
  (e.g., by running $M$).
\item\label{item:MHloop} 
  Draw the next state $t$ at random with probability density $q(s,t)$.
\item Compute $\alpha$ as below.
  \begin{equation}
    \alpha=\min\left(1,\frac{\Termsmv M(t)}{\Termsmv M(s)}\cdot
    \frac{q(t,s)}{q(s,t)}\right)\label{eq:alpha}
  \end{equation}
 \begin{varitemize}
 \item With probability $\alpha$, output $t$ and repeat from \ref{item:MHloop} with $s:=t$.
 \item Otherwise, output $s$ and repeat from \ref{item:MHloop} with $s$ unchanged.
 \end{varitemize}
\end{varenumerate}
The formula used for the number $\alpha$ above is often called the
Hastings \emph{acceptance probability}.  
Different probabilistic programming language implementations use
different choices for the density~$q$ above, based on pragmatics.
The trivial choice would be to let $q(s,t)=\Termsmv M(t)$ for all $s$,
which always yields $\alpha=1$ and so is equivalent to rejection sampling.
We here define another simple density function~$q$ (based
on~\citet{corsamp}), giving emphasis to the conditions that it needs
to satisfy in order to prove the convergence of the Markov chain given
by the Metropolis-Hastings algorithm
(Theorem~\ref{thm:sampling-trace}).

\subsection{Proposal Density}
In the following, let $M$ be a fixed program. 
Given a trace $s=[c_{1},\dots,c_{n}]$, we write $s_{i..j}$ for the trace $[c_i, \dots, c_j]$ when $1\le i \le j \le n$.  
Intuitively, the  following procedure describes how to obtain the
proposal kernel density ($q$ above):
\begin{varenumerate}
\item Given a trace $s$ of length $n$, let $t=[t_{1},\dots,t_{n}]$ 
  where each $t_{i}$ is drawn independently 
  from a normal distribution with mean $s_{i}$ and variance $\sigma^{2}$, and
  let $p_{i}$ be the probability density of $t_{i}$.
\item Let $k\le n$ be the largest number such that $(M,1,t_{1..k})\Rightarrow (M',w,[])$.
There are three cases:
  \begin{varitemize}
  \item If $k=n$, run $M'\Downarrow_{t'}^{w'}V$, and let $q(s,t@t')=p_{1}\dots p_{n}w'$.
  \item If $k<n$ and $M'\Downarrow_{[]}^{1}V$, let $q(s,t_{1..k})=p_{1}\dots p_{k}$.
  \item Otherwise, let $q(s,t_{1..k})=0$ and propose the trace $[]$.
  \end{varitemize}
\end{varenumerate}

%To define this kernel formally, 
To define this density formally, 
we first give a function that partially evaluates $M$ given a trace.
Let $\mathtt{peval}$ be a function taking a closed term $M$ and trace $s$
and returning the closed term $M'$ obtained after applying just as many
reduction steps to $M$ as required to use up the entire trace $s$
%
%performing the minimum number
%of reduction steps required to 
%
(or $\fail$ if this cannot be done).
% from closed terms and traces to 
%closed terms, which returns the expressions obtained 
\[
  \mathtt{peval}(M,s) = 
  \begin{cases}
    M & \text{if~} s=[]\\
    M' & \text{if}\ (M,1,s)\Rightarrow(M_{k},w_k,s_{k})\to(M',w',[])\\
    & \text{~for some~} M_{k},w_k,s_{k}, w' \text{~such that~}s_k\neq []\\
    \mathtt{fail} & \text{otherwise}
  \end{cases}
\]
\begin{lemma} \label{lemma:peval-assoc-step}
  For every $M \in C\Lambda, c \in \mathbb{R}, s \in \mathbb{S}$,
  $\mathtt{peval}(\mathtt{peval}(M,[c]), s) =
\mathtt{peval}(M, c \mathrel{::} s)$.
\end{lemma}
\begin{proof}
By splitting the equality into two inequalities, substituting the alternative definition of $\mathtt{peval}$ (described in the appendix)
and using Scott induction.\qed
\end{proof}
\begin{lemma} \label{lemma:peval-assoc}
For every $M \in C\Lambda, c \in \mathbb{R}, s \in \mathbb{S}$,
  $\mathtt{peval}(\mathtt{peval}(M,s), t) = \mathtt{peval}(M,s @ t) $.
\end{lemma}
\begin{proof}
By induction on $|s|$, with appeal to Lemma \ref{lemma:peval-assoc-step}.\qed
\end{proof}
%(TODO: define
%well-formedness conditions?)

%We can now define the transition density $q(s,t)$ (for $t \neq []$)
%as follows:
%\begin{display}[0.2]{Transition Density $q(s,t)$ for $t \neq []$ for Program $M$}
%\clause{q(s,t) = (\Pi_{i=1}^k\pdf{\mathsf{Gaussian}}(s_i,\sigma^2,t_i))\cdot
%  \termsmv{N}{t_{k+1..\Abs{t}}}\text{~if~}\Abs{t} \neq 0}\\[.5ex]
% \subclause{\text{where~}k=\min\Set{\Abs{s},\Abs{t}}\text{~and~}N=\texttt{peval}(M,t_{1..k})}
%\end{display}
%%%%%%%%%%%%%%%%%%%%%%%%%%%%%%%%%%%%%%%%%%%%%%%%%%
\subsection{A Metropolis-Hastings Proposal Kernel}
%%%%%%%%%%%%%%%%%%%%%%%%%%%%%%%%%%%%%%%%%%%%%%%%%%
We define the transition kernel $Q(s,A)$ of the Markov chain constructed
by the algorithm by integrating a density $q(s,t)$ (as a function of
$t$) over $A$ with respect to the stock measure $\mu$ on program
traces. For technical reasons, we need to ensure that $Q$ is a probability
kernel, i.e., that $Q(s,\sampseq)=1$ for all $s$.  We normalize
$q(s,\cdot)$ by giving non-zero probability $q(s,[])$ to transitions
ending in $[]$ (which is not a completed trace of $M$ by assumption).
All this is in Figure~\ref{fig:trandensitykernel}.
\begin{figure}
\begin{center}
\fbox{
\begin{minipage}{0.442\textwidth}
\begin{align*}
q(s,t)=\;&(\Pi_{i=1}^k\pdf{\mathsf{Gaussian}}(s_i,\sigma^2,t_i))\cdot\termsmv{N}{t_{k+1..\Abs{t}}}\\
&\mbox{if~}\Abs{t} \neq 0,\mbox{~where~}k=\min\Set{\Abs{s},\Abs{t}}\\
&\mbox{and~}N=\texttt{peval}(M,t_{1..k})\\
q(s,[])=\;&1-\int_A q(s, t)\, dt,
\mbox{~where~}A=\Set{t\mid \Abs{t} \neq 0}\\
Q(s,A)=\;&\int_A q(s, t)\, dt
\end{align*}
\end{minipage}}
\condnocr
\caption{Proposal Density $q(s,t)$ and Kernel $Q(s,A)$ for Program $M$}\label{fig:trandensitykernel}
\end{center}
\end{figure}
%\begin{display}[0.2]{Transition Density $q(s,[])$ and Kernel $Q(s,A)$ for Program $M$:}
%\clause{q(s,t) = (\Pi_{i=1}^k\pdf{\mathsf{Gaussian}}(s_i,\sigma^2,t_i))\cdot
%  \termsmv{N}{t_{k+1..\Abs{t}}}\text{~if~}\Abs{t} \neq 0}\\[.5ex]
% \subclause{\text{where~}k=\min\Set{\Abs{s},\Abs{t}}\text{~and~}N=\texttt{peval}(M,t_{1..k})} \\[.5ex]
%\clause{q(s,[]) = 1-\int_A q(s, t)\, dt\text{~where~}A=\Set{t\mid \Abs{t} \neq 0}} \\[.5ex]
% \clause{q(s,t) = (\Pi_{i=1}^k\pdf{\mathsf{Gaussian}}(s_i,\sigma^2,t_i))\cdot
%   \termsmv{\texttt{peval}(M,t_{1..k})}{t_{(k+1)},\dots,t_{\Abs{t}}}}\\[.5ex]
% \subclause{\text{where~}k=\max\Set{k\mid \texttt{peval}(M,t_{1..k}) \neq\mathtt{fail}}} \\[2ex]
%\clause{Q(s,A) = \int_A q(s, t)\, dt}
%\end{display}

The integral $\int_A q(s,t) dt$ is well-defined if and only if $q(s,
\cdot)$ is non-negative and measurable for every $s$. In order to show
that this property is satisfied, we first need to prove that the
$\mathtt{peval}$ function, used in the definition of $q$, is
measurable:

%We prove the following lemmas in \Appref{section:proof-of-measurability}.
\begin{lemma} \label{lemma:peval-measurable}
$\mathtt{peval}$ is a measurable function $\cterms  \times  \sampseq \rightarrow  \cterms$.
\end{lemma}

Using this result, we can show that $q$, as a function defined on pairs
of traces, is measurable.

\begin{lemma} \label{lemma:q-measurable}
For any closed program $M$, the transition density $q(\cdot, \cdot) : 
(\sampseq \times \sampseq)
\rightarrow \mathbb{R}_{+}$ is measurable.
\end{lemma}

By a well-known result in measure theory \cite[Theorem
  18.1]{billingsley95}, it follows that $q(s, \cdot)$ is measurable
for every $s \in \mathbb{S}$. To define the transition kernel for the
algorithm in terms of the proposal kernel $Q$, we need to show that
$Q$ is a probability kernel.

\begin{lemma} \label{lemma:q-kernel}
The function $Q$ is a probability kernel on $(\sampseq, \mathcal{S})$.
\end{lemma}

The proofs of lemmas \ref{lemma:peval-measurable},
\ref{lemma:q-measurable} and \ref{lemma:q-kernel} can be found in the
long version of this paper~\cite{mhlambda-arxiv}.
%%%%%%%%%%%%%%%%%%%%%%%%%%%%%%%%%%%%%%%%%%%%%%%%%%
\subsection{Transition Kernel of the Markov Chain}
%%%%%%%%%%%%%%%%%%%%%%%%%%%%%%%%%%%%%%%%%%%%%%%%%%
We now use the proposal kernel $Q$ to construct the transition kernel
of the Markov chain induced by the algorithm. To avoid trivial cases,
we assume that $M$ has positive success probability and does not
behave deterministically, i.e., that $\tsq M(\valset) > 0$ and
$\Tracedist{\termone}(\{[]\})=0$.

Hastings' Acceptance Probability $\alpha$ is defined as in Equation~(\ref{eq:alpha}) on page~\pageref{eq:alpha},
  % \[\alpha(s,t) =
  % \min\left\{1, \frac{\Termsmv M(t) q(t,s)}{\Termsmv M(s) q(s,t)} \right\}\]
  where we let $\alpha(s,t)=0$ if $\Termsmv M(t)=0$ and otherwise
  $\alpha(s,t)=1$ if $\Termsmv M(s) \cdot q(s,t)=0$.
Given the proposal transition kernel $Q$ and the
acceptance ratio~$\alpha$,  the Metropolis-Hastings algorithm yields a
Markov chain over traces with the following transition probability kernel.

\begin{equation} \label{equation:transition-kernel}
P(s, A) = \int_A \alpha(s,t)\, Q(s, dt) + 
 \indfun As \cdot \int (1 - \alpha(s,t))\,Q(s, dt).
\end{equation}

Define $P^n(s, A)$ to be the probability of the $n$:th element of the chain with transition kernel $P$ starting at $s$ being in $A$:
\begin{eqnarray*}
P^0(s,A) &=& \indfun As\\
P^{n+1}(s,A) &=& \int P(t,A) P^n(s,dt)
\end{eqnarray*}
% If $f:\setone\to\RealInf$ we let $\supp{f}$ be the \emph{support} of $f$, that is, $\Set{x\in\setone \mid f(x) \neq 0}$.

\begin{lemma}\label{lem:supp-stationary}
  If $s_0\in\mathbf{O}^{-1}_M(\valset)$ then $P^{n}(s_0,\mathbf{O}^{-1}_M(\valset))=1$.
\end{lemma}
\begin{proof}
  By induction on $n$. 
    The base case holds, since $s_0\in\mathbf{O}^{-1}_M(\valset)$ by
    assumption.  For the induction case, we have
    \(P^{n+1}(s_0,\mathbf{O}^{-1}_M(\valset)) = \int P(s,\mathbf{O}^{-1}_M(\valset))
    P^n(s_0,ds)\).  If $s\in \mathbf{O}^{-1}_M(\valset)$ we have
    \begin{align*}
      P(s,\mathbf{O}^{-1}_M(\valset)) 
      &= \int_{\mathbf{O}^{-1}_M(\valset)} \alpha(s,t)\, Q(s, dt) +  \int (1 - \alpha(s,t))\,Q(s, dt)\\
      &= \int_{\mathbf{O}^{-1}_M(\valset)} q(s,t)\, dt + (1 - \alpha(s,[]))q(s, [])\\
      &= \int_{\mathbf{O}^{-1}_M(\valset)} q(s,t)\, dt + (1 - \alpha(s,[]))(1-\int_{\mathbf{O}^{-1}_M(\valset)} q(s,t)\, dt)\\
      &= 1 - \alpha(s,[])(1-\int_{\mathbf{O}^{-1}_M(\valset)} q(s,t)\, dt)
    \end{align*}
    where $\alpha(s,[])=0 $ since $\Termsmv M([])=0$ by
    assumption. Then
    \begin{align*}
      \int P(s,\mathbf{O}^{-1}_M(\valset)) P^n(s_0,ds) 
      &= 
        \int_{\mathbf{O}^{-1}_M(\valset)} P(s,\mathbf{O}^{-1}_M(\valset)) P^n(s_0,ds)\\
      &=
        \int_{\mathbf{O}^{-1}_M(\valset)} 1 P^n(s_0,ds)\\
      &= 1
    \end{align*}
    where the first and the third equality follow from the induction
    hypothesis.
  \qed
\end{proof}
\begin{lemma}\label{lem:supp-absorbing}
  There is $0\le c<1$ such that $P^{n}([],\mathbf{O}^{-1}_M(\valset))=1-c^{n}$ and $P^n([],\Set{[]})=c^n$.
\end{lemma}
\begin{proof} \belowdisplayskip=-9pt
  Let $c=1-\Tracevaldist M(\sampseq\setminus\Set{[]})$. 
  By assumption $[]\not\in \mathbf{O}^{-1}_M(\valset)$ and $c<1$, 
  and since $\Tracevaldist M$ is a sub-probability distribution we have $0\le c$.
  We proceed by induction on $n$. The base case is trivial. 
  For the induction case, we have $P(s,\sampseq\setminus\Set{[]})=1$ for all $s\in \mathbf{O}^{-1}_M(\valset)$.
Finally
\begin{SHORT}
\begin{multline*}
  P([],\mathbf{O}^{-1}_M(\valset)) =\int_{\mathbf{O}^{-1}_M(\valset)}\Termsmv M
  =\\\Tracevaldist M(\mathbf{O}^{-1}_M(\valset)) =\Tracevaldist
  M(\sampseq\setminus\Set{[]}).
\end{multline*}
\end{SHORT}
\begin{LONG}
\[
  P([],\mathbf{O}^{-1}_M(\valset)) =\int_{\mathbf{O}^{-1}_M(\valset)}\Termsmv M
  =\Tracevaldist M(\mathbf{O}^{-1}_M(\valset)) =\Tracevaldist
  M(\sampseq\setminus\Set{[]}).
\]
\end{LONG}
\qed
\end{proof}

Based on Lemma~\ref{lem:supp-stationary} and \ref{lem:supp-absorbing}, 
we below consider the Markov chain with kernel $P$ restricted to $\mathbf{O}^{-1}_M(\valset)\cup\Set{[]}$. 

\subsection{Correctness of Inference}
By saying that the inference algorithm is correct, we mean that as the 
number of steps goes to infinity, the distribution of generated samples
approaches the distribution specified by the sampling-based semantics of
the program. 

Formally, we define $T^n(s,A) = P^n(s, \mathbf{O}_{M}^{-1}(A))$ 
as the value sample distribution at step $n$ of the Metropolis-Hastings Markov chain. 
For two measures defined on the same measurable space $(X, \mAlg)$,
we also define the variation norm $||\mu_1 - \mu_2||$ as:
\begin{equation*}
||\mu_1 - \mu_2|| = \sup_{A \in \mAlg}|\mu_1(A) - \mu_2(A)|
\end{equation*}

%For any probability measure $\nu$ on $\mathcal{E}$, define 

%Let us write $P$ for the MH transition kernel for program $M$, as defined in the previous section.

We want to prove the following theorem:

\begin{theorem}[Correctness] \label{thm:sampling-trace}
For every trace $s$ with $\termsmv Ms \neq 0$,
\[
 \lim_{n \rightarrow \infty}
 || T^n(s, \cdot) - \valdist{M} || = 0.
\]
\end{theorem}

To do so, we first need to investigate the convergence of $P^{n}$ 
to our target distribution $\pi$, defined as follows:
$$
\pi(A) = {\Tracevaldist M(A)}/{\Tracevaldist M (\sampseq)}.
$$
%\begin{display}{Target Distribution $\pi$}
%  \clause{p(A) = \int_A \Termsmv M (s)\, \mu(ds)}\\
%  \clause{\pi(A) = {\Tracevaldist M(A)}/{\Tracevaldist M (\sampseq)}}
%\end{display}

We use a sequence of known results for Metropolis-Hastings Markov
chains~\cite{tierney1994} to prove that $P^{n}$ converges to $\pi$.
We say that a Markov chain transition kernel $P$ is $\distrone$-\emph{irreducible} 
if $\distrone$ is a non-zero %$\sigma$-finite 
sub-probability distribution on $(\sampseq,\mathcal{S})$, 
and for all $x \in \sampseq,A \in \mathcal{S}$ 
there exists an integer $n>0$ such that $\distrone(A) > 0$ implies $P^n(x,A) > 0$.  
We say that $P$ is $\distrone$-\emph{aperiodic} if there do not exist $d \geq 2$ and disjoint $B_1, \dots,
B_d$ such that $\distrone(B_1) > 0$, 
and $x \in B_d$ implies  $P(x,B_1) = 1$,
and $x \in B_i$ implies that $P(x,B_{i+1}) = 1$ for $i \in \Set{1,\dots,d-1}$.
\begin{comment}
%
  We say that $\distrone$ is a stationary distribution for $P$ if for all
  measurable $A$,
  \[\int P(x,A)\,\distrone (dx)=\distrone (A).\]
\end{comment}

\begin{lemma}[\citet{tierney1994}, Theorem 1 and Corollary 2]\label{lem:tierneyConvergence}
  Let $K$ be the transition kernel of a Markov chain given by the
  Metropolis-Hastings algorithm with target distribution $\distrone$.
  If $K$ is $\distrone$-irreducible and aperiodic, then
  % \begin{enumerate}
  % \item $\pi$ is a stationary distribution for $P$; and
  % \item 
for all $s$,  $ \lim_{n \rightarrow \infty}|| K^n(s, \cdot) - \distrone || = 0$.
%  \end{enumerate}
\end{lemma}

\begin{lemma}[Strong Irreducibility] \label{lem:strong-irred}
    If $\Termsmv M(s)>0$ and $\Tracevaldist{M}(A)>0$ then $P(s,A)>0$.
\end{lemma}
\begin{SHORT}
\begin{proof}\belowdisplayskip=-9pt
  \newcommand{\An}{A\rvert_n}
  There is $n$ such that $\Termsmv{M}(A\cap\sampseqp{n})>0$. 
  Write $\An=A\cap\sampseqp{n}$. For all $t\in\An$, $q(s,t)>0$ by case
  analysis on whether $n\le\Abs{s}$.
    If $n\le\Abs{s}$, then for all $t\in\An$,
    \begin{align*}
      q(s, t) &= \Pi_{i=1}^n\pdf{\mathsf{Gaussian}}(s_i,\sigma^2,t_i)> 0 & \text{and}\\
      q(t, s) &= (\Pi_{i=1}^n\pdf{\mathsf{Gaussian}}(t_i,\sigma^2,s_i))\cdot
                \termsmv{\texttt{peval}(M,s_{1..n})}{s_{(n+1)},\dots,s_n}>0.\\
    \end{align*}
    Similarly, if $n>\Abs{s}$, then for all $t\in\An$,
    \begin{align*}
      q(s, t) &= (\Pi_{i=1}^{\Abs{s}}\pdf{\mathsf{Gaussian}}(s_i,\sigma^2,t_i)) \cdot
                \termsmv{\texttt{peval}(M,t_{1..{\Abs{s}}})}{t_{({\Abs{s}}+1)},\dots,t_{\Abs{s}}}> 0 & \text{and}\\
      q(t, s) &= \Pi_{i=1}^{\Abs{s}}\pdf{\mathsf{Gaussian}}(t_i,\sigma^2,s_i)>0.
    \end{align*}
Since  $\mu(\An)>0$ and $\Termsmv M(t)>0$ for all $t\in\An$, 
  \begin{align*}
    P(s, A) &\ge P(s,\An) 
 %   \\ &= \int_{\An} \alpha(s,t)\, Q(s, dt) + \indfun{\An}s \cdot \int (1 - \alpha(s,t))\,Q(s, dt)
    \\ &\ge \int_{\An} \alpha(s,t)\, Q(s, dt) 
    \\ &= \int_{\An} \alpha(s,t) q(s, t)\,dt
    \\ &=\int_{\An} \min\Set{q(s, t),\frac{\Termsmv M(t) q(t,s)}{\Termsmv M(s)}}\,dt
\\ &> 0.
  \end{align*}
  \qed
\end{proof}
\end{SHORT}
\begin{LONG}
\begin{proof}\belowdisplayskip=-9pt
  \newcommand{\An}{A\rvert_n}
  There is $n$ such that $\Termsmv{M}(A\cap\sampseqp{n})>0$. 
  Write $\An=A\cap\sampseqp{n}$. For all $t\in\An$, $q(s,t)>0$ by case
  analysis on whether $n\le\Abs{s}$.
    If $n\le\Abs{s}$, then for all $t\in\An$,
    \begin{align*}
      q(s, t) &= \Pi_{i=1}^n\pdf{\mathsf{Gaussian}}(s_i,\sigma^2,t_i)> 0 & \text{and}\\
      q(t, s) &= (\Pi_{i=1}^n\pdf{\mathsf{Gaussian}}(t_i,\sigma^2,s_i))\cdot
                \termsmv{\texttt{peval}(M,s_{1..n})}{s_{(n+1)},\dots,s_n}>0.\\
    \end{align*}
    Similarly, if $n>\Abs{s}$, then for all $t\in\An$,
    \begin{align*}
      q(s, t) &= (\Pi_{i=1}^{\Abs{s}}\pdf{\mathsf{Gaussian}}(s_i,\sigma^2,t_i)) \cdot
                \termsmv{\texttt{peval}(M,t_{1..{\Abs{s}}})}{t_{({\Abs{s}}+1)},\dots,t_{\Abs{s}}}> 0 & \text{and}\\
      q(t, s) &= \Pi_{i=1}^{\Abs{s}}\pdf{\mathsf{Gaussian}}(t_i,\sigma^2,s_i)>0.
    \end{align*}
Since  $\mu(\An)>0$ and $\Termsmv M(t)>0$ for all $t\in\An$, 
  \begin{align*}
    P(s, A) &\ge P(s,\An) 
 %   \\ &= \int_{\An} \alpha(s,t)\, Q(s, dt) + \indfun{\An}s \cdot \int (1 - \alpha(s,t))\,Q(s, dt)
    \\ &\ge \int_{\An} \alpha(s,t)\, Q(s, dt) 
    \\ &= \int_{\An} \alpha(s,t) q(s, t)\,dt
    \\ &=\int_{\An} \min\Set{q(s, t),\frac{\Termsmv M(t) q(t,s)}{\Termsmv M(s)}}\,dt
\\ &> 0.
  \end{align*}
  \qed
\end{proof}
\end{LONG}
\begin{corollary}[Irreducibility]\label{lem:irreducible}
  $P$ as given by Equation~(\ref{equation:transition-kernel}) is $\pi$-irreducible.
\end{corollary}

\begin{comment}
  \begin{lemma}[\citet{tierney1994}]\label{lem:tierney.aperiodic}
    Let $r(s)=\int (1 - \alpha(s,t))\,Q(s, dt)$ be the probability
    that Metropolis-Hastings kernel $P$ with target distribution
    $\distrone$ does not move from state $s$.  If
    $\pi(\Set{s\mid r(s)>0})>0$ then $P$ is $\distrone$-aperiodic.
  \end{lemma}
\end{comment}
\begin{lemma}[Aperiodicity]\label{lem:aperiodic}
  $P$ as given by Equation~(\ref{equation:transition-kernel}) is $\pi$-aperiodic.  
\end{lemma}
\begin{proof}
Assume that $B_{1},B_{2}$ are disjoint sets such that $\pi(B_1) > 0$ and $P(s,B_2) = 1$ for all $s \in B_1$.
If $s\in B_{1}$, Lemma~\ref{lem:strong-irred} gives that $P(s, B_{1})>0$, so $P(s, B_{2})<P(s,\sampseq)=1$, which is a contradiction. 
A fortiori, $P$ is $\pi$-aperiodic.\qed
\end{proof}

% \begin{corollary}\label{cor:trace-convergence}
%   For all $s$ such that $\Termsmv M(s)>0$,  $ \lim_{n \rightarrow \infty}|| P^n(s, \cdot) - \pi || = 0$.
% \end{corollary}

%Question: as a corollary can we prove a similar correctness statement for the value distributions as opposed to the trace distributions?

%Below we write $\pi_M$ for $\pi$ and $P_M$ for $P$, to make the dependence on $M$ explicit.

\begin{lemma} \label{lemma:transform-norm}
If $\mu_1$ and $\mu_2$ are measures on $(X_1, \Sigma_1)$ and $f:X_1 \rightarrow X_2$
is measurable $\Sigma_1 / \Sigma_2$, then
\[
 ||\mu_1 f^{-1} - \mu_2 f^{-1} || \le ||\mu_1 - \mu_2||
\]
\end{lemma}
\begin{proof}
We have $\sup_{B \in \Sigma_2} | \mu_1 f^{-1}(B) - \mu_2 f^{-1}(B) | = \sup_{A \in \Sigma'_1} ||\mu_1(A) - \mu_2(A) ||$,
where $\Sigma'_1 = \{ f^{-1}(B) | B \in \Sigma_2 \}$.
By measurability of $f$ we get $\Sigma'_1 \subseteq \Sigma_1$, 
so by monotonicity of $\sup$ we get 
$\sup_{A \in \Sigma'_1} |\mu_1(A) - \mu_2(A) | \le \sup_{A \in \Sigma_1} |\mu_1(A) - \mu_2(A)|$.\qed
\end{proof}

\begin{restate}{Theorem \ref{thm:sampling-trace}}
\mbox{For every trace $s$ with $\Termsmv M(s)\neq 0$,}
\[
 \lim_{n \rightarrow \infty}
 || T^n(s, \cdot) - \valdist{M} || = 0.
\]
\end{restate}
\begin{proof}
By Corollary~\ref{lem:irreducible}, $P$ is $\pi$-irreducible, and by
Lemma~\ref{lem:aperiodic}, $P$ is $\pi$-aperiodic.
Lemma~\ref{lem:tierneyConvergence} then yields that
\[
 \lim_{n \rightarrow \infty}
 || P^n(x, \cdot) - \pi || = 0.
 \]
By definition, $T^n(s,A) = P^n(s, \mathbf{O}_{M}^{-1}(A))$ and
$\valdist M (A) = \tsq{M}(A\cap\valset)/\tsq{M}(\valset)$.  By Theorem
\ref{thm:sampling-distribution}, $\tsq{M}(A\cap\valset) =
\sbd{M}(A\cap\valset) = \Tracedist{M}(\mathbf{O}_M^{-1}(A \cap
\valset) = \Tracedist{M}(\mathbf{O}_M^{-1}(A) \cap
\mathbf{O}_M^{-1}(\valset)) = \Tracevaldist{M}(\mathbf{O}_M^{-1}(A))$
and similarly $\tsq{M}(\valset) =
\Tracedist{M}(\mathbf{O}_M^{-1}(\valset)) =
\Tracevaldist{M}(\mathbb{S})$, which gives $\valdist M (A) =
\pi(\mathbf(O_M^{-1}(A))$.  Thus, by Lemma \ref{lemma:transform-norm}
and the squeeze theorem for limits we get \belowdisplayskip=-10pt
\[
 \lim_{n \rightarrow \infty}
 || T^n(s, \cdot) - \valdist{M} || \leq  
 \lim_{n \rightarrow \infty}
 || P^n(s, \cdot) - \pi || = 0.
\]
\qed
\end{proof}

%\section{Extensions}
%Call-by-name?

\subsection{Examples}
To illustrate how inference works, we revisit the geometric distribution
and linear regression examples from Section~\ref{sec:syntax}. Before discussing the transition
kernels for these models, note that the products of Gaussian densities
always cancel out in the acceptance probability $\alpha$, because 
$\pdf{\mathsf{Gaussian}}(s_i,\sigma^2,t_i) = \pdf{\mathsf{Gaussian}}(t_i,\sigma^2,s_i)$
by the definition of the Gaussian PDF.

\paragraph{Geometric Distribution}
Let us begin with the implementation of the geometric distribution, described in \ref{subsec:geom}, which we will call $M_1$. 
Since the only random primitive used in $M_1$ is $\textsf{rnd}$,
whose density is $1$ on all its support, and there are no calls to $\mathtt{score}$, the weight of any
trace that yields a value must be $1$. 
Because the $\mathit{geometric}$ function applied to $0.5$ returns a value immediately when
the call to $\mathsf{rnd}$ returns a number smaller than $0.5$, and recursively calls itself otherwise,
otherwise returns a value immediately

The function $\mathit{geometric}$ applied to $0.5$ calls itself recursively if the call to $\mathsf{rnd}$
returned a value greater or equal to  a half, and returns a value immediately otherwise. Hence, every valid trace 
consists of a sequence of numbers in $[0.5,1]$, followed by a number in $[0, 0.5)$,
and so the set of valid traces is
$S_1= \{s\mid s_i \in [0,0.5)$ for $i<\Abs s \land s_{\Abs{s}}\in[0.5,1]\land \Abs s > 0\}$.
The proposal density is \[q(s,t) = [t \in S_1] \Pi_{i = 1}^k \pdf{\mathsf{Gaussian}}(s_i, \sigma^2, t_i)\]
where $k=\min\Set{\Abs{s},\Abs{t}}$. The term $[t \in S_1]$ reflects the fact that for every non-valid
trace, $\mathbf{P}^{\valset}_{\mathtt{peval}(M, t_{1..k})}(t_{k+1..|t|}) = 0$.

As noted above, the Gaussians cancel out in the acceptance ratio, and
the density of every valid trace is $1$, so $\alpha(s,t) = [t \in
  S_1]$. This means that every valid trace is accepted.  The
transition kernel of the Markov chain induced by the MH algorithm is
\begin{multline*}
P(s,A) = \int_{A \cap S_1}\Pi_{i = 1}^{\min\Set{\Abs{s},\Abs{t}}} \pdf{\mathsf{Gaussian}}(s_i, \sigma^2, t_i) \mu(dt) \\
 + [s \in A] \int_{\mathbb{S} \setminus S_1} \Pi_{i = 1}^{\min\Set{\Abs{s},\Abs{t}}} \pdf{\mathsf{Gaussian}}(s_i, \sigma^2,t_i) \mu(dt) 
\end{multline*}

\paragraph{Linear Regression with $\mathsf{flip}$}

Now, consider the linear regression model from section
\ref{sec:exampl-line-regr}, which can be translated from Church to the
core calculus by applying the rules in Figure \ref{fig:translation}
(details are omitted, but the translation is straightforward as there
is no recursion).

In every trace in this translated model, which we call $M_2$, we have
two draws from $\mathsf{Gaussian}(0,2)$ , whose values are assigned to
variables $m$ and $b$.
They are followed by four calls to $\mathtt{rnd}$ made while
evaluating the four calls to $\mathtt{softeq}$. The conditioning
statement at the end sets the return value to $\tfalse$ if at least
one call to $\mathtt{softeq}$ evaluates to $\tfalse$. Since
$\mathtt{softeq}\ x\ y$ returns $\ttrue$ if and only if the
corresponding call to $\mathtt{rnd}$ returned a value less than
$\mathtt{squash}\ x\ y$, it follows from the definitions of
$\mathtt{squash}$ and $\mathtt{f}$ that the element of the trace
consumed by $\mathtt{softeq}\ (f\ x)\ y$ must be in the interval $[1,
  \frac{1}{e^{(m \cdot x + b - y)^2}})]$.  Note that since the pdf of
$\mathtt{rnd}$ is flat, the weight of any trace depends only on the
first two random values, drawn from the Gausians, as long as the
remaining four random values are in the right intervals.

The full density for this model is
\begin{eqnarray*}
 \mathbf{P}_{M_2}^\valset(s) &=& 
\left(\Pi_{i=1}^2 \pdf{\mathsf{Gaussian}}(0, 2, s_i) \right) \cdot \\
&&\quad \left(\Pi_{i=1}^4\left[s_{i+2} \in\left[1, \frac{1}{e^{(s_i \cdot x_i + s_2   - y_i)^2}}\right)\right]\right) 
\end{eqnarray*}
\noindent
if $s~\in~\mathbb{R}^6$ and $ \mathbf{P}_{M_2}^\valset(s) = 0$ if $s~\notin~\mathbb{R}^6$.

%In the above equation, $m$ and $b$ are the outcomes of the first and second random
%
%Since the pdf of $\mathtt{rnd}$ is flat, the weight of any valid trace depends only on the first two random 
%values, drawn from the Gausians. However, in order to be valid, every trace 
%
%The full proposal density for this model is
%
%%
%They are followed by
%four calls to $\mathsf{flip}$ (encoded using the uniform distribution, like in the previous example),
%which simulate soft conditioning. The conditioning statement at the end sets the return value to $\fail$ if
%at least one of the $\mathsf{flip}$s returned $\tfalse$, so  %rephrase it?
%a trace $t$ is only valid (i.e. $\mathbf{P}^{\valset}(t) \neq 0$) if for every $i \in 1..4$, the value
%returned from $i$-th call to $\mathsf{rnd}$ is in $[0,c_i)$, where $c_i$ is the corresponding 
%argument to $\mathsf{flip}$.
%
%all the values drawn
%from $\mathsf{rnd}$ are in $[0,c)$.
%
%%  observations encoded as calls to $\mathsf{flip}$ (defined in terms of the uniform
%%distribution, like in the previous example), which must all return $\mathtt{true}$
%%for a trace to have non-zero probability.
%%not to return $\fail$. 
%%We assume that $\mathsf{flip}$ is itself encoded using the uniform
%%distribution, like in the previous example. 
%Note that all valid traces have length $6$. We have 
%$\mathbf{P}_{M_2}^\valset(s) = 
%\left(\Pi_{i=1}^2 \pdf{\mathsf{Gaussian}}(0, 2, s_i) \right)\left(\Pi_{i=1}^4\left[s_{i+2} \in\left[1, \frac{1}{e^{(s_i \cdot x_i + s_2   - y_i)^2}}\right)\right]\right)$ when $s~\in~\mathbb{R}^6$.
%
Note that the partial derivative of $\mathbf{P}_{M_2}^\valset(s)$ with respect to each of $s_3,s_4,s_5,s_6$ is zero wherever defined, precluding the use of efficient gradient-based methods for searching over these components of the trace.

%NEW

Now, let us derive the density $q(s,t)$, assuming that
 $\mathbf{P}_{M_2}^\valset(s)  > 0$ 
 (which implies $s \in \mathbb{R}^6$, as shown above) and $t \in \mathbb{R}^6$. Since we
have $|s| = |t| =6$, the formula for $q$ reduces to:
\[
q(s,t) = (\Pi_{i=1}^6 \pdf{\mathsf{Gaussian}}(s_i, \sigma^2, t_i)
\mathbf{P}^\valset_{M_2'}([])
\]

where $M_2' = \mathtt{peval}(M_2, t)$.
%The first two Gaussians express the probabilities of 
Because there cannot be more than six
random draws in any run of the program, $M_2'$ is determinstic.
%One can also check that if 
% Add new lemma?
This means that $\mathbf{P}^\valset_{M_2'}([]) = 0$ if $M_2' \Downarrow^{[]}_1 \fail$
and  $\mathbf{P}^\valset_{M_2'}([]) = 1$ if $M_2' \Downarrow^{[]}_1 \fail$
for some $V$. 

It is easy to check that if $t \notin \mathbb{R}^6$, then $q(s,t) = 0$---since
there is no trace of length other than $6$ leading to a value, the value of
$\mathbf{P}_{M_2'}^\valset(t_{k+1..|t|})$ in the definition of $q$ must be $0$ in this case.

%It is easy to chceck that $M_2' \Downarrow^{[]}_1 \fail$ if and only
%if $M_2 \Downarrow^t_w V$ for some $w>0$. As already noted,  
%$M_2$ reduces to a value if and only if $t_3,t_4,t_5,t_6$
%are in the right intervals (dependent on $t_1$ and $t_2$), so we can write
%\[
%\mathbf{P}^\valset_{M_2'}([]) = 
%\Pi_{i=1}^4[t_{i+2} \in [0, \frac{1}{e^{(t_1 \cdot x_i + t_2   - y_i)^2}})]
%\]
%%Assuming $\mathbf{P}_{M_2}^\valset(s)  > 0$, the proposal density is
%If $|t| < 6$, then we have
%\[
%q(s,t) = (\Pi_{i=1}^{|t|} \pdf{\mathsf{Gaussian}}(s_i, \sigma^2, t_i)
%\mathbf{P}^\valset_{M_2'}([])
%\]
%where $M_2' = \mathtt{peval}(M_2, t)$. Hence, $q(s,t) = 0$, becase
%$M_2'$ cannot reduce to a value with an empty trace, as any
%trace in $M_2'$ contains at least one random draw. Similarly,
%if $|t|>6$, $q(s,t)  = 0$, since $M_2' = \mathtt{peval}(M_2, t_{1..6})$
%is deterministic, so $M_2'$ cannot reduce to a value with a non-empty
%trace (as $\mathbf{P}$ requires all elements of the trace to be consumed).

Thus, the proposal density is
\begin{eqnarray*}
q(s,t) &=& \left(\Pi_{i=1}^6 \pdf{\mathsf{Gaussian}}(s_i, \sigma^2, t_i) \right) \cdot \\
&&\qquad \qquad
\left(\Pi_{i=1}^4[t_{i+2} \in [0, \frac{1}{e^{(t_i \cdot x_1 + t_2   - y_i)^2}})]\right) 
\end{eqnarray*}
for $t~\in~\mathbb{R}^6$ and $q(s,t) = 0$ for $t~\notin~\mathbb{R}^6$.
%
%In the above equation, the term $(\Pi_{i=1}^6 \pdf{\mathsf{Gaussian}}(s_i, \sigma^2, t_i)$
%comes from (note that $k  = |s| = |t| = 6$), and

Hence, the acceptance ratio reduces to
%\[
%\alpha(s,t) = [t~\in~\mathbb{R}^6]
%\frac{\Pi_{i=1}^2 \pdf{\mathsf{Gaussian}}(0, 2, t_i)}{\Pi_{i=1}^2 \pdf{\mathsf{Gaussian}}(0, 2, s_i)} 
%\Pi_{i=1}^4[t_{i+2} \in [0, \frac{1}{e^{(t_1 \cdot x_i + t_2   - y_i)^2}})] . 
%\]
%
\begin{eqnarray*}
\alpha(s,t)&=& \min \{1, 
\frac{\Pi_{i=1}^2 \pdf{\mathsf{Gaussian}}(0, 2, t_i)}
{\Pi_{i=1}^2 \pdf{\mathsf{Gaussian}}(0, 2, s_i)}\cdot \\ 
&&\qquad \qquad 
 \Pi_{i=1}^4[t_{i+2} \in [0, \frac{1}{e^{(t_1 \cdot x_i + t_2   - y_i)^2}})]  \}
\end{eqnarray*}
if $t \in \mathbb{R}^6$ and $\alpha(s,t) = 0$ otherwise.

Note that $\alpha(s,t)$ is only positive if each of $t_3,t_4,t_5,t_6$ are within a certain (small) interval.
This is problematic for an implementation, since it will need to find suitable values for all these components of the trace for every new trace to be proposed, leading to inefficiencies due to a slowly mixing Markov chain.

% This means that $\alpha$ only depends on the prior distributions and not
% on the conditions, which in practice will lead to inefficient inference. This illustrates
% the advantage of using soft conditioning in MCMC algorithms.
% %and $\mathbf{P}_{M_2}^\valset(s) = 0$ if $\Abs{s} \neq 6$. 

% %The proposal density is
% %$q(s, [t_1, \dots, t_6]) = $

\paragraph{Linear Regression with $\mathsf{score}$}
In this alternative version $M_3$ of the previous model, we also have two draws from $\mathsf{Gaussian}(0,2)$ 
at the beginning, but the calls to $\mathsf{flip}$ are replaced with calls to $\mathtt{score}$,
which multiply the trace density by a positive number 
without consuming any elements of the trace. Because the support 
of the Gaussian PDF is $\mathbb{R}$ and there are precisely two random draws 
(both from Gaussians)
in every trace leading to a value,
the set of valid traces is $\mathbb{R}^2$.
We have $\mathbf{P}_{M_3}^\valset(s) = \Pi_{i=1}^2 \pdf{\mathsf{Gaussian}}(0, 2, s_i)\ \Pi_{i=1}^4  e^{- (s_1 \cdot x_i + s_2   - y_i)^2}$
if  $s~\in~\mathbb{R}^2$ and  $\mathbf{P}_{M_3}^\valset(s) = 0$ otherwise.
Assuming $\mathbf{P}_{M_3}^\valset(s)  > 0$ and $t\in \mathbb{R}^2$, we get the proposal density 
\[
q(s,t) = \Pi_{i=1}^2 \pdf{\mathsf{Gaussian}}(s_i, \sigma^2, t_i) \Pi_{i=1}^4 
\frac{1}{e^{(t_1 \cdot x_i + t_2   - y_i)^2}}
\]
where $x = [0,1,2,3]$ and $y = [0,1,4,6]$. If $t\notin \mathbb{R}^2$, then
$q(s,t)=0$, because otherwise there would be a trace of length different than 2
leading to a value.

Thus, the acceptance ratio is
\begin{align*}
\alpha(s,t) &{}= \frac{ \Pi_{i=1}^2 \pdf{\mathsf{Gaussian}}(0, 2, t_i)} 
{\Pi_{i=1}^2 \pdf{\mathsf{Gaussian}}(0, 2, s_i)}%\cdot
%\frac{\Pi_{i=1}^4  e^{(t_i \cdot x_1 + t_2   - y_i)^2}}{ \Pi_{i=1}^4  e^{(s_i \cdot x_1 + s_2   - y_i)^2}}
\end{align*}
if $t~\in~\mathbb{R}^2$ and $\alpha(s,t) = 0$ otherwise.

In contrast to the previous example, here the acceptance ratio is positive for all proposals, non-zero gradients exist almost everywhere, and there are four fewer nuisance parameters to deal with (one per data point!). This makes inference for this version of the model much more tractable in practice.
%
%The resulting transition kernel for this model is
%\begin{align*}
%P(s,A) =& \int_{\mathbb{R}^6}\Pi_{i = 1}^{\min\Set{\Abs{s},\Abs{t}}} \pdf{\mathsf{Gaussian}}(s_i, \sigma^2, t_i) \mu(dt) \\
%& + [s \in A] \int_{\mathbb{S} \setminus S_1} \Pi_{i = 1}^{\min\Set{\Abs{s},\Abs{t}}} \pdf{\mathsf{Gaussian}}(s_i, \sigma^2,t_i) \mu(dt) 
%\end{align*}

%
%
%The only probability distribution used in the geometric distribution is $\textsf{rnd}$,
%the uniform distribution on the unit interval. Since its density is constant and equal to $1$,
%reducing a call to $\textsf{rnd}$ does not change the weight of the sample, 
\section{Related Work}\label{sec:related}
To the best of our knowledge, the only previous theoretical
justification for trace MCMC is the recent work by~\citet{corsamp}, 
who show correctness of trace MCMC for the imperative probabilistic language \Rtwo~\cite{DBLP:conf/aaai/NoriHRS14}.
Their result does not apply to higher-order languages such as \Church or our $\lambda$-calculus, nor to programs that do not almost surely terminate~\cite{Hermanns15.terminationProbabilistic}.
Their algorithm is different from ours in that it exploits the explicit storage locations found in imperative programs, keeping one sample trace per location.
%
\begin{comment} % text below relates to our previous algorithm
The inference algorithm of \citet{corsamp} differs from ours in that
in every step, the values of all random variables in a program are
resampled (if a previous value is available, the algorithm draws a new
value from a Gaussian centred at the old value).  This is unlike in
our algorithm, where all the random choices occurring in the trace
before the variable to be perturbed are left unchanged. %Arguably, our
approach is more efficient (TODO: need to verify this claim somehow),
\end{comment}
%
The authors do state that the space of traces in their language is
equipped with a ``stock'' measure, and that the distributions of
program traces and transitions are given by densities with respect to
that measure.  They do not, however, show that these densities are
measurable.
Their proof of correctness only shows that the acceptance ratio
$\alpha$ computed by their algorithm matches Hasting's formula: the
authors do not prove convergence of the resulting Markov chain.
Indeed, properties such as irreducibility and aperiodicity depend on
the choices of parameters in the algorithm.

% IMPLEMENTATIONS
Other probabilistic language implementations also use trace MCMC inference,
including \Church~\cite{DBLP:conf/uai/GoodmanMRBT08}, 
\Venture~\cite{DBLP:journals/corr/MansinghkaSP14}, \WebPPL~\cite{dippl}, 
and \Anglican~\cite{DBLP:conf/pkdd/TolpinMW15}.  
These works focusing on efficiency and convergence properties, 
and do not state formal correctness claims for their implementations.

% LIGHTWEIGHT
\citet{wingate2011lightweight} give a general program transformation
for a probabilistic language to support trace MCMC, with a focus on
labelling sample points in order to maximise sample reuse.  Extending
our trace semantics with such labelling is important future work,
given that \citet{problems-kiselov16} points out some difficulties
with the transformation and proposes alternatives.

% SCORING
Many recent probabilistic languages admit arbitrary non-negative
scores.  This is done either by having an explicit
$\texttt{score}$-like function, as in 
%\Stan (called \texttt{increment\_log\_prob}) or 
\WebPPL (called \texttt{factor}), or by
observing that a particular value $V$ was drawn from a given
distribution $\distone(\vec{c})$ (without adding it to the trace), as
in \WebChurch (written $(\distone\ \vec{c}\ V)$) or \Anglican
(\lstinline{observe} $(\distone\ \vec{c})\ V$).
In recent work for a non-recursive $\lambda$-calculus with
\texttt{score},~\citet{DBLP:journals/corr/StatonYHKW16} note that
unbounded scores introduce the possibility of ``infinite model
evidence errors''. As seen in Section~\ref{sec:motivation}, 
even statically bounded scores exhibit this problem in the presence of recursion.

%SAY SOMETHING ABOUT \cite{DBLP:journals/corr/RitchieSG15}.

% KOZEN's TRACE SEMANTICS AND FOLLOWERS
\citet{DBLP:conf/focs/Kozen79} gives a semantics of imperative
probabilistic programs as partial measurable functions from infinite
random traces to final states, which serves as the model for our trace
semantics. Kozen also proves this semantics equivalent to a
domain-theoretic one.
\citet{park08sampling} give an operational version of Kozen's
trace-based semantics for a $\lambda$-calculus with recursion, but
``do not investigate measure-theoretic properties''.
\citet{CousotMonerau-ESOP2012-PAI} generalise Kozen's trace-based
semantics to consider probabilistic programs as measurable functions
from a probability space into a semantics domain, and study abstract
interpretation in this setting.
\citet{DBLP:conf/esop/TorontoMH15} use a pre-image version of Kozen's
semantics to obtain an efficient implementation using rejection
sampling.
\citet{DBLP:conf/haskell/ScibiorGG15} define a monadic embedding of
probabilistic programming in Haskell along the lines of Kozen's
semantics; their paper describes various inference algorithms but has
no formal correctness results.
 
% MISC. DENOTATIONAL SEMANTICS
\citet{DBLP:conf/popl/RamseyP02} provide a monadic denotational
semantics for a first-order functional language with discrete
probabilistic choice, and a Haskell implementation of the expectation
monad using variable elimination.
\citet{BBGR12:DerivingPDFs} define a denotational semantics based on
density functions for a restricted first-order language with
continuous distributions.  They also present a type system ensuring
that a given program has a density.

% DOMAIN-THEORETIC APPROACHES
\citet[Chapter 8]{jones90:PhD} defines operational and
domain-theoretic semantics for a $\lambda$-calculus with discrete
probabilistic choice and a fixpoint construct.  Our distribution-based
operational semantics generalises Jones's to deal with continuous
distributions.
Like Kozen's and Jones's semantics, our operational semantics makes
use of the partially additive structure on the category of
sub-probability kernels~\cite{panangaden99markovkernels} in order to
treat programs that make an unbounded number of random draws.
\citet{DBLP:journals/corr/StatonYHKW16} give a domain-theoretic
semantics for a $\lambda$-calculus with continuous distributions and
unbounded \texttt{score}, but without recursion.
While giving a fully abstract domain theory for probabilistic
$\lambda$-calculi with recursion is known to be
hard~\cite{JonesPlotkin89}, there have been recent % JungTix 
advances using probabilistic coherence
spaces~\cite{DanosEhrhard11,EhrhardTassonPagani14} and game
semantics~\cite{DanosHarmer02}, which in some cases are fully
abstract.  We see no strong obstacles in applying any of these to a
typed version of our calculus, but it is beyond the scope of this
work.
Another topic for future work are methodologies for equivalence
checking in the style of logical relations or bisimilarity, which have
been recently shown to work well in \emph{discrete} probabilistic
calculi \cite{DBLP:conf/fossacs/BizjakB15}.

%\url{http://mit-church.googlecode.com/svn/trunk/church/xrp-preamble.church}

% Thomas Leventis has some recent work on equivalence for probabilistic lambda-calculi
% \url{http://chocola.ens-lyon.fr/events/seminaire-2015-10-01/talks/leventis/}
%%%%%%%%%%%%%%%%%%%%%
\section{Conclusions and Remarks}\label{sec:conc}
%%%%%%%%%%%%%%%%%%%%%
As a foundation for probabilistic inference in languages such as
\Church, we defined a probabilistic $\lambda$-calculus with draws from
continuous probability distributions and both hard and soft
constraints, defined its semantics as distributions on terms, and
proved correctness of a trace MCMC inference algorithm via a sampling
semantics for the calculus.

Although our emphasis has been on developing theoretical
underpinnings, we also implemented our algorithm in F\# to help
develop our intuitions and indeed to help debug definitions.
The algorithm is correct and effective, but not optimized.
In future, we aim to extend our proofs to cover more efficient
algorithms, inspired by \citet{wingate2011lightweight} and
\citet{problems-kiselov16}, for example.

% ADG: following is a bit crude, sorry

\appendix

\section{Proofs of Measurability}\label{section:proof-of-measurability}
This appendix contains the proofs of measurability of $\mathbf{P}_M$,
$\mathbf{O}_M$, $\mathbf{P}^\valset_M$,  $\mathtt{peval}$ and $q$, as well as
a proof that $Q$ is a probability kernel.

The proofs usually proceed by decomposing the functions into simpler operations.
However, unlike \citet{TorontoPhD}, we do not define these functions entirely in
terms of general measurable operators, because the scope for reuse is limited here. 
We would have, for instance, to define multiple functions projecting different 
subexpressions of different expressions, and prove them measurable. Hence, the overhead
resulting from these extra definitions would be greater than the benefits.

First we recap some useful results from measure theory:

%\MS{Include Buckley74, leave \ref{lemma:image-measurable} as corollary}

\begin{itemize}
\item
  A function $f: X_1 \rightarrow X_2$ between metric spaces $(X_1,d_1)$
  and $(X_2, d_2)$ is \emph{continuous} if for every $x \in X_1$ and
  $\epsilon > 0$, there exists $\delta$ such that for every $y \in X_1$, if $d_1(x,y) < \delta$,
  then $d_2(f(x), f(y)) < \epsilon$.
%\item
%  A sequence $\{x_n\}_{n \in \mathbb{N}}$ is \emph{Cauchy} if
%  \[
%  \forall \epsilon >0\ \exists n \in \mathbb{N}\ \forall k,l \ge n \quad
%  d(x_k, x_l) < \epsilon
%  \]
%\item
%  A metric space $(X,d)$ is \emph{complete} if every Cauchy sequence
%  $\{x_n\}$ on $X$ has a limit in $X$.
\item
  A subset $A$ of a metric space $(X,d)$ is \emph{dense} if
  \[
  \forall x \in X, \epsilon>0\ \exists y \in A\quad d(x,y) < \epsilon
  \]
\item
  A metric space is \emph{separable} if it has a countable dense subset.
\item  
  Given a sequence of points $x_n$ in a metric space $(X,d)$, we say
  that $x$ is the \emph{limit} of $x_n$ if for all $\epsilon > 0$, there
  exists an $N$ such that $d(x_n,x) < \epsilon$.
\item
  %A \emph{limit point} of a sequence 
  A subset $A$ of a metric space is \emph{closed} if it contains all the limit points,
  that is if $x_n \in A$ for all $n$ and $x_n \rightarrow x$, then $x \in A$. 

% A closed subset
%   of a metric space is also closed in the topology induces by this space, and hence is
%   a measurable subset of the Borel $\sigma$-algebra generated by the open sets of that topology.
  
%  A Borel $\sigma$-algebra generated by a topology contains all closed sets in that topology.
\end{itemize}

\begin{lemma}[{{\citet[ex. 13.1]{billingsley95}}}] \label{lemma:split-space}
Let $(\Omega, \Sigma)$ and $(\Omega', \Sigma')$ be two measurable spaces,
$T: \Omega \rightarrow \Omega'$ a function and
$A_1, A_2, \dots$ a countable collection of sets in $\Sigma$ whose union is $\Omega$.
Let $\Sigma_n = \{A \mid A \subseteq A_N, A \in \Sigma \}$ be a $\sigma$-algebra in $A_n$
and $T_n: A_n \rightarrow \Omega'$ a restriction of $T$ to $A_n$.
Then $T$ is measurable $\Sigma / \Sigma'$ if and only if $T_n$ is measurable
$\Sigma_n / \Sigma'$ for every $n$.
\end{lemma}

\hspace{0.1in}

A convenient way of showing that a function is Borel-measurable is to show that it is
continuous as a function between metric spaces.
% which induce the topologies generating
%the Borel $\sigma$-algebras on the domain and codomain of the function.

Let us represent the product $\sigma$-algebra $\mathcal{M} \times \mathcal{R} \times \mathcal{S}$
as a Borel $\sigma$-algebra induced by a metric. 
First, we define the standard metric on $\mathbb{R}$, and the disjoint union of Manhattan metrics for $\mathbb{S}$:
\begin{eqnarray*}
  d_{\mathbb{R}}(w,w') &\deq& |w - w'|\\
  d_{\mathbb{S}}(s,s') &\deq&
  \begin{cases}
  \sum_{i=1}^{|s|} |s_i - s'| & \text{if}\ |s| = |s'|
  \\
  \infty \quad &\text{otherwise}
  \end{cases}
\end{eqnarray*}

%Now, obviously $(\mathbb{R}, d)$ induces the standard topology on $\mathbb{R}$. 
We can easily verify that $(\mathbb{S}, d)$ generates $\mathcal{S}$. We define the metric on
$\Lambda \times \mathbb{R} \times \mathbb{S}$ to be the Manhattan metric:

\[
  d((M,w,s), (M',w',s')) \deq d_{\Lambda}(M,M') + d_{\mathbb{R}}(w,w') + d_{\mathbb{S}}(s,s')
\]

% It is a well known fact that the topological space induced by the Manhattan metric
% is the product topology if the spaces induced by the component metrics.

The following is a standard result in measure theory:
\begin{lemma}[{{\citet[Proposition 4.2 b)]{gallay2009mesure}}}]
If $X_1, X_2$ are %topological spaces with topologies induced by
separable metric spaces then %and $\mathcal{B}(X)$ is the Borel $\sigma$-algebra generated by the topology of $X$, then
\[
\Borel(X_1 \times X_2) = \Borel(X_1) \times \Borel(X_2)
\]
\end{lemma}

It is obvious that $(\mathbb{R}, d)$ and $(\mathbb{S}, d)$ are separable. Now, 
let $\Lambda_Q$ be the subset of $\Lambda_P$ in which all constants are rational. Then, it is easy
to show that $\Lambda_Q$ is countable. 

\begin{lemma} \label{lemma:terms-dense}
$\Lambda_Q$ is a dense subset of $(\Lambda_P, d)$
\end{lemma}
\begin{proof}
We need to prove that
\[
  \forall M\in \Lambda_P, \epsilon>0\ \exists M_Q \in \Lambda_Q \quad
  d(M,M_q) < \epsilon
\]

This can be easily shown by induction (the base case follows from the fact that
$\mathbb{Q}$ is a dense subset of $\mathbb{R}$.
\qed
\end{proof}

\begin{lemma} \label{lemma:terms-separable}
The metric space $(\Lambda_P, d)$ is separable.
\end{lemma}
\begin{proof}
Corollary of Lemma \ref{lemma:terms-dense}.
\qed
\end{proof}

%\begin{lemma} \label{lemma:redexes-separable}
%The metric space $(P_{det}, d)$ is separable.
%\end{lemma}
%\begin{proof}
%Analogous to the proof of separability of $(\Lambda_P, d)$.
%\end{proof}

\begin{corollary} \label{lemma:triples-complete-separable}
The $\sigma$-algebra on $\lamterms \times \mathbb{R} \times \mathbb{S}$ generated by the
metric $d$ is $\mathcal{M} \times \mathcal{R} \times \mathcal{S} $.
\end{corollary}

%We need to define a metric space on the set $\lamterms \times \mathbb{R}_+ \times \mathbb{S}$,
%show that it is complete separable and that the $\sigma$-algebra generated by its metric
%is the product $\sigma$-algebra. %Let:

%Let $\mathcal{T} = \lamterms \times \mathbb{R}_+ \times \mathbb{S} \cup \{ \bot \}$. Define:
%\begin{eqnarray*}
%d((M,w,s), \bot) &\deq& \infty \\
%d(\bot, \bot) &\deq& 0
%\end{eqnarray*}

%\begin{lemma} \label{lemma:triple-space-complete-separable}
%The metric space $(\lamterms \times \mathbb{R}_+ \times \mathbb{S}, d)$ is separable and complete.
%\end{lemma}

Throughout this section, we call a function ``measurable'' if it is Borel measurable
and ``continuous'' if it is continuous as a function between metric spaces.

\hspace{0.1in}

We can use lemma \ref{lemma:split-space} to split the space $\mathcal{M}$ of expressions
into subspaces of expressions of different type, and restrict functions (such as the reduction 
relation) to a given type of expression, to process different cases separately.

We write $\mathbf{Subst}(M,x,v)$ for $\sbst{M}{x}{V}$, to emphasize the fact that
substitution is a function. %  $\Subst{a}{v}{q}$.

\begin{display}[0.2]{Detailed definition of substitution}
%: 
% $\mathbf{Subst}(M,x,v): \termset \times  \times \mathcal{V}  $}
\clause{\mathbf{Subst}(c, x, V) \triangleq c }\\
\clause{\mathbf{Subst}(x, x, V) \triangleq V }\\
\clause{\mathbf{Subst}(x, y, V) \triangleq y \quad \text{if}\ x \neq y }\\
\clause{\mathbf{Subst}(\lambda x . M, x, V) \triangleq \lambda x . M }\\
\clause{\mathbf{Subst}(\lambda x . M, y, V) \triangleq \lambda x . (\mathbf{Subst}(M, y, V)) \quad \text{if}\ x \neq y }\\
\clause{\mathbf{Subst}(M\ N, x, V) \deq \mathbf{Subst}(M, x, V)\ \mathbf{Subst}(N, x, V)  }\\
\clause{\mathbf{Subst}(\distone(V_1,\dots,V_{|\distone|}), x, V) \deq
  \distone(\mathbf{Subst}(V_1, x, V),\dots,\mathbf{Subst}(V_{|\distone|}, x, V)) }\\
\clause{\mathbf{Subst}(g(V_1,\dots,V_{|g|}), x, V) \deq
  g(\mathbf{Subst}(V_1, x, V),\dots,\mathbf{Subst}(V_{|g|}, x, V)) }\\
\clause{\begin{prog} \mathbf{Subst}(\ite{W}{M}{L},x, V) \triangleq {} \\ \quad
  \ite{\mathbf{Subst}(W,x,V)}{\mathbf{Subst}(M,x,V)}{\mathbf{Subst}(L,x,V)} \end{prog}}\\
\clause{\mathbf{Subst}(\mathtt{score}(V'), x, V) \deq \mathtt{score}(\mathbf{Subst}(V', x, V))}\\
\clause{\mathbf{Subst}(\fail, x, V) \deq \fail}
\end{display}

%\MS{Do we need this?}
%
%Let us define the metric on triples of arguments of $\mathbf{Subst}$ and on contexts.
%
%%REDUNDANT: Manhattan metric.
%\begin{eqnarray*}
%d((M,x,V), (N,x,W)) &\deq& d(M,N) + d(V,W)\\
%d((M,x,V), (N,y,W)) &\deq& \infty \quad \text{if}\ x \neq y
%\end{eqnarray*}
%\ectxone\bnf\hole\midd\ectxone\termone\midd(\lambda\varone.\termone)\ectxone 
%
%\begin{lemma}
%$\mathcal{S}$ is the Borel $\sigma$-algebra on $E$.
%\end{lemma}
%\begin{proof}
%$\mathcal{S}$ is a countable direct sum of Borel $\sigma$-algebras.
%\end{proof}

For convenience, let us also define a metric on contexts:

\begin{eqnarray*}
d(\hole, \hole) &\deq& 0\\
d(E M, F N) &\deq& d(E, F) + d(M, N)\\
d((\lambda x. M) E, (\lambda x . N) F) &\deq& d(M,N) + d(E,F)\\
d(E, F) &\deq& \infty \quad \text{otherwise}
\end{eqnarray*}

\begin{lemma} \label{lemma:metric-contexts}
$d(E[M],F[N]) \le d(E, F) + d(M,N)$.
\end{lemma}
\begin{proof}
By induction on the structure of $E$.

If $d(E, F) = \infty$, then the result is obvious, since $d(M', N') \le \infty$ for all $M', N'$.

Now let us assume $d(E, F) \neq \infty$ and prove the result
by simultaneous induction on the structure on $E$ and $F$:

\begin{itemize}
\item Case $E = F = \hole$: in this case, $E[M] = M$, $F[N] = N$, and $d(E,F) = 0$, so obviously
$d(E[M],F[N]) = d(E, F) + d(M,N)$ 

\item Case $E = E'\ L_1$, $F = F'\ L_2$:

We have $d(E[M], F[N]) = d(E'[M]\ L_1, F'[N]\ L_2) = d(E'[M], F'[N]) + d(L_1, L_2)$.
By induction hypothesis, $d(E'[M], F'[N]) \le d(E', F') + d(M,N)$, so
$d(E[M], F[N]) \le d(E', F') + d(M,N) + d(L_1, L_2) = d(E,F) + d(M,N)$.

\item Case $E = (\lambda x . L_1)\ E'$, $F = (\lambda x . L_2)\ F'$:

We have $d(E[M], F[N]) = d((\lambda x. L_1) (E'[M]), (\lambda x. L_2) (F'[N]))
= d(\lambda x. L_1, \lambda x. L_2) + d(E'[M], F'[N])$. By induction hypothesis, 
$d(E'[M], F'[N]) \le d(E', F') + d(M,N)$, so $d(E[M], F[N]) \le d(E',F') + d(\lambda x. L_1, \lambda x. L_2) + d(M,N)
= d(E,F) + d(M,N)$. \qed
\end{itemize}
\end{proof}

\begin{lemma} \label{lemma:metric-infty}
If $d(E,F) = \infty$, then for all $R_1$, $R_2$, $d(E[R_1], F[R_2])$.
\end{lemma}
\begin{proof}
By induction  on the structure of $E$:
\begin{itemize}
\item If $E = [\ ]$, then $d(E,F) = \infty$ implies $F \neq [\ ]$:
\begin{itemize}
\item If $F = (\lambda x . M)\ F'$, then $d(E[R_1], F[R_2]) = d(R_1, (\lambda  x . M)\ F'[R_2]) = \infty$,
because $R_1$ is either not an application or of the form $V_1\ V_2$, and $F'[R_2]$ is not a
value.
\item If $F = F'\ N$, then $d(E[R_1], F[R_2]) = d(R_1, F'[R_2]\ N) = \infty$, because
$R_1$ is either not an application or of the form $V_1\ V_2$, and $F'[R_2]$ is not a
value.
\end{itemize}
\item If $E = (\lambda x . M)\ E'$, then:
\begin{itemize}
\item If $F = F'\ N$, then $d(E[R_1], F[R_2]) = d(\lambda x. M, F'[R_2]) + d(E'[R_1], N) = \infty$,
because $d(\lambda x. M, F'[R_2]) \ = \infty$, as$F'[R_2]$ cannot be a lambda-abstraction.
\item If $F = (\lambda x . N)\ F'$, then $d(E,F) = \infty$ implies that either $d(M,N) = \infty$ or $d(E', F') = \infty$. We have
$d(E[R_1], F[R_2]) = d(M,N) + d(E'[R_1], F'[R_2])$. If $d(M,N) = \infty$, then obviously $d(E[R_1], F[R_2]) = \infty$.
Otherwise, by induction hypothesis, $d(E', F') = \infty$ gives $d(E'[R_1], F'[R_2]) = \infty$, and
so $d(E[R_1], F[R_2]) = \infty$.
\end{itemize}
\item If $E = E'\ M$ and $F = F'\ N$, then $d(E,F) = \infty$ implies that either $d(M,N) = \infty$ or $d(E', F') = \infty$. 
We have $d(E[R_1], F[R_2]) = d(M,N) + d(E'[R_1], F'[R_2])$, so $d(E'[R_1], F'[R_2]) = \infty$ follows like in the previous case.
\end{itemize}
The property also holds in all remaining cases by symmetry of $d$.
\qed
\end{proof}

\begin{lemma} \label{lemma:metric-contexts-redex}
$d(E[R_1],F[R_2]) = d(E, F) + d(R_1,R_2)$.
\end{lemma}
\begin{proof}
If $d(E, F) = \infty$, then $d(E[R_1],F[R_2])  = \infty$ by Lemma \ref{lemma:metric-infty},
otherwise the proof is the same as the proof of lemma \ref{lemma:metric-contexts},
with inequality replaced by equality when applying the induction hypothesis.
\qed
\end{proof}

\begin{lemma} \label{lemma:subst-bound}
$d(\mathbf{Subst}(M,x,V), \mathbf{Subst}(N,x,W)) \le
d(M,N) + k\cdot d(V,W)$ where k is the max of the multiplicities of x in M and N
\end{lemma}
\begin{proof}
By simultaneous induction on the structure of $M$ and $N$.
\qed
\end{proof}
%
%\begin{lemma} \label{lemma:subst-measurable}
%$\mathbf{Subst}(M,x,v)$ is continuous, and so measurable
%$(\measterms \times \varset \times \measterms) / \measterms$.
%\end{lemma}
%\begin{proof}
%Corollary of Lemma \ref{lemma:subst-bound}.
%\end{proof}

%Deterministic reduction:
%
%\begin{align*}
%\ct{\ectxone}{(\abstr{\varone}{\termone}) \valone}&\red\ct{\ectxone}{\mathbf{Subst}(M,x,V)}\\%{\sbst{\termone}{\varone}{\valone}}\\
%    \ct{\ectxone}{g({\vec{c}})}&\red\ct{\ectxone}{\intfun{g}(\vec{c})}\\
%    \ct{\ectxone}{\expone}&\red\expone\\
%    \ct{\ectxone}{\ite{\ttrue}{\termone}{\termtwo}}&\red\ct{\ectxone}{\termone}\\
%    \ct{\ectxone}{\ite{\tfalse}{\termone}{\termtwo}}&\red\ct{\ectxone}{\termtwo}%\\
%%    \ct{\ectxone}{\ertone}&\red\terror
%\end{align*}

%Simple expression.
%
%\begin{display}[0.2]{Deterministic reduction}
%$(\abstr{\varone}{\termone}) \valone\red^\epsilon \mathbf{Subst}(M,x,V)$\\%{\sbst{\termone}{\varone}{\valone}}\\
%    $g({\vec{c}}) \red^\epsilon \intfun{g}(\vec{c})$\\
%%    $\expone \red\expone$\\
%    $\ite{\ttrue}{\termone}{\termtwo}\red^\epsilon \termone $\\
%    $\ite{\tfalse}{\termone}{\termtwo}\red^\epsilon \termtwo $\\ \\
%%    \ct{\ectxone}{\ertone}&\red\terror
%\staterule{Red Ctx}
%{\termone \rightarrow^\epsilon \termtwo \quad M \neq \expone}
%{\ectxone[\termone] \rightarrow \ectxone[\termtwo]}\\ \\
%$E[\alpha] \red \alpha$
%\end{display}

%\begin{lemma} \label{lemma:det-head-red-measurable}
%The deterministic head reduction $M \red^\epsilon N$ is measurable.
%\end{lemma}

%\MS{$\mathcal{E}$ has two meanings}

Let $\mathcal{C}$ denote the set of contexts and $\mathcal{G}$ the set of primitive functions. Let:
\begin{itemize}
\item $\Lambda_{\mathit{appl}} \triangleq \{ E[(\lambda x . M) V] \mid E \in \mathcal{C} ,(\lambda x. M) \in \cterms, V \in \valset \}$
\item $\Lambda_{\mathit{applc}} \triangleq \{ E[c\ V] \mid E \in \mathcal{C} ,c\ \in \mathbb{R}, V \in \valset \}$
\item $\Lambda_{\mathit{iftrue}} \triangleq \{ E[\ite{\bf{true}}{M}{N}]\mid E \in \mathcal{C}, M, N \in \cterms \} $
\item $\Lambda_{\mathit{iffalse}} \triangleq \{ E[\ite{\bf{false}}{M}{N}]\mid E \in \mathcal{C}, M, N \in \cterms \} $
\item $\Lambda_{\mathit{fail}} \triangleq \{ E[\fail]\mid E \in \mathcal{C} \setminus \{[\ ]\} \} $
%\item $\Lambda_{error} \triangleq \{ E[T]\mid E \in \mathcal{C}, T \ \text{erroneous} \} $
\item $\Lambda_{\mathit{prim}}(g) \triangleq \{ E[g(\vec{c})]\mid E \in \mathcal{C}, \vec{c} \in \mathbb{R}^{|g|} \} $
\item $\Lambda_{\mathit{prim}} \triangleq \bigcup_{g \in \mathcal{G}} \Lambda_{\mathit{prim}}(g) $
%\item $P_{det} \triangleq P_A \cup P_T \cup P_F \cup P_E \cup P_G \cup P_{TT}$
\item $A\Lambda_{\mathit{if}} \triangleq \{ E[\ite{G}{M}{N}]\mid E \in \mathcal{C}, M, N \in \cterms, G \in \gvalset \} $
\item $\Lambda_{\mathit{dist}}(D) \triangleq \{ E[D(\vec{c})]\mid E \in \mathcal{C}, \vec{c} \in \mathbb{R}^{|D|}\} $
\item $\Lambda_{\mathit{dist}} \triangleq \bigcup_{D \in \mathcal{D}} P_{rnd}(D) $
\item $A\Lambda_{\mathit{prim}} \triangleq  \bigcup_{g \in \mathcal{G}} E[g(G_1, \dots, G_{|g|})]\mid E \in \mathcal{C}, G_1, \dots, G_{|g|} \in \gvalset \}$
\item $A\Lambda_{\mathit{dist}} \triangleq \bigcup_{D \in \mathcal{D}} E[D(G_1, \dots, G_{|D|})]\mid E \in \mathcal{C}, G_1, \dots, G_{|D|} \in \gvalset \}$
\item $A\Lambda_{\mathit{scr}}\triangleq \{E[\score{c}] \mid E \in \mathcal{C}, c \in \mathbb{R} \} $
\item $\Lambda_{\mathit{scr}} \triangleq \{E[\score{c}] \mid E \in \mathcal{C}, c \in (0,1] \}$
\end{itemize}

\begin{lemma} \label{lemma:partitions-measurable}
All the sets above are measurable.
%The sets $\Lambda_{\mathit{appl}}$, $\Lambda_{\mathit{iftrue}}$, $\Lambda_{\mathit{iffalse}}$, 
%$\Lambda_{\mathit{fail}}$, $\Lambda_{\mathit{prim}}$, $\Lambda_{\mathit{dist}}$ are Borel-measurable.
\end{lemma}
\begin{proof}
All these sets except for $\Lambda_{\mathit{scr}}$ are closed, so they are obviously measurable.
The set $\Lambda_{\mathit{scr}}$ is not closed (for example, we can define
a sequence of points in $\Lambda_{\mathit{scr}}$ whose limit is
$\mathtt{score}(0) \notin \Lambda_{\mathit{scr}}$), but it is still measurable:

Define a function $i_{\mathit{scr}}: A\Lambda_{\mathit{scr}} \rightarrow \mathbb{R}$  by
$i_{\mathit{scr}}(E[\score{c}]) = c$. This function is continuous and so measurable.
Since the interval $(0,1]$ is a Borel subset of $\mathbb{R}$, $i_{\mathit{scr}}^{-1}( (0,1] ) =
\Lambda_{\mathit{scr}}$ is measurable.
\qed
\end{proof}

Now, we need to define the set of erroneous redexes of all types.
\begin{itemize}
\item $R\Lambda_{\mathit{if}} \triangleq 
  A\Lambda_{\mathit{if}}
   \setminus (\Lambda_{\mathit{iftrue}} \cup \Lambda_{\mathit{iffalse}}))$
\item $R\Lambda_{\mathit{prim}} \triangleq 
 A\Lambda_{\mathit{prim}}
  \setminus \Lambda_{\mathit{prim}}$
\item $R\Lambda_{\mathit{dist}} \triangleq 
  A\Lambda_{\mathit{dist}}
  \setminus \Lambda_{\mathit{dist}}$
\item $R\Lambda_{\mathit{scr}} \triangleq A\Lambda_{\mathit{scr}} \setminus \Lambda_{\mathit{scr}}$
\item $\Lambda_{error} \triangleq R\Lambda_{\mathit{if}} \cup R\Lambda_{\mathit{prim}} \cup R\Lambda_{\mathit{dist}} \cup R\Lambda_{\mathit{scr}}$
\end{itemize}
\begin{lemma}
The set $\Lambda_{error}$ is measurable.
\end{lemma}
\begin{proof}
It is constructed from measurable sets by operations preserving measurability.
\qed
\end{proof}

Define:

$\Lambda_{det} = \Lambda_{\mathit{appl}} \cup \Lambda_{cappl} \cup \Lambda_{\mathit{iftrue}} \cup \Lambda_{\mathit{iffalse}} \cup \Lambda_{\mathit{fail}} \cup
\Lambda_{\mathit{prim}} \cup \Lambda_{error} $
\begin{lemma} \label{lemma:lambda-det-measurable}
$\Lambda_{det}$ is measurable.
\end{lemma}
\begin{proof}
$\Lambda_{det}$ is a union of measurable sets.
\qed
\end{proof}

\begin{lemma}
$\gvalset$ is measurable.
\end{lemma}
\begin{proof}
It is easy to see that $\gvalset$ is precisely the union of sets of all closed expressions of the form
$c$, $\lambda x . M$, $x$ and $\fail$, so it is closed, and hence measurable.
\qed
\end{proof}

\begin{lemma}
$\valset$ is measurable.
\end{lemma}
\begin{proof}
$\valset$ is the union of sets of all closed expressions of the form
$c$, $\lambda x . M$ and $x$ , so it is closed, and hence measurable.
\qed
\end{proof}

%\MS{Define product metric explicitly}

%\begin{proof}
%Define functions:
%\begin{itemize}
%\item $\mathtt{fun}(E,x,M,V) \triangleq E[(\lambda x.M) V]$
%\item $\mathtt{ift}(E,M,N) \triangleq E[\ite{\bf{true}}{M}{N}]$
%\item $\mathtt{iff}(E,M,N) \triangleq E[\ite{\bf{false}}{M}{N}]$
%\item $\mathtt{exc}(E,\alpha) \triangleq E[\alpha]$
%\item $\mathtt{det}_g(E,V_1, \dots, V_{|g|}) \triangleq E[g(V_1, \dots, V_{|g|}] $
%\end{itemize}
%
%It is easy to check that these functions are continuous with respect to trivial product metrics 
%(and hence Borel-measurable) and injective,
%and their images are respectively $P_A$, $P_T$, $P_F$, $P_E$, $P_G(g)$ so by Lemma \ref{lemma:image-measurable},
%the sets $P_A$, $P_T$, $P_F$, $P_E$, $P_G(g)$ are Borel measurable. Because the set of primitive functions
%$g$ is countable, the set $P_G$ is measurable, as a union of measurable sets.
%\end{proof}

%\subsection{Small- step reduction as a measurable function}

\subsection{Deterministic reduction as a measurable function}

Let us define a function performing one step of the reduction relation.
This function has to be defined piecewise.
Let us start with sub-functions reducing deterministic redexes of the given type.
\begin{eqnarray*}
g_{\mathit{appl}} &:& \Lambda_{\mathit{appl}} \rightarrow \cterms\\
g_{\mathit{appl}}(E[(\lambda x. M)\ V]) &=& E[\mathbf{Subst}(M,x,v)]
\end{eqnarray*}
\begin{lemma}
$g_{\mathit{appl}}$ is measurable.
\end{lemma}
\begin{proof}
By Lemma \ref{lemma:metric-contexts-redex}, we have
  $d(E[(\lambda x. M)V], F[(\lambda x.N)W]) = d(E,F) + d(M,N) + d(V,W)$
  and by Lemma \ref{lemma:subst-bound},
  $d(E[\mathbf{Subst}(M,x,V)], F[\mathbf{Subst}(N,x,W)]) \leq
   d(E,F) + d(M,N) + k \cdot d(V,W)$, where $k$ is the 
   maximum of the multiplicities of $x$ in $M$ and $N$.
   
  For any $\epsilon > 0$, take $\delta = \frac{\epsilon}{k+1}$. Then, if
  $d(E[(\lambda x. M)V], F[(\lambda x.N)W]) < \delta$, then
  \begin{eqnarray*}
  d(E[\mathbf{Subst}(M,x,V)], F[\mathbf{Subst}(N,x,W)]) &\le&
  d(E,F) + d(M,N) + k \cdot d(V,W)\\
  &\le& (k + 1) \cdot (d(E,F) + d(M,N) + d(V,W))\\
  &=& (k+1) \cdot d(E[(\lambda x. M)V], F[(\lambda x.N)W])\\
  &<& \epsilon
  \end{eqnarray*}
  
Thus, $g_{\mathit{appl}}$ is continuous, and so measurable. \qed
\end{proof}
\begin{eqnarray*}
g_{\mathit{applc}} &:& \Lambda_{\mathit{applc}} \rightarrow \cterms\\
g_{\mathit{applc}}(E[c\ M]) &=& E[\fail]
\end{eqnarray*}
\begin{lemma}
$g_{\mathit{applc}}$ is measurable.
\end{lemma}
\begin{proof}
It is easy to check that $g_{\mathit{applc}}$ is continous.
\qed
\end{proof}
\begin{eqnarray*}
g_{\mathit{prim}} &:& \Lambda_{\mathit{prim}} \rightarrow \cterms\\
g_{\mathit{prim}}(E[g(\vec{c})]) &=& E[\intfun{g}(\vec{c})]
\end{eqnarray*}
\begin{lemma}
$g_{\mathit{prim}}$ is measurable.
\end{lemma}
\begin{proof}
By assumption, every primitive function $g$ is measurable.
$g_{\mathit{prim}}$ is a composition of a function splitting a context and a redex,
$g$ and a function combining a context with a redex, all of which are measurable.
\qed
\end{proof}
\begin{eqnarray*}
g_{\mathit{iftrue}} &:& \Lambda_{\mathit{iftrue}} \rightarrow \cterms\\
g_{\mathit{iftrue}}(E[\ite{\ttrue}{M_1}{M_2}]) &=& E[M_1]\\
g_{\mathit{iffalse}} &:& \Lambda_{\mathit{iffalse}} \rightarrow \cterms\\
g_{\mathit{iffalse}}(E[\ite{\tfalse}{M_1}{M_2}]) &=& E[M_2]
\end{eqnarray*}
\begin{lemma}
$g_{\mathit{iftrue}}$ and $g_{\mathit{iffalse}}$ are measurable.
\end{lemma}
\begin{proof}
We have 
$d(E[\ite{\bf{true}}{M_1}{N_1}], F[\ite{\bf{true}}{M_2}{N_2}]) = d(E, F) + d(M_1,M_2) + d(N_1,N_2)
  \geq d(E[M_1],F[M_2])$,  so $g_{\mathit{iftrue}}$ is continuous, and so measurable, and similarly for 
  $g_{\mathit{iffalse}}$.
  \qed
\end{proof}
\begin{eqnarray*}
g_{\mathit{fail}} &:& \Lambda_{\mathit{fail}} \rightarrow \cterms\\
g_{\mathit{fail}}(E[\fail]) &=& \fail\\
\end{eqnarray*}
\begin{lemma}
$g_{\mathit{fail}}$ is measurable.
\end{lemma}
\begin{proof}
Obvious, since it is a constant function.
\qed
\end{proof}
\begin{eqnarray*}
g_{error} &:& \Lambda_{error} \rightarrow \cterms\\
g_{error}(E[T]) &=& E[\fail]\\
\end{eqnarray*}
\begin{lemma}
$g_{error}$ is measurable.
\end{lemma}
\begin{proof}
We have $d(E[T_1], F[T_2]) \geq d(E,F) = d(E[\fail],F[\fail])$, so
$g_{error}$ is continuous and hence measurable.
\qed
\end{proof}

%\begin{eqnarray*}
%g_{det} &:& P_{det} \times \mathbb{R} \times \mathbb{S} \rightarrow \Lambda \times \mathbb{R} \times \mathbb{S}
%\end{eqnarray*}
\begin{eqnarray*}
g_{det}' &:& \Lambda_{det} \rightarrow \cterms\\
g_{det}' &=& g_{\mathit{appl}} \cup g_{\mathit{applc}} \cup g_{\mathit{prim}} \cup g_{\mathit{iftrue}} \cup g_{\mathit{iffalse}} \cup g_{\mathit{fail}} \cup g_{error}  \\
\end{eqnarray*}
\begin{lemma}
$g_{det}'$ is measurable.
\end{lemma}
\begin{proof}
Follows directly from Lemma \ref{lemma:split-space}.
\qed
\end{proof}

\begin{lemma} \label{lemma:detred-function}
$M \detred N$ if and only if $g_{det}'(M) = N$.
\end{lemma}
\begin{proof}
By inspection.
\qed
\end{proof}

\subsection{Small- step reduction as a measurable function}

Let 
\[
\begin{array}{r@{}l@{}l}
\mathcal{T}_{\mathit{val}} = {}&\gvalset  \times \mathbb{R} \times \mathbb{S}\\
\mathcal{T}_{\mathit{det}} = {}&\Lambda_{\mathit{det}} \times \mathbb{R} \times \mathbb{S}\\
\mathcal{T}_{\mathit{scr}} = {}&\Lambda_{\mathit{scr}} \times  \mathbb{R} \times \mathbb{S}\\
\mathcal{T}_{\mathit{rnd}}= {}&\{ (E[D(\vec{c})], w, c \mathrel{::} s)\ |\ &E \in \mathcal{C}, D \in \mathcal{D}, 
                             \vec{c} \in \mathbb{R}^{|D|}, w \in \mathbb{R}, s \in \mathbb{S}, c \in \mathbb{R}, \\&&\pdf{D}(\vec{c},c) >0 \}
\end{array}
\]

\begin{lemma} \label{lemma:trnd-measurable}
$\mathcal{T}_{\mathit{val}}$, $\mathcal{T}_{\mathit{det}}$, $\mathcal{T}_{\mathit{scr}}$ and $\mathcal{T}_{\mathit{rnd}}$ are measurable.
\end{lemma}
\begin{proof}
The measurability of  $\mathcal{T}_{\mathit{val}}$, $\mathcal{T}_{\mathit{det}}$ and $\mathcal{T}_{\mathit{scr}}$ is obvious (they are products of measurable sets), so let us focus on $\mathcal{T}_{\mathit{rnd}}$.

For each distribution $D$, define a function $i_D: \Lambda_{\mathit{rnd}}(D) \times \mathbb{R} \times (\mathbb{S} \setminus \{[]\}) \rightarrow \mathbb{R}^{|D|} \times \mathbb{R}$
by $i_D(E[D(\vec{c})], w, c \mathrel{::} s) = (c, \vec{c})$. This function is continuous, and so measurable. Then, since for each
$D$, $\pdf{D}$ is measurable by assumption, the function $j_d = \pdf{D} \circ i_D$ is measurable.
Then, $\mathcal{T}_{\mathit{rnd}} = \bigcup_{D \in \mathcal{D}} j_D^{-1}((0, \infty))$, and since the set of distributions is countable,
$\mathcal{T}_{\mathit{rnd}}$ is measurable. 
\qed
\end{proof}

Let $\mathcal{T} = \cterms \times  \mathbb{R} \times \mathbb{S}$  and let 
$\mathcal{T}_{\mathit{blocked}} = \mathcal{T} \setminus (\mathcal{T}_{\mathit{val}} \cup \mathcal{T}_{\mathit{det}}
\cup \mathcal{T}_{\mathit{scr}} \cup \mathcal{T}_{\mathit{rnd}})$
be the set of non-reducible (``stuck'') triples,
whose first components are not values. Obviously, $\mathcal{T}_{\mathit{blocked}}$ is measurable.

Define:

\begin{eqnarray*}
g_{\mathit{val}} &:& \mathcal{T}_{\mathit{val}} \rightarrow \mathcal{T} \\
g_{\mathit{val}}(G,w,s) &=& (\fail, 0, [])
%g_{\mathit{val}}(G,w,s) &=& (G,w,s)
\end{eqnarray*}

Obviously, $g_{\mathit{val}}$ is measurable.

\begin{eqnarray*}
g_{\mathit{det}} &:& \mathcal{T}_{\mathit{det}} \rightarrow \mathcal{T} \\
g_{\mathit{det}}(M,w,s) &=& (g_{\mathit{det}}'(M),w,s)
\end{eqnarray*}

\begin{lemma}
$g_{\mathit{det}}$ is measurable.
\end{lemma}
\begin{proof}
All components of $g_{\mathit{det}}$ are measurable.
\qed
\end{proof}

\begin{eqnarray*}
g_{\mathit{rnd}} &:& \mathcal{T}_{\mathit{rnd}} \rightarrow \mathcal{T}\\
g_{\mathit{rnd}} &\deq& (g_1, g_2, g_3)\\
g_1(E[\distone(\vec{c})], w, c \mathrel{::} s ) &\deq& E[c]\\
g_2(E[\distone(\vec{c})], w, c \mathrel{::} s )  &\deq& w\cdot \pdf{\distone }(\vec{c}, c), \\
g_3(E[\distone(\vec{c})], w, c \mathrel{::} s ) &\deq& s
\end{eqnarray*}
\begin{lemma}
$g_{\mathit{rnd}}$ is measurable.
\end{lemma}
\begin{proof}
For $g_1$, we have $d(E[c], E'[c']) \le d(E,E') + d(c,c') \le d(E,E') + d(\vec{c},\vec{c'}) + d(w,w')
+ d(s,s') = d((E[\distone(\vec{c})], w, c \mathrel{::} s), (E'[\distone(\vec{c'})], w', c' \mathrel{::} s'))$ and
$d((E[\distone(\vec{c})], w, c \mathrel{::} s), (E'[\disttwo(\vec{c'})], w',c' \mathrel{::} s')) = \infty$ if $\distone \neq \disttwo$,
so $g_1$ is continuous and hence Borel-measurable.

For $g_2$, we have $g_2(E[\distone(\vec{c})], w, c \mathrel{::} s) = 
g_w(E[\distone(\vec{c})], w, c \mathrel{::} s) \times (\pdf{\distone } \circ g_c)(E[\distone(\vec{c})], w,c \mathrel{::} s )$,
where $g_w(E[\distone(\vec{c})], w, c \mathrel{::} s) = w$ and $g_c(E[\distone(\vec{c})], w, c \mathrel{::} s) = (\vec{c}, c)$.
The continuity (and so measurability) of $g_w$ and $g_c$ can be easily checked (as for $g_1$ above). Thus,
$\pdf{\distone } \circ g_c$ is a composition of measurable functions (since distributions are assumed to 
be measurable), and so $g_2$ is a pointwise product of measurable real-valued functions, so it is measurable.

The continuity (and so measurability) of $g_3$ can be shown in a similar way to $g_1$.

Hence, all the component functions of $g_{\mathit{rnd}}$  are measurable, so $g_{\mathit{rnd}}$ is itself measurable.
\qed
\end{proof}

\begin{eqnarray*}
g_{\mathit{scr}} &:& \mathcal{T}_{\mathit{scr}} \rightarrow \mathcal{T}\\
g_{\mathit{scr}}(E[\score{c}], w, s ) &\deq& (E[\ttrue], c\cdot w, s)
\end{eqnarray*}
\begin{lemma}
$g_{\mathit{scr}}$ is measurable.
\end{lemma}
\begin{proof}
The first component function of $g_{\mathit{scr}}$ can easily be shown continuous, and so measurable, and ditto for the third component.
The second component is a pointwise product of two measurable functions, like in the $g_{\mathit{rnd}}$ case. Hence, $g_{\mathit{scr}}$ is measurable.
%Let us now concentrate on the second component, which is a function $g_{scr2}: \mathcal{T}_{\mathit{scr}} \rightarrow \mathbb{R}$
%defined as $g_{scr2}(E[\score{c}], w, s ) = cw$.
\qed
\end{proof}

For completeness, we also define:

\begin{eqnarray*}
g_{\mathit{blocked}} &:& \mathcal{T}_{\mathit{blocked}} \rightarrow \mathcal{T}\\
g_{\mathit{blocked}}(M, w, s ) &\deq& (\fail, 0, [])
\end{eqnarray*}

This function is trivially measurable.

Define
\begin{eqnarray*}
g &:& \mathcal{T} \rightarrow \mathcal{T}\\
g &\deq& g_{\mathit{val}} \cup g_{\mathit{det}} \cup g_{\mathit{scr}} \cup g_{\mathit{blocked}}
\end{eqnarray*}
\begin{lemma}
$g$ is measurable.
\end{lemma}
\begin{proof}
Follows from Lemma \ref{lemma:split-space}. \qed
\end{proof}
\begin{lemma} \label{lemma:red-g}
For every $(M, w, s) \in \mathcal{T}$, 
\begin{enumerate}
\item If $(M,w,s) \red (M',w',s')$, then $g(M,w,s) = (M',w',s')$.
\item If $g(M,w,s) = (M',w',s') \neq (\fail, 0, [])$ , then $(M,w,s) \red (M',w',s')$.
\end{enumerate}
%either $(M, w, s) \red g(M,w,s)$ or
%$g(M,w,s) = (\fail, 0, [])$.
\end{lemma}
\begin{proof}
By inspection. \qed
\end{proof}

\subsection{Measurability of $\mathbf{P}$ and $\mathbf{O}$}

It is easy to check that the sets $\gvalset$ and $\mathbb{R}_{+}$ (nonnegative reals) %$\cterms$ and $\mathbb{R}$
form $\oCPO$s with the orderings $\fail \leq M$ for all $M$ and 
$0 \leq x$, respectively. This means that functions into
$\gvalset$ and $\mathbb{R}_{+}$ also form $\oCPO$s with pointwise ordering.

Define:

\[
\Theta_{\Lambda}(f)(M,w,s) \triangleq
\begin{cases}
M & \text{if}\ M \in \gvalset, s = []\\
f(g(M,w,s)) & \text{otherwise}
\end{cases}
\]

\[
\Theta_{w}(f)(M,w,s) \triangleq
\begin{cases}
w & \text{if}\ M \in \gvalset, s = []\\
f(g(M,w,s)) & \text{otherwise}
\end{cases}
\]

It can be shown that these functions are continuous, so we can define:
\[
\bot_\Lambda = (M,w,s)\mapsto \fail
\]
\[
\bot_w = (M,w,s)\mapsto 0
\]
\newcommand{\psup}[2]{\mathbf{P}'({#1},{#2})}
\newcommand{\osup}[2]{\mathbf{O}'({#1},{#2})}
\[
\osup{M}{s} \triangleq \sup_n\ \Theta^n_{\Lambda}(\bot_\Lambda)(M,1,s)
\]
\[
\psup{M}{s} \triangleq \sup_n\ \Theta^n_{w}(\bot_w)(M,1,s)
\]

%\begin{lemma}
%For any $M$, $s$, either $\psup{M}{s} = 0$ and $\osup{M}{s} = \fail$ or
%$(M,1,s) \Rightarrow (\osup{M}{s}, \psup{M}{s}, [])$.
%\end{lemma}

\begin{lemma}
If $(M,w_0,s) \Rightarrow (G,w,[])$, then $\sup_n\ \Theta^n_{w}(\bot_w)(M,w_0,s) = w$ and
$\sup_n\ \Theta^n_{\Lambda}(\bot_\Lambda)(M,w_0,s) = G$.
\end{lemma}
\begin{proof}
By induction on the derivation of $(M,w_0,s) \Rightarrow (G,w,[])$:
\begin{itemize}
\item If $(M,w_0,s) \red^0 (G,w,[])$, and so $M \in \gvalset$ and $s = []$, then
the equalities follow directly from the definitions of $\Theta_w$ and $\Theta_\Lambda$. %$\mathbf{P}'$ and $\mathbf{O}'$.
\item If $(M,w_0,s) \red (M',w',s') \Rightarrow (G,w,[])$, assume that
$\sup_n\ \Theta^n_{w}(\bot_w)(M',w',s') = w$ and
$\sup_n\ \Theta^n_{\Lambda}(\bot_\Lambda)(M',w',s') = G$, We have $M \notin \gvalset$. 
By Lemma \ref{lemma:red-g}, $g(M,w_0,s) = (M',w',s')$.
Hence
$\sup_n\ \Theta^n_{w}(\bot_w)(M,w_0,s) = \sup_n\ \Theta^n_{w}(\bot_w)(g(M,w_0,s))
= \sup_n\ \Theta^n_{w}(\bot_w)(M',w',s') = w$ by induction hypothesis.
Similarly, $\sup_n\ \Theta^n_{\Lambda}(\bot_\Lambda)(M,w_0,s) = G$. \qed
\end{itemize}
\end{proof}
\begin{corollary}
If $(M,1,s) \Rightarrow (G,w,[])$, then $\psup{M}{s} = w$ and
$\osup{M}{s} = G$.
\end{corollary}

\begin{lemma} \label{lemma:red-closure-w}
If $\sup_n\ \Theta^n_{w}(\bot_w)(M,w_0,s) =w \neq 0$, then
$(M,w_0,s) \Rightarrow (G,w,[])$ for some $G \in \gvalset$.
\end{lemma}
\begin{proof}
Because the supremum is taken with respect to a flat $\oCPO$, 
$\sup_n\ \Theta^n_{w}(\bot_w)(M,w_0,s) =w>0$ implies
$\Theta^k_{w}(\bot_w)(M,w_0,s) =w$ for some $k>0$. We can then prove
the result by indiction on $k$:
\begin{itemize}
\item Base case, $k=1$:
We must have $\Theta_{w}(\bot_w)(M,w_0,s) =w_0$, $M = G \in \gvalset$ and $ s = []$
as otherwise we would obtain $\bot_w(M,w_0,s)=0$. Hence $(M,w_0,s)$ reduces to $(G,w_0,[])$ in $0$ steps.
\item Induction step: $\Theta^{k+1}_{w}(\bot_w)(M,w_0,s) =w$. If $M \in \gvalset$ and $s = []$, then $w = w_0$ and $(M,w_0,s)$ reduces to
itself in $0$ steps, like in the base case. Otherwise, we have
$\Theta^k_{w}(\bot_w)((M',w',s')) =w$, where $g(M,w_0,s) = (M',w',s')$.
We know that $(M',w',s') \neq (\fail, 0, [])$, because otherwise we
would have $w = 0$.
Thus, by Lemma \ref{lemma:red-g}, $(M,w_0,s) \red (M',w',s')$.
By induction hypothesis, $(M',w',s') \Rightarrow (G,w, [])$, which
implies $(M,w_0,s) \Rightarrow (G,w,[])$. \qed 
\end{itemize}
\end{proof}

\begin{lemma} \label{lemma:red-closure-G}
If $\sup_n\ \Theta^n_{\Lambda}(\bot_w)(M,w_0,s) =V \in \valset$, then
$(M,w_0,s) \Rightarrow (V,w,[])$ for some $w \in \mathbb{R}$.
\end{lemma}
\begin{proof}
Similar to the proof of Lemma \ref{lemma:red-closure-w}.
\qed
\end{proof}

\begin{corollary}
If there are no $G,w$ such that $(M,1,s) \Rightarrow (G,w,[])$,
then $\psup{M}{s} = 0$ and $\osup{M}{s} = \fail$.
\end{corollary}

%\begin{proof}
%By induction on the derivation of $(M,1,s) \Rightarrow (\osup{M}{s}, \psup{M}{s}, [])$
%\end{proof}

\begin{corollary}
For any $M$, $\mathbf{P}_M = \psup{M}{\cdot}$ and $\mathbf{O}_M = \osup{M}{\cdot}$.
\end{corollary}

\begin{lemma} \label{lemma:sup-flat-measurable}
If $(X,\Sigma_1)$ and $(Y, \Sigma_2)$ are measurable spaces,
$Y$ forms a flat $\oCPO$ with a bottom element $\bot$ such that $\{\bot\} \in \Sigma_2$
and $f_1, f_2, \dots$ is a $\omega$-chain of $\Sigma_1/\Sigma_2$ measurable functions
(on the $\oCPO$ with pointwise ordering),
then $\sup_i f_i$ is $\Sigma_1/\Sigma_2$ measurable.
\end{lemma}
\begin{proof}
%Let $\bot$ be the bottom element of $Y$. 
Since $f^{-1}(A \cup \{\bot\}) = f^{-1}(A) \cup f^{-1}(\{\bot\})$,
we only need to show that $(\sup_i f_i)^{-1}(\{\bot\}) \in \Sigma_1$ and $(\sup_i f_i)^{-1}(A) \in \Sigma_1$
for all $A \in \Sigma_2$ such that $\bot \notin A$.

We have $(\sup_i f_i)^{-1}(\{\bot\}) = \bigcap_i f_i^{-1}(\{\bot\}) $, which is measurable by definition.
If $\bot \notin A$, then $\sup_i f_i(x) \in A$ if and only if $f_i(x) \in A$ for some $i$,
so by extensionality of sets, $\sup_i f_i^{-1}(A) = \bigcup_i f_i^{-1}(A) \subseteq \Sigma_1$.
\qed
\end{proof}

\begin{lemma} \label{lemma:p-prim-measurable}
$\mathbf{P}'$ is measurable $(\cterms \times \mathcal{S}) / \restr{\mathcal{R}}{\mathbb{R}_{+}}$. %$\cterms \times \mathcal{S} / \mathbb{R}_{+}$.
\end{lemma}
\begin{proof}
First, let us show by induction on $n$ that $\Theta_w^n(\bot_w)$ is measurable for
every $n$:
\begin{itemize}
\item Base case, $n=0$: $\Theta_w^0(\bot_w) = \bot_w$ is a constant function,
and so trivially measurable.
\item Induction step: suppose $\Theta_w^n(\bot_w)$ is measurable. Then we have
$\Theta_w^{n+1}(\bot_w) = \Theta_w(\Theta_w^n(\bot_w))$, so it is enough to
show that $\Theta_w(f)$ is measurable if $f$ is measurable:

%First, let us show that $\Theta(\bot_w)$ is measurable.

The domain of the first case is $\gvalset \times \mathbb{R} \times \{[ ]\}$, which is
clearly measurable. The domain of the second case is measurable as
the complement of the above set in $\mathcal{T}$.

The sub-function corresponding to the first case returns the second component of
its argument, so it is continuous and hence measurable. The second case is a composition
of two measurable functions, hence measurable.

Thus, $\Theta_w(f)$ is measurable for any measurable $f$, and so $\Theta^{n+1}_w(\bot_w)$
is measurable.
\end{itemize}
%Thus, $\Theta_w(\bot_w)$ is measurable. Since a composition of measurable functions
%is measurable, $\Theta_w^n(\bot_w)$ is measurable for all $n$. 
By Lemma
\ref{lemma:sup-flat-measurable}, $\sup_n \Theta^n_w(\bot_w)$ is measurable.
Since $\mathbf{P}'$ is a composition of $\sup_n \Theta^n_w(\bot_w)$
and a continuous function mapping $(M,s)$ to $(M,1,s)$, it is a composition
of measurable functions, and so it is measurable.
\qed
\end{proof}

\begin{lemma}
$\mathbf{O}'$ is measurable $(\cterms \times \mathcal{S}) / \restr{\measterms}{\gvalset}$.%  $\cterms \times \mathcal{S} / \gvalset$.
\end{lemma}
\begin{proof}
Similar to the proof of Lemma \ref{lemma:p-prim-measurable}.
\qed
\end{proof}

\begin{lemma} \label{lemma:pi-measurable}
For any closed term $M$, the function $\mathbf{P}_{M}$ is measurable $\mathcal{S} / \restr{\mathcal{R}}{\mathbb{R}_{+}}$.% $\sampseq \to \mathbb{R}_{+}$.
\end{lemma}
\begin{proof}
Since $\mathbf{P'}$ is measurable, $\mathbf{P}_M = \mathbf{P'}(M, \cdot)$ is measurable for every $M \in \cterms$.
\qed
\end{proof}

\begin{lemma} \label{lemma:o-measurable}
For each $M$, the function $\mathbf{O}_M$ is measurable $\mathcal{S} / \restr{\measterms}{\gvalset}$. %$\sampseq \to \gvalset$.
\end{lemma}
\begin{proof}
Since $\mathbf{O'}$ is measurable, $\mathbf{O}_M = \mathbf{O'}(M, \cdot)$ is measurable for every $M \in \cterms$.
\qed
\end{proof}

\begin{lemma} \label{lemma:pv-rewrite}
For all $M$, $s$, $\mathbf{P}_M^\valset(s) = \mathbf{P}_M(s) [ \mathbf{O}_M(s) \in \valset]$ 
\end{lemma}
\begin{proof}
By Lemma \ref{lemma:big-step-determined}, if $M \Downarrow^w_s G$, then $w$, $G$ are unique.
If $M \Downarrow^w_s V$, then $\mathbf{P}_M(s) = w$, $\mathbf{P}_M^\valset(s)$ and $ \mathbf{O}_M(s) \in \valset$,
so the equality holds. If $M \Downarrow^w_s \fail$, then $\mathbf{P}_M(s) = w$, $\mathbf{P}_M^\valset(s) = 0$
and $\mathbf{O}_M(s) \notin \valset)$, so both sides of the equation are $0$. If there is no $G$ such
that $M \Downarrow^w_s G$, then both sides are also $0$.
\qed
\end{proof}

\begin{lemma} \label{lemma:p-v-measurable}
$\mathbf{P}_M^\valset$ is measurable for every $M$.
\end{lemma}
\begin{proof}
By Lemma \ref{lemma:pv-rewrite}, $\mathbf{P}_M^\valset(s) = \mathbf{P}_M(s) [ \mathbf{O}_M(s) \in \valset]$ ,
so $\mathbf{P}_M^\valset$ is a pointwise product of a measurable function and 
a composition of $\mathbf{O}_M$ and an indicator function
for a measurable set, hence it is measurable.
\qed
\end{proof}

\begin{restate}{Lemma~\ref{lemma:pi-measurable}}
  For any closed term $M$, the functions $\mathbf{P}_{M}$, $\mathbf{O}_M$ and~$\mathbf{P}_{M}^\valset$ are all measurable; 
$\Tracedist{\termone}$ and $\Tracevaldist\termone$ are measures on $(\sampseq, \mathcal{S})$; 
$\sbd{\termone}$ is a measure on $(\gvalset,\restr\measterms\gvalset)$; 
and $\restr{(\sbd{\termone})}\valset$ is a measure on $(\valset,\restr\measterms\valset)$.
\end{restate}
  \begin{proof}
For all $M \in \cterms$, $\mathbf{P}_{M}$, $\mathbf{O}_M$ and $\mathbf{P}_{M}^\valset$ are measurable by lemmas 
\ref{lemma:pi-measurable}, \ref{lemma:o-measurable} and \ref{lemma:p-v-measurable}, respectively. Since
$\mathbf{P}_{M}$, and $\mathbf{P}_{M}^\valset$ are obviously nonnegative, the functions
$\Tracedist{\termone}$ and $\Tracevaldist\termone$ are measures of densities $\mathbf{P}_{M}$ and
$\mathbf{P}_{M}^\valset$ with respect to the stock measure $\mu$ \cite[Section 2.3.3]{gallay2009mesure}.
The function $\sbd{\termone}$ is a transformation of
the measure $\Tracedist{\termone}$ on  $(\sampseq, \mathcal{S})$  by the $\mathcal{S} / \restr{\measterms}{\gvalset}$ 
-measurable function $\mathbf{O}_{M}$, so it is a measure on $(\gvalset,\restr\measterms\gvalset)$
 \cite[Section 13, Transformations of Measures]{billingsley95}.
%
%The functions $\restr{(\sbd{\termone})}\valset$ is a restriction of $\sbd{\termone}$ to $\valset$,
%so it is a measure on $(\valset, \restr{\measterms}{\valset})$.
%
Since %we can easily show that 
$ \restr{(\sbd{\termone})}{\valset}(\tsetone) = 
\sbd{\termone}(\tsetone \cap \valset)$ for every measurable set $\tsetone \in \restr\measterms\gvalset $,
the function $\restr{(\sbd{\termone})}\valset$ is a restriction of $\sbd{\termone}$ to $\valset$,
so it is a measure on $(\valset, \restr{\measterms}{\valset})$.
%\int\termsmv{\termone}{s}\cdot[\obssmp{\termone}{s}\in\tsetone]\,ds$,

%  \restr{(\sbd{\termone})}{\valset}(\tsetone)%:={}&\sbd{\termone}(\tsetone\cap\valset)
%:= \Tracevaldist M(\mathbf{O}^{-1}_M(A)) = \int\termsmv{\termone}{s}\cdot[\obssmp{\termone}{s}\in\tsetone]\,ds
%
%The functions $\sbd{\termone}$ and $\restr{(\sbd{\termone})}\valset$ are transformations of
%measures on  $(\sampseq, \mathcal{S})$  by $\mathcal{S} / \measterms$ measurable functions, so they
%are measures on $(\cterms,\measterms)$ \cite[Section 13, Transformations of Measures]{billingsley95}.
%
%$(\gvalset,\restr\measterms\gvalset)$
%Functions $(\gvalset,\restr\measterms\gvalset)$ and $\restr{(\sbd{\termone})}\valset$, meanwhile,
%are measures of densities $\mathbf{P}_{M}$ and $\mathbf{P}_{M}^\valset$ 
%(both of which are obviously measurable, as products of measurable functions and indicator functions
%of measurable sets).
%with respect to $\mu$.
  \end{proof}

\subsection{Measurability of $\mathtt{peval}$}

Like in the previous section, we start by giving an alternative definition of $\mathtt{peval}$,
using the function $g$ instead of referring to the reduction relation directly.

%Let: 

%For brevity, we write $\lambda (M, s) .\ \fail$ as $\bot^\lambda$ below.

%Recall the defintion of $\mathtt{peval}$:

%\[
%\Phi(f)(M,s) =
%\begin{cases}
%M & \text{if}\ s = []\\
%f(M',s') & \text{if}\ s \neq [], (M,1,s) \red (M',w,s') \\
%%f(M',s) & \text{if}\ M \detred M'\\
%%f(E[c], s') & \text{if}\ M = E[\distone(\underline{c})], s = [c] @ s',\pdf{\distone}(\underline{c},c) > 0\\
%\fail & \text{otherwise}
%\end{cases}
%\]

%\[
%  \mathtt{peval}^k = \Phi^k(\bot^\lambda)
%\]
%
%and

%\[
%  \mathtt{peval} = \sup_k \Phi^k(\bot^\lambda)
%\]

%\begin{LONG}
%Because of the possibility of nontermination, it is convenient to define 
%such a function as a fixpoint of a higher-order function, since this simplifies 
%the proof of measurability.
%\end{LONG}

The set of closed terms $C\Lambda$ is a $\oCPO$ with respect to the partial order defined by
$\fail \leq G$ for all $G$. Hence the set $\funs$ of all functions
$(C\Lambda \times \mathbb{R} \times \mathbb{S}) \rightarrow C\Lambda$ is a $\oCPO$ with respect to the pointwise order.
Define $\Phi : \funs \rightarrow \funs$ as:

\[
\Phi(f)(M,w,s) =
\begin{cases}
M & \text{if}\ s = []\\
f(g(M,w,s)) & \text{otherwise} \\
\end{cases}
\]

It is easy to check that $\Phi$ is monotone and preserves suprema of $\omega$-chains,
so it is continuous. Hence, we can define:

\[
  \mathtt{peval}'(M,s) = \sup_k \Phi^k(\bot_\Lambda)(M,1,s)
\]

where $\bot_\Lambda(M,w,s) = \fail$, as before.

%\begin{lemma}
%If $$
%\end{lemma}
%\begin{proof}
%Trivial. \qed
%\end{proof}

\noindent
We first need to show that the original $\mathtt{peval}$ function is well-defined.
\begin{lemma} \label{lemma:peval-wd}
If $(M,w_0,s) \Rightarrow (M_k,w_k,s_k) \red (M', w', [])$ and $s_k \neq []$
and $(M,w_0,s) \Rightarrow (M_l,w_l,s_l) \red (M'', w'', [])$ and $s_l \neq []$,
then $M'=M''$ and $w'=w''$.
\end{lemma}
\begin{proof}
By induction on the derivation of $(M,w_0,s) \Rightarrow (M_k,w_k,s_k)$:
\begin{itemize}
\item If  $(M,w_0,s) \Rightarrow (M_k,w_k,s_k)$ was derived in $0$ steps,
we have $M_k=M$, $w_k=w$ and $s_k=s$, and so 
$(M,w_0,s) \red (M', w', [])$, where $s \neq []$.

If $(M,w_0,s) \Rightarrow (M_l,w_l,s_l)$ was derived in $0$ steps, then
$(M_l,w_l,s_l) = (M,w_0,s)$, and so $M'' = M'$ and $w''=w'$ by 
Lemma \ref{lemma:small-step-determined}.

If $(M,w_0,s) \Rightarrow (M_l,w_l,s_l)$ was derived in $1$ or more steps, we have
$(M,w_0,s)  \red (\hat{M},\hat{w},\hat{s}) 
\Rightarrow (M_l,w_l,s_l) \red (M'', w'', [])$ and $s_l \neq []$,
for some $\hat{M}$, $\hat{w}$, $\hat{s}$. By Lemma \ref{lemma:small-step-determined},
$\hat{s} = []$. We have $ (\hat{M},\hat{w}, []) 
\Rightarrow (M_l,w_l,s_l)$, where $s_l \neq []$.
This leads to a contradiction,
as it is easy to show that reducing a term with an empty trace cannot
yield a triple with a non-empty trace (there is no rule which adds an element 
to a trace)

\item If  $(M,w_0,s) \Rightarrow (M_k,w_k,s_k)$ was derived in $1$ on more steps,
we have $(M,w_0,s) \red (M^*,w^*,s^*) \red^k (M_k,w_k,s_k) \red (M', w', [])$
for some $k \ge 0$, $M^*$, $w^*$, $s^*$. Now, if 
$(M,w_0,s) \Rightarrow (M_l,w_l,s_l)$ was derived in $1$ or more steps, we have
$(M,w_0,s) \red (\hat{M}, \hat{W}, \hat{s}) \Rightarrow (M_l,w_l,s_l) \red (M'', w'', [])$ and $s_l \neq []$
for some $\hat{M}$, $\hat{w}$, $\hat{s}$, where
$(\hat{M}, \hat{W}, \hat{s}) =  (M^*,w^*,s^*)$ by Lemma 
\ref{lemma:small-step-determined}.
Hence, the result follows by the induction hypothesis.

If $(M,w_0,s) \Rightarrow (M_l,w_l,s_l)$ was derived in $0$ steps, then
$(M_l,w_l,s_l) = (M,w_0,s)$, and so $(M,w_0,s) \red (M'', w'', [])$. By
Lemma \ref{lemma:small-step-determined}, this implies $s^* = []$,  so
$(M^*,w^*,[]) \Rightarrow (M_k,w_k,s_k)$ for $s_k \neq []$, which is 
impossible, as explained in the previous case.
\end{itemize}
\end{proof}

\begin{lemma} \label{lemma:peval-to-prim}
If $(M,w_0,s) \Rightarrow (M_k,w_k,s_k) \red (M', w', [])$ and $s_k \neq []$, then
$\sup_n \Phi^n (\bot_\Lambda)(M,w_0,s) = M'$.
\end{lemma}
\begin{proof}
By induction on the length of derivation of $(M,w_0,s) \Rightarrow (M_k,w_k,s_k) \red (M', w', [])$.
Suppose $(M,w_0,s) \red^k (M_k,w_k,s_k) \red (M', w', [])$.
\begin{itemize}
\item Base case, $k=0$: We have $(M,w_0,s) \red (M', w', [])$ and $s \neq []$. Hence,
by Lemma \ref{lemma:red-g}, $g(M,w_0,s) = (M',w',[])$, and so,
by monotonicity of $\Phi$,
$\sup_k \Phi^k(\bot_\Lambda)(M,w_0,s) = \sup_k \Phi(\Phi(\Phi^k(\bot_\Lambda)))(M,w_0,s) =
\sup_k \Phi(\Phi^k(\bot_\Lambda))(M',w',[])= M'$, as required.
\item Induction step: Let $(M,w_0,s) \red^{k+1} (M_k,w_k,s_k) \red (M', w', [])$.
Then there exist $M^*$, $w^*$, $s^*$ such that $(M,w_0,s) \red (M^*, w^*, s^*) \red^{k} (M_k,w_k,s_k) \red (M', w', [])$.
Now, we have $\sup_k \Phi^k(\bot_\Lambda)(M,w_0,s) = \sup_k \Phi(\Phi^k(\bot_\Lambda))(M,w_0,s) =
\sup_k \Phi^k(\bot_\Lambda)(M^*,w^*,s^*) $, and $\sup_k \Phi^k(\bot_\Lambda)(M^*,w^*,s^*)  = M'$
by induction hypothesis, which ends the proof.
\end{itemize} \qed
\end{proof}

\begin{corollary} \label{corr:peval-to-prim}
If $(M,1,s) \Rightarrow (M_k,w_k,s_k) \red (M', w', [])$ and $s_k \neq []$, then
$\mathtt{peval}'(M,s) = M'$.
\end{corollary}

\begin{lemma} \label{lemma:prim-to-peval}
If $\sup_n \Phi^n (\bot_\Lambda)(M,w_0,s) = M' \neq \fail$,
then either $s=[]$ or $(M,w_0,s) \Rightarrow (M_k,w_k,s_k) \red (M', w', [])$
for some $M_k$, $w_k$, $s_k$, $w'$, where $s_k \neq []$.
\end{lemma}
\begin{proof}
Like in lemma \ref{lemma:red-closure-w}, for every $M$, $w_0$, $s$, we must have
$\Phi^k (\bot_\Lambda)(M,w_0,s) =M'$ for some $k>0$, and we can prove the result
by induction on $k$.
\begin{itemize}
\item Base case. $k=1$: we must have $s=[]$ as otherwise we would have $M' = \fail$.
\item Induction step: suppose $\Phi^{k+1} (\bot_\Lambda)(M,w_0,s) =M'$. By definition of 
$\Phi$, if $s \neq []$, we have $\Phi^k (\bot_\Lambda)(M^*, w^*, s^*)=M'$, where
$g(M,w_0,s) \red (M^*, w^*, s^*)$. Since $M' \neq \fail$ by assumption, 
%implicit lemma?e-
Lemma \ref{lemma:red-g} yields $(M,w_0,s) \red (M^*, w^*, s^*)$.
By induction hypothesis, either $s^* = []$ or $(M^*,w^*,s^*) \Rightarrow (M^{**}, w^{**},s^{**}) \red (M', w'', [])$
for some $M^{**}$, $w^{**}$, $s^{**}$, $w''$, where $s^{**} \neq []$.
In the former case, we have $(M,w_0,s) \Rightarrow (M_k,w_k,s_k) \red (M', w', [])$
with $(M,w_0,s) = (M_k,w_k,s_k) $,
$(M^*,w^*,s^*) = (M', w', [])$ and $s_k \neq []$, as required.
In the latter case, we have  $(M,w_0,s) \Rightarrow (M^{**}, w^{**},s^{**}) \red (M', w'', [])$,
with $s^{**} \neq []$.

\end{itemize}
\qed
\end{proof}

\begin{lemma} \label{lemma:peval-two-def-eqiv}
$\mathtt{peval} = \mathtt{peval}'$
\end{lemma}
\begin{proof}
We need to show that $\mathtt{peval}(M,s) = \mathtt{peval}'(M,s)$ for all $M \in \cterms$, $s \in \mathbb{S}$.

If $s =[]$, then the equality follows trivially from the two definitions. Now, assume $s \neq []$. 

If $\mathtt{peval}'(M,s) =M' \neq \fail$, then it follows from Lemma \ref{lemma:prim-to-peval} that $\mathtt{peval}(M,s) = M'$,

Now, let $\mathtt{peval}'(M,s) = \fail$ and suppose that $\mathtt{peval}(M,s) = M' \neq \fail$. Since $s \neq []$, 
by definition of $\mathtt{peval}$ there must be $M_k$, $w_k$, $s_k$, $w'$ such that 
$(M,1,s) \Rightarrow (M_k,w_k,s_k) \red (M', w', [])$ and $s_k \neq []$. But by Corollary \ref{corr:peval-to-prim},
this implies that $\mathtt{peval}'(M,s) = M' \neq \fail$, which yields a contradiction. Hence $\mathtt{peval}(M,s) = \fail$.
\qed
\end{proof}
%
%\begin{corollary} \label{corr:peval-two-def-eqiv}
%$\mathtt{peval} = \mathtt{peval}'$
%\end{corollary}

\begin{lemma} \label{lemma:phi-measurable}
For every $k$, $\mathtt{peval}_k = \Phi^k(\bot^\lambda)$ is measurable.
\end{lemma}
\begin{proof}
By induction on $k$:
\begin{itemize}
\item Base case: $k = 0$:
$\mathtt{peval}_0 = \bot^\lambda$ is a constant function on $\cterms \times \mathbb{S}$, so trivially measurable.
\item Induction step : we have $\mathtt{peval}_{k+1} = \Phi(\mathtt{peval}_k)$, so it is enough
to show that $\Phi(f)$ is measurable if $f$ is measurable. $\Phi(f)$ is defined in pieces,
so we want to use Lemma \ref{lemma:split-space}.

The domain of the first case is $\cterms \times \{[]\}$, so obviously measurable.
The domain of the second case is 
$p^{-1}(g^{-1}(\cterms \times \mathbb{R} \times \mathbb{S}) \cap (\cterms \times \{ 1\} \times (\mathbb{S}\setminus \{[] \} )))$,
and $p(M,s) = (M,1,s)$ is continuous, and so measurable. Hence, the domain is measurable.
Finally, the domain of the last case is the complement of the union of the two
above measurable sets, which means it is also measurable. 

Thus, we only need to show that the functions corresponding to these three cases are
measurable. This is obvious in the first and third case, because the corresonding
functions are constant. The function for the second case is $\phi(M,s) = f(g(p(M,s)))$,
where $p$ is as defined above and  $g'$ is the restriction of $g$ to $g^{-1}(\cterms \times \mathbb{R} \times \mathbb{S})$,
which is measurable since restrictions preserve measrability.
Since composition of measurable functions is measurable, $\phi$ is measurable.

Thus, $\mathtt{peval}_{k+1} $ is measurable, as required. \qed
\end{itemize}
\end{proof}

%Note that $\mathtt{peval} = \sup_k \mathtt{peval}^k$,
%and the following equality holds for all 

\begin{lemma} \label{lemma:peval-prim-measurable}
$\mathtt{peval}'$ is a measurable function $\cterms  \times  \sampseq \rightarrow  \cterms$.
\end{lemma}
\begin{proof}
Corollary of Lemmas \ref{lemma:phi-measurable} and \ref{lemma:sup-flat-measurable}.
\qed
\end{proof}

\begin{restate}{Lemma~\ref{lemma:peval-measurable}}
$\mathtt{peval}$ is a measurable function $\cterms  \times  \sampseq \rightarrow  \cterms$.
\end{restate}
\begin{proof}
Corollary of Lemma \ref{lemma:peval-prim-measurable} and Lemma \ref{lemma:peval-two-def-eqiv}.
\qed
\end{proof}

\begin{lemma} \label{lemma:peval-assoc-pt1}
  For every $M \in C\Lambda, c \in \mathbb{R}, s \in \mathbb{S}$,
  $\mathtt{peval}(\mathtt{peval}(M,[c]), s) \leq
  \mathtt{peval}(M, c \mathrel{::} s)$
\end{lemma}
\begin{proof}
Define a property $P \subseteq F$:
\[
  P(f) \Leftrightarrow  \sup_k \Phi^k(\bot^\lambda) (\Phi(f)(M, [c]), s)
  \leq \sup_k \Phi^k(\bot^\lambda)(M, c \mathrel {::} s)\quad \forall M, c, s
\]

Since $\sup_k \Phi^k(\bot^\lambda)$ is a fixpoint of $\Phi$,
$\Phi(\sup_k \Phi^k(\bot^\lambda))= \sup_k \Phi^k(\bot^\lambda)$, so

\[
  P(\sup_k \Phi^k(\bot^\lambda)) \iff \sup_k \Phi^k(\bot^\lambda) (\sup_l \Phi^l(\bot^\lambda) (M, [c]), s)
  \leq \sup_k \Phi^k(\bot^\lambda)(M, c \mathrel {::} s)\quad \forall M, c, s
\]

So we only need to prove $P(\sup_k \Phi^k(\bot^\lambda))$.

To show that the property $P$ is $\omega$-inductive, let $f_1 \leq f_2 \leq \dots$ be a
$\omega$-chain and $\sup_i f_i$ its limit. Then $\Phi(f_1) \leq \Phi(f_2) \leq \dots$
and $\sup_i \Phi(f_i) = \Phi(\sup_i f_i)$. For all $M$, $s$, we have

\[
  \Phi(\sup_i f_i) (M,s) = (\sup_i \Phi(f_i))(M,s) = \sup_i (\Phi(f_i)(M,s))
\]

Note that either $\Phi(f_i)(M,s) = \fail$ for all $i$ or there is some $n$ such that
$\Phi(f_n)(M,s) \in \valset$ and $\Phi(f_m)(M,s) = \Phi(f_n)(M,s)$ for all $m > n$.
In either case, there is a $n(M,s)$ such that $\sup_i (\Phi(f_i)(M,s)) = \Phi(f_{n(M,s)})(M,s)$.
Hence

\begin{eqnarray*}
  P(\sup_i f_i) &\Leftrightarrow&  \sup_k \Phi^k(\bot^\lambda) (\Phi(\sup_i f_i)(M, [c]), s)
  \leq \sup_k \Phi^k(\bot^\lambda)(M, c \mathrel {::} s)\quad \forall M, c, s \\
  &\Leftrightarrow& \sup_k \Phi^k(\bot^\lambda) (\Phi(f_{n(M,[c])})(M, [c]), s)
  \leq \sup_k \Phi^k(\bot^\lambda)(M, c \mathrel {::} s)\quad \forall M, c, s \\
  &\Leftrightarrow& P(f_{n(M,[c])})
\end{eqnarray*}
as required.

Now we can prove the desired property by Scott induction:
\begin{itemize}
  \item Base case:
    \begin{eqnarray*}
      P(\bot^\lambda) &\Leftrightarrow& \sup_k \Phi^k(\bot^\lambda) (\Phi(\bot^\lambda)(M, [c]), s)
       \leq \sup_k \Phi^k(\bot^\lambda)(M, c \mathrel {::} s)\quad \forall M, c, s 
    \end{eqnarray*}
    For any $M$, $c$, $s$, we have
    \[
      \sup_k \Phi^k(\bot^\lambda) (\Phi(\bot^\lambda)(M, [c]), s)
      = \sup_k \Phi^k(\bot^\lambda) (\fail, s) = \fail \leq \sup_k \Phi^k(\bot^\lambda)(M, c \mathrel{::}  s)
    \]
    as required.

  \item Induction step:
    We need to show that for all $f$ such that $P(f)$, $P(\Phi(f))$ holds,
    that is 
\[
 \sup_k \Phi^k(\bot^\lambda) (\Phi(\Phi(f))(M, [c]), s)
  \leq \sup_k \Phi^k(\bot^\lambda)(M, c \mathrel {::} s)\quad \forall M, c, s
\]

    \begin{itemize}
      \item Case $(M,1,[c]) \rightarrow (M',w,[c])$:
        \begin{eqnarray*}
          LHS &=& \sup_k \Phi^k(\bot^\lambda) (\Phi(\Phi(f))(M, [c]), s)\\
              &=& \sup_k \Phi^k(\bot^\lambda) (\Phi(f)(M', [c]), s)\\
              (\text{by assumption}) &\leq& \sup_k \Phi^k(\bot^\lambda)(M', c \mathrel{::} s)\\
              &=& \Phi(\sup_k \Phi^k(\bot^\lambda))(M, c \mathrel{::} s)\\
              &=& (\sup_k \Phi^k(\bot^\lambda)(M, c \mathrel{::} s)\\
              &=& RHS 
        \end{eqnarray*}

      \item Case $(M,1,[c]) \rightarrow (M',w,[])$:
        \begin{eqnarray*}
          LHS &=& \sup_k \Phi^k(\bot^\lambda) (\Phi(\Phi(f))(M, [c]), s)\\
              &=& \sup_k \Phi^k(\bot^\lambda) (\Phi(f)(M', []), s)\\
              &=& \sup_k \Phi^k(\bot^\lambda) (M', s)\\
              &=& \Phi(\sup_k \Phi^k(\bot^\lambda))(M, c \mathrel{::} s)\\
              &=& (\sup_k \Phi^k(\bot^\lambda)(M, c \mathrel{::} s)\\
              &=& RHS 
        \end{eqnarray*}

      \item Case $(M,1,[c]) \not\rightarrow $:
        \begin{eqnarray*}
          LHS &=& \sup_k \Phi^k(\bot^\lambda) (\Phi(\Phi(f))(M, [c]), s)\\
              &=& \sup_k \Phi^k(\bot^\lambda) (\fail, s)\\
              &=& \Phi(\sup_k \Phi^k(\bot^\lambda))(M, c \mathrel{::} s)\\
              &=& (\sup_k \Phi^k(\bot^\lambda)(M, c \mathrel{::} s)\\
              &=& RHS 
        \end{eqnarray*}

    \end{itemize}
\end{itemize}

Therefore, $P(\sup_k \Phi^k(\bot^\lambda))$, holds, and so 
  $\mathtt{peval}(\mathtt{peval}(M,[c]), s) \leq
  \mathtt{peval}(M, c \mathrel{::} s)$ for all closed $M$, $c$, $s$.
\qed
\end{proof}

\begin{lemma} \label{lemma:peval-assoc-pt2}
  For every $M \in C\Lambda, c \in \mathbb{R}, s \in \mathbb{S}$,
  $\mathtt{peval}(\mathtt{peval}(M,[c]), s) \geq 
  \mathtt{peval}(M, c \mathrel{::} s)$
\end{lemma}
\begin{proof}
Like in the previous lemma, we use Scott induction. Define the property:
\begin{eqnarray*}
  Q(f) &\Leftrightarrow& 
  f (M, c \mathrel {::} s)  \leq \sup_k \Phi^k(\bot^\lambda)  
  (\sup_l \Phi^l(\bot^\lambda)(M, [c]), s) \quad \forall M, c, s \\
  && \wedge\  f \leq \sup_k \Phi^k(\bot^\lambda)
\end{eqnarray*}

We need to show that $Q(\sup_k \Phi^k(\bot^\lambda))$ holds.

First, we need to verify that $Q$ is $\omega$-inductive. This is obvious for the second 
conjunct, so let us concentrate on the first. Once again, we use the
property that for all $\omega$-chains $f_1 \leq f_2 \leq \dots$ and $M$, $s$,
the chain $f_1(M,s) \leq f_2(M,s) \leq \dots$ will eventually be stationary.
For all $M$,$c$,$s$, we have 
$(\sup_i f_i) (M, c \mathrel {::} s) = f_{n(M,c,s)}(M, c \mathrel {::} s)$ for
some $n(M,c,s)$.
Then for every $M$,$c$,$s$, the inequality
\[
(\sup_i f_i) (M, c \mathrel {::} s) \leq \sup_k \Phi^k(\bot^\lambda)  
  (\sup_l \Phi^l(\bot^\lambda)(M, [c]), s)
\]
follows from $Q(f_{n(M,c,s)})$.

\begin{itemize}
  \item Base case:
    \begin{eqnarray*}
  Q(\bot^\lambda) &\Leftrightarrow&  
  \bot^\lambda (M, c \mathrel {::} s)  \leq \sup_k \Phi^k(\bot^\lambda)  
  (\sup_l \Phi^l(\bot^\lambda)(M, [c]), s) \quad \forall M, c, s \\
  && \wedge\  \bot^\lambda \leq \sup_k \Phi^k(\bot^\lambda)
  %\bot^\lambda (\sup_k \Phi^k(\bot^\lambda)(M, [c]), s)
  %\geq \bot^\lambda (M, c \mathrel {::} s)\quad \forall M, c, s
    \end{eqnarray*}
    Both inequalities are obvious, because the LHS is always $\fail$.

  \item Induction step: Give $Q(f)$, for every $M$, $c$, $s$ we need to show:
    \begin{eqnarray*}
     \Phi(f) (M, c \mathrel {::} s) &\leq&
     \sup_k \Phi^k(\bot^\lambda) (\sup_l \Phi^l(\bot^\lambda)(M, [c]), s)
    \quad \forall M, c, s\\
    \Phi(f) &\leq& \sup_k \Phi^k(\bot^\lambda)
    \end{eqnarray*}

    Again, the second inequality is obvious, so let us concentrate on the first.
  \begin{itemize}
   \item Case $(M,1,[c]) \rightarrow (M',w,[c])$:
     \begin{eqnarray*}
       LHS &=& f(M', c \mathrel{::} s)\\
       \text{(by first assumption)} &\leq& \sup_k \Phi^k(\bot^\lambda) (\sup_l \Phi^l(\bot^\lambda)(M', [c]), s)\\
       &=& \sup_k \Phi^k(\bot^\lambda) (\Phi(\sup_l \Phi^l(\bot^\lambda))(M, [c]), s)\\
       &=& \sup_k \Phi^k(\bot^\lambda) (\sup_l \Phi^l(\bot^\lambda)(M, [c]), s)\\
       &=& RHS
     \end{eqnarray*}

   \item Case $(M,1,[c]) \rightarrow (M',w,[])$:
     \begin{eqnarray*}
      LHS &=& f(M',s)\\
      \text{(by second assumption)} &\leq& \sup_k \Phi^k(\bot^\lambda)(M',s)\\
      &=& \sup_k \Phi^k(\bot^\lambda)(\Phi(\sup_l \Phi^l(\bot^\lambda))(M', []) ,s)\\
      &=& \sup_k \Phi^k(\bot^\lambda)(\sup_l \Phi^l(\bot^\lambda)(M', []) ,s)\\
      &=& \sup_k \Phi^k(\bot^\lambda)(\Phi(\sup_l \Phi^l(\bot^\lambda))(M, [c]) ,s)\\
      &=& \sup_k \Phi^k(\bot^\lambda)(\sup_l \Phi^l(\bot^\lambda)(M, [c]) ,s)\\
      &=& RHS
     \end{eqnarray*}
   \item Case $(M,1,[c]) \not\rightarrow $:
     \begin{eqnarray*}
     LHS &=& \fail\\
         &\leq& RHS
     \end{eqnarray*}
\end{itemize}
\end{itemize}
As required.
\qed
\end{proof}

\begin{restate}{Lemma~\ref{lemma:peval-assoc-step}}
  $\mathtt{peval}(\mathtt{peval}(M,[c]), s) = \mathtt{peval}(M,c \mathrel{::} s) $
\end{restate}
\begin{proof}
  Follows from Lemmas \ref{lemma:peval-assoc-pt1} and \ref{lemma:peval-assoc-pt2}. \qed
\end{proof}

%\begin{restate}{Lemma~\ref{lemma:peval-assoc-step}}
% Again $\mathtt{peval}(\mathtt{peval}(M,[c]), s) = \mathtt{peval}(M,c \mathrel{::} s) $
%\end{restate}
%\begin{proof}
%  
%  Follows from Lemmas \ref{lemma:peval-assoc-pt1} and \ref{lemma:peval-assoc-pt2}. \qed
%\end{proof}

\begin{restate}{Lemma~\ref{lemma:peval-assoc}}
For all closed $M$, $s$, $t$,  $\mathtt{peval}(\mathtt{peval}(M,s), t) = \mathtt{peval}(M,s @ t) $
\end{restate}
\begin{proof}
By induction on $|s|$.
\begin{itemize}
  \item Base case: $s = []$.
    We have $\mathtt{peval}(M,[]) = M$ (by definition of $\mathtt{peval}$),
    so the result is trivial.
  \item Induction step: $s = c \mathrel{::} s'$.

    We want $\mathtt{peval}(\mathtt{peval}(M,c \mathrel{::}s), t) = \mathtt{peval}(M, c \mathrel{::} s @ t)$.

    We have:
    \begin{eqnarray*}
    LHS &=& \mathtt{peval}(\mathtt{peval}(M,c \mathrel{::}s), t)\\
    \text{(by Lemma \ref{lemma:peval-assoc-step})}&=& \mathtt{peval}(\mathtt{peval}(\mathtt{peval}(M,[c]), s), t)\\
    \text{(by induction hypothesis)} &=& (\mathtt{peval}(\mathtt{peval}(M,[c]), s@t) \\
    \text{(by Lemma \ref{lemma:peval-assoc-step})}&=& \mathtt{peval}(M, [c] \mathrel{::} s@t)
    \end{eqnarray*}
\end{itemize}
\qed
\end{proof}

\subsection{Measurability of $q$ and $Q$}

\begin{lemma} \label{lemma:q-less-one}
  For all $s \in \mathbb{S}$ and $M \in \cterms$, $\int_{\mathbb{S} \setminus {[] }} q_M(s,t) \mu(dt) \leq 1 $
\end{lemma}

To prove this lemma, we need some auxiliary results:

%\begin{lemma} \label{lemma:peval-composition}
%  \[
%    \mathtt{peval}(\mathtt{peval}(M,s),t) = \mathtt{peval}(M, (s@t))
%  \]
%\end{lemma}
\begin{lemma} \label{lemma:eval-empty}
  If $M \Downarrow^{[]}_w G$ and $M \Downarrow^s_{w'} G'$, then $s = []$.
\end{lemma}
\begin{proof}
By induction on the derivation of $M \Downarrow^{[]}_w G$.
\qed
\end{proof}

\begin{lemma} \label{lemma:P-empty}
  If $\mathbf{P}^{\valset}_M([]) > 0$, then $\mathbf{P}^{\valset}_M(t) = 0$ for all $t \neq []$.
\end{lemma}
\begin{proof}
  Follows directly from Lemma \ref{lemma:eval-empty}.
  \qed
\end{proof}

\begin{lemma}[Tonelli's theorem for sums and integrals, 1.4.46 in \cite{tao2011measure}] \label{lemma:Tao-1446} 
If $(\Omega, \Sigma, \mu)$ is a measure space and $f_1, f_2, \dots$ a sequence of non-negative
measurable functions, then
\[
  \int_{\Omega}\ \sum_{i=1}^{\infty} f_i(x)\ \mu(dx) = \sum_{i=1}^{\infty}\ \int_{\Omega} f_i(x)\ \mu(dx)
\]
\end{lemma}
\begin{proof}
Follows from the monotone convergence theorem.
\qed
\end{proof}

\begin{lemma}[Linearity of Lebesgue integral, 1.4.37 ii) from \cite{tao2011measure}]
If $(\Omega, \Sigma)$ is a measurable space, $f$ a non-negative
measurable function, and $\mu_i, \mu_2, \dots$ a sequence of measures
on $\Sigma$, then
\[
  \int_\Omega f(x)\ \sum_{i=1}^{\infty}\ \mu_i(dx) = \sum_{i=1}^{\infty}\ \int_\Omega f(x)\ \mu_i(dx) 
\]
\end{lemma}

\begin{lemma}[Ex. 1.4.36 xi) from \cite{tao2011measure}] \label{lemma:Tao-1436}
If $(\Omega, \Sigma, \mu)$ is a measure space and $f$ a nonnegative measurable function on $\Omega$
and $B \in \Sigma$ and $f^B$ a restriction of $f$ to $B$, then

\[
  \int_\Omega f(x) [x \in B]\ \mu(dx) = \int_B f(x)\ \mu^B(dx)
\]

\end{lemma}

Below we write $q(s,t)$ as $q_M(s,t)$, to make the dependency on $M$ explicit.

Let $q_M^*$ be defined as follows:

\[
  q_M^*(s,t) =
  \begin{cases}
    \mathbf{P}_M^{\valset} ([]) & \text{if}\ t = []\\
    q_M(s,t) & \text{otherwise}
  \end{cases}
\]

\begin{lemma} \label{lemma:qstar-rec}
  For all $M \in \cterms$ and  $s,y \in \mathbb{S}$
\[
  q_M^*(s,t) =
  \begin{cases}
    \mathbf{P}_M^{\valset}(t) & \text{if}\ s = []\ \text{or}\ t = []\\
    \pdf{\mathsf{Gaussian}}(s_1, \sigma^2, t_1) q_{\mathtt{peval}(M, [s_1])}^*([s_2, \dots s_{|s|}], [t_2, \dots t_{|t|}])  & \text{otherwise}
  \end{cases}
\]
\end{lemma}
\begin{proof}
By induction on $\Abs s$:
\begin{itemize}
  \item Case $s = []$:

    If $t = []$, the result follows directly from the definition of $q_M^*$.
    Otherwise, $q_M^*([],t) = q_M([], t) = P_M^{\valset}(t)$, as required.
  \item Case $|s| = n+1 > 0$:

    Again, if $t = []$, the result follows immediately. Otherwise,
    we have
    \[
      q_M^*(s,t) = q_M(s,t) =  \Pi_{i=1}^k (\pdf{\mathsf{Gaussian}}(s_i, \sigma^2, t_i)) 
      \mathbf{P}_{\mathtt{peval}(M,[t_1,\dots,t_k])}^{\valset}(t) 
    \]
    where $k = \mathtt{min}(|s|,|t|) > 0$. Hence
    \begin{eqnarray*}
      q_M^*(s,t) &=& \pdf{\mathsf{Gaussian}}(s_1, \sigma^2, t_1) \Pi_{i=2}^k
      (\pdf{\mathsf{Gaussian}}(s_i, \sigma^2, t_i)) \mathbf{P}_{\mathtt{peval}(M,[t_1,\dots,t_k])}^{\valset}([t_{k+1}, \dots, t_{|t|}])\\
      \text{(by Lemma \ref{lemma:peval-assoc-step})}&=& \pdf{\mathsf{Gaussian}}(s_1, \sigma^2, t_1) \Pi_{i=2}^k
      (\pdf{\mathsf{Gaussian}}(s_i, \sigma^2, t_i)) 
      \mathbf{P}_{\mathtt{peval}(\mathtt{peval}(M, [t_1]),[t_2,\dots,t_k])}^{\valset}([t_{k+1}, \dots, t_{|t|}]) \\
      &=& (\pdf{\mathsf{Gaussian}}(s_1, \sigma^2, t_1)) 
      q_{\mathtt{peval}(M,[s_1])}^*([s_2, \dots, s_{|s|})], [t_2, \dots, t_{|t|}])
    \end{eqnarray*}
    as required.
    \qed
%    Let $s' = [s_2, \dots, s_{|s|}]$ and $t' = [t_2, \dots, t_{|t|}]$. By induction hypothesis, we have:
%\[
%  q_{\mathtt{peval}(M,[s_1]) }^*(s',t') =
%  \begin{cases}
%    \mathbf{P}_M^{\valset}(t') & \text{if}\ s' = []\ \text{or}\ t' = []\\
%    \pdf{\mathsf{Gaussian}}(s_1, \sigma^2, t_1) q_{\mathtt{peval}(M, [s_1])}^*([s_3, \dots s_{|s|}], [t_3, \dots t_{|t|}])  & \text{otherwise}
%  \end{cases}
%\]
\end{itemize}
\end{proof}

\begin{lemma} \label{lemma:peval-zero}
If $\mathbf{P}^{\valset}_M([]) > 0$, then $\mathtt{peval}(M,t) = \fail$ for every
$t \neq []$.
\end{lemma}
\begin{proof}
It $\mathbf{P}^{\valset}_M([]) = w > 0$, then we must have $M \Downarrow^{[]}_w V$ for some
$V \in \valset$, which implies $(M,1,[]) \Rightarrow (G,w,[])$. Using
Lemma \ref{lemma:add-suffix}, we can easily show by induction that
$(M,1,t) \Rightarrow (G,w,t)$ for any $t \neq []$. Because the reduction relation is deterministic,
this implies that there are no $M'$, $w'$ such that $(M,1,t) \Rightarrow (M',w',[])$
(if there were, we would have $(M',w',[]) \Rightarrow (G,w,t)$, but no reduction rule can add an element to a trace).
This means that $\mathtt{peval}$, by applying reduction repeatedly, will never reach $(M',[])$
for any $M'$, so $\mathtt{peval}(M,t) = \fail$.
\qed
\end{proof}

\begin{lemma} \label{lemma:p-or-q}
  If $\mathbf{P}^{\valset}_M([]) > 0$, then $q_M^*(s,t) = 0$ for all $s \in \mathbb{S},t \neq []$.
\end{lemma}
\begin{proof}
Follows easily from Lemma \ref{lemma:peval-zero}. %By Lemma \ref{aa}
\qed
\end{proof}
%\begin{proof}
%For $t \neq []$, we have $q_M^*(s,t) = q_M(s,t) = $

%\ref{lemma:P-empty}
%\end{proof}

\begin{restate}{Lemma~\ref{lemma:q-less-one} }
For all $s \in \mathbb{S}$ and $M \in \cterms$, $\int_{\mathbb{S} \setminus {[] }} q_M(s,t) \mu(dt) \leq 1 $
\end{restate}
\begin{proof}
By induction on $|s|$.
\begin{itemize}  
      \item Base case: $s = []$
  \begin{eqnarray*}
    &&\int_{\mathbb{S} \setminus \{[]\}} q_M([],t)\ \mu(dt)\\
    &=&\int_{\mathbb{S} \setminus \{[]\}} \mathbf{P}_M^{\valset}(t)\ \mu(dt)\\
    &\leq& \int_{\mathbb{S}} \mathbf{P}_M(t)\ \mu(dt)\\
    &=& \Tracedist{M}(\mathbb{S})\\
    &=& \Tracedist{M}(\mathbf{O}^{-1}_M(\gvalset))\\
    &=& \sbd{M}(\gvalset)\\
    \text{by Theorem \ref{thm:sampling-distribution}} &=& \tsq{M}(\gvalset)\\
    &\leq& 1
  \end{eqnarray*}
  because $\tsq{M}$ is a sub-probability distribution.
          
  \item Induction step: $s \neq []$
  
  We have:
  \begin{eqnarray*}
    &&\int_{\mathbb{S} \setminus \{[]\}} q_M(s,t)\ \mu(dt)\\
    &=&\int_{\mathbb{S} \setminus \{[]\}} q_M^*(s,t)\ \mu(dt)\\
    \text{(by Thm 16.9 from Billingsley) }    &=& \sum_{i=1}^\infty \int_{\mathbb{R}^i} q_M^*(s,t)\ \mu(dt)\\
    \text{(by Lemma \ref{lemma:Tao-1436})}  &=& \sum_{i=1}^\infty \int_{\mathbb{R}^i} q_M^*(s,t)\ \lambda^i(dt)\\
    &=& \sum_{i=1}^\infty \int_{\mathbb{R}^i} \pdf{\mathsf{Gaussian}}(s_1,\sigma^2,t_1) 
    q_{\mathtt{peval}(M,[t_1]) }^*([s_2, \dots, s_{|s|}],[t_2, \dots, t_{|t|}])\ \lambda^i(dt)\\
    \text{(by Fubini's theorem)}
    &=& \sum_{i=1}^\infty \int_{\mathbb{R}} \pdf{\mathsf{Gaussian}}(s_1,\sigma^2,t_1) 
    \int_{\mathbb{R}^{i-1}} q_{\mathtt{peval}(M,[t_1]) }^*(s',t')\ \lambda^{i-1}(dt')\ \lambda(dt_1)\\
    \text{(by Lemma \ref{lemma:Tao-1446})}
    &=&  \int_{\mathbb{R}} \pdf{\mathsf{Gaussian}}(s_1,\sigma^2,t_1) \sum_{i=0}^\infty
    \int_{\mathbb{R}^{i}} q_{\mathtt{peval}(M,[t_1]) }^*(s',t')\ \lambda^{i}(dt')\ \lambda(dt_1)\\
    &=&  \int_{\mathbb{R}} \pdf{\mathsf{Gaussian}}(s_1,\sigma^2,t_1) 
    \left(\int_{\{[]\}}\mathbf{P}_{\mathtt{peval}(M,[t_1]) }^{\valset}(t')\ \mu (dt')\right. + \\ 
    &&\qquad \left.\int_{\mathbb{S} \setminus \{[]\}} q_{\mathtt{peval}(M,[t_1]) }^*(s',t')\ \mu(dt') \right)\ \lambda(dt_1)
%    \text{(by definition of $q_M^*$)} &=&  \int_{\mathbb{R}} \pdf{\mathsf{Gaussian}}(s_1,\sigma^2,t_1) 
%    (\mathbf{P}_{\mathtt{peval}(M,[t_1]) }^{\valset}([]) +\\ 
%    &&\qquad \int_{\mathbb{S} \setminus \{[]\}} q_{\mathtt{peval}(M,[t_1]) }(s',t')\ \lambda^{i}(dt') )\ \lambda(dt_1)\\
%    &=&  \int_{\mathbb{R}} \pdf{\mathsf{Gaussian}}(s_1,\sigma^2,t_1) 
%    (\mathbf{P}_{\mathtt{peval}(M,[t_1]) }^{\valset}([]) +\\ 
%    &&\qquad \int_{\mathbb{S} \setminus \{[]\}}
%    \Sigma_{i=2}^_{\mathtt{peval}(M,[t_1]) }(s',t')\ \lambda^{i}(dt') )\ \lambda(dt_1)\\
  \end{eqnarray*}

Now, we need to show that for all $N$,

\begin{equation} \label{eqn:-q-less-one}
    \int_{\{[]\}}\mathbf{P}_{N}^{\valset}(t')\ \mu (dt') + 
    \int_{\mathbb{S} \setminus \{[]\}} q_{N}^*(s',t')\ \mu(dt') \leq 1
  \end{equation}

  First, note that $\int_{\{[]\}}\mathbf{P}_{N}^{\valset}(t')\ \mu (dt')
  \leq \int_{\mathbb{S}}\mathbf{P}_{N}^{\valset}(t')\ \mu (dt') \leq 1$,
  by the same property as the one used in the base case.
  We also have 
  $\int_{\{[]\}}\mathbf{P}_{N}^{\valset}(t')\ \mu (dt)
  = \mathbf{P}_{N}^{\valset}([])$, so by Lemma \ref{lemma:p-or-q}, if $\mathbf{P}_{N}^{\valset}([])>0$,
  then 
\[
  \int_{\{[]\}}\mathbf{P}_{N}^{\valset}(t')\ \mu (dt') + 
  \int_{\mathbb{S} \setminus \{[]\}} q_{N}^*(s',t')\ \mu(dt') =  
  \int_{\{[]\}}\mathbf{P}_{N}^{\valset}(t')\ \mu (dt') \leq 1 
\]

On the other hand, if $\mathbf{P}_{N}^{\valset}([])= 0$, then
\[
  \int_{\{[]\}}\mathbf{P}_{N}^{\valset}(t')\ \mu (dt') + 
  \int_{\mathbb{S} \setminus \{[]\}} q_{N}^*(s',t')\ \mu(dt') =  
  \int_{\mathbb{S} \setminus \{[]\}} q_{N}^*(s',t')\ \mu(dt') = 
  \int_{\mathbb{S} \setminus \{[]\}} q_{N}(s',t')\ \mu(dt') \leq 1 
\]
by induction hypothesis.

Hence:
  \begin{eqnarray*}
    &&\int_{\mathbb{S} \setminus \{[]\}} q_M(s,t)\ \mu(dt)\\
    &\leq&  \int_{\mathbb{R}} \pdf{\mathsf{Gaussian}}(s_1,\sigma^2,t_1) 
    \ \lambda(dt_1)\\
    &=& 1 
  \end{eqnarray*}
  as required.
  \qed
\end{itemize}
\end{proof}

\begin{restate}{Lemma~\ref{lemma:q-measurable}}
For any closed program $M$, the transition density 
$q(\cdot, \cdot) : (\sampseq \times \sampseq) \rightarrow \mathbb{R}_{+}$ 
is measurable.
\end{restate}
\begin{proof}
It is enough to show that $q(s,t)$ is measurable for every $|s|=n$ and $|t|=m$, then
the result follows from Lemma \ref{lemma:split-space}.

Note that a function taking a sequence $s$ and returning any subsequence of it
is trivially continuous and measurable, so for any function of $s$ and $t$ to be measurable,
it is enough to show that it is measurable as a function of some projections of $s$ and $t$.

\begin{itemize}
\item If $m > 0$ and $n < m$, then we have
$q(s,t) = \Pi_{i=1}^n \pdf{\mathsf{Gaussian}}(s_1, \sigma^2, t_i) 
\mathbf{P}_{\mathtt{peval}(M,t_{1..n} )}^\valset(t_{n+1 .. m})
= \Pi_{i=1}^n \pdf{\mathsf{Gaussian}}(s_i, \sigma^2, t_i) 
\mathbf{P'}(\mathtt{peval}(M,t_{1..n}), t_{n+1 .. m}) 
[\mathbf{O'}(\mathtt{peval}(M,t_{1..n}), t_{n+1 .. m}) \in \valset]$.

Each $\pdf{\mathsf{Gaussian}}(s_i, \sigma^2, t_i)$ is measurable, as a composition
of a function projecting $(s_i,t_i)$ from $(s,t)$ and the Gaussian pdf, so their
pointwise product must be measurable.

Now, $\mathbf{P'}$ is measurable, and the function mapping
$(s,t)$ to $(\mathtt{peval}(M,t_{1..n}), t_{n+1 .. m})$
is a pair of two measurable functions, one of which is a composition
of the measurable $\mathtt{peval}(M, \cdot)$ and a projection of $t_{1..n}$,
and the other just a projection of $t_{n+1 .. m})$.
Hence, the function mapping $(s,t)$ to
$\mathbf{P'}(\mathtt{peval}(M,t_{1..n}), t_{n+1 .. m})$ is a composition 
of measurable functions.

Finally, $[\mathbf{O'}(\mathtt{peval}(M,t_{1..n}), t_{n+1 .. m}) \in \valset]$ is
a composition of the measurable function mapping
$(s,t)$ to $(\mathtt{peval}(M,t_{1..n}), t_{n+1 .. m})$ and the indicator function
for the measurable set $\valset$, thus it is measurable.

Hence, $q(s,t)$ is a pointwise product of measurable functions, so it is measurable.

\item If $m > 0$ and $n \geq m$, then 
$q(s,t) = \Pi_{i=1}^m \pdf{\mathsf{Gaussian}}(s_1, \sigma^2, t_i) 
\mathbf{P}_{\mathtt{peval}(M,t )}^\valset([])$

$=  \Pi_{i=1}^m \pdf{\mathsf{Gaussian}}(s_i, \sigma^2, t_i) 
\mathbf{P'}(\mathtt{peval}(M,t), []) 
[\mathbf{O'}(\mathtt{peval}(M,t), []) \in \valset]$.

Now, the function mapping $(s,t)$ to $ \Pi_{i=1}^m \pdf{\mathsf{Gaussian}}(s_i, \sigma^2, t_i) $
is measurable like in the previous case. The function mapping $(s,t)$ to $(\mathtt{peval}(M,t), [])$
is a pairing of two measurable functions, one being a composition of the projection of $t$ and 
$\mathtt{peval}(M, \cdot)$, the other being a constant function returning $[]$. Hence,
$\mathbf{P'}(\mathtt{peval}(M,t), []) $ is a composition of two measurable functions.
Meanwhile, $[\mathbf{O'}(\mathtt{peval}(M,t), []) \in \valset]$ is a composition of a measurable function
and an indicator function.

%is clearly
%measurable, beacause a product of Gaussians is measurable, and $\mathbf{P}$
%and $\mathtt{peval}$ are also measurable.

\item If $m = 0$, then $q(s,[]) = 1-\int_{\mathbb{S} \setminus \{[] \}} q(s, t)\, \mu(dt)$.
Since we have already shown that $q(s,t)$ is measurable on $\mathbb{S} \times 
(\mathbb{S} \setminus [])$, $\int_{\mathbb{S} \setminus \{[] \}} q(s, t)\, \mu(dt)$
is measurable by Fubini's theorem, so $q(s,[])$ is a difference of measurable functions, and
hence it is measurable.\qed

%the measurability of $q(s,t)$ follows from the measurability
%of $q(s,t)$ for $t \neq []$, 
\end{itemize}
% 
%   Outline:
%  \begin{enumerate}
%  \item $q$ restricted to $\sampseq \times (\sampseq\setminus\Set{[]})$ is measurable, 
%    since it is composed from $\pdf{\mathsf{Gaussian}}$ (which is continuous) 
%       and \texttt{peval}(which is measurable).
%  \item $\int_{\sampseq\setminus\Set{[]}}q(s,dt)\le 1$ for all $s$, 
%    by the same property for $\pdf{\mathsf{Gaussian}}$ and $\Termsmv{\termone}$.
%  \end{enumerate}
\end{proof}

\begin{restate}{Lemma~\ref{lemma:q-kernel}}
The function $Q$ is a probability kernel on $(\sampseq, \mathcal{S})$.
\end{restate}
\begin{proof}
We need to verify the two properties of probability kernels:
\begin{enumerate}
\item For every $s \in \sampseq$, $Q(s, \cdot)$ is a probability distribution on $\sampseq$.
Since for every $s \in \sampseq$, 
$q(s,\cdot)$ is non-negative measurable $\mathcal{S}$ 
(by \cite[Theorem 18.1]{billingsley95}),
$Q(s,B) = \int_B q(s,y) \mu(dy)$ (as a function of $B$) is a 
well-defined measure for all $s \in \sampseq$. 
Finally, $Q(s, \sampseq) = Q(s,[])+Q(s, \sampseq\setminus\Set{[]}) = 1$.
\item For every $B \in \mathcal{S}$, $Q(\cdot, B)$ is a non-negative 
measurable function on $\sampseq$:
Since $(\sampseq, \mathcal{S}, \mu)$ is a $\sigma$-finite measure space, 
$q(\cdot, \cdot)$ is non-negative and measurable 
$\mathcal{S} \times \mathcal{S}$ and
$Q(s,B)  =\int_B q(s,y) \mu(ds)$,
this follows from \cite[Theorem 18.3]{billingsley95}. \qed
\end{enumerate}

\end{proof}

%\ifdraft % whole of this appendix
\bibliographystyle{abbrvnat}
\bibliography{bibliography}

\begin{thebibliography}{42}
\providecommand{\natexlab}[1]{#1}
\providecommand{\url}[1]{\texttt{#1}}
\expandafter\ifx\csname urlstyle\endcsname\relax
  \providecommand{\doi}[1]{doi: #1}\else
  \providecommand{\doi}{doi: \begingroup \urlstyle{rm}\Url}\fi

\bibitem[Bhat et~al.(2013)Bhat, Borgstr{\"o}m, Gordon, and
  Russo]{BBGR12:DerivingPDFs}
S.~Bhat, J.~Borgstr{\"o}m, A.~D. Gordon, and C.~V. Russo.
\newblock Deriving probability density functions from probabilistic functional
  programs.
\newblock In N.~Piterman and S.~A. Smolka, editors, \emph{Proceedings of TACAS
  2013}, volume 7795 of \emph{LNCS}, pages 508--522. Springer, 2013.

\bibitem[Billingsley(1995)]{billingsley95}
P.~Billingsley.
\newblock \emph{Probability and Measure}.
\newblock Wiley-Interscience, third edition, 1995.

\bibitem[Bizjak and Birkedal(2015)]{DBLP:conf/fossacs/BizjakB15}
A.~Bizjak and L.~Birkedal.
\newblock Step-indexed logical relations for probability.
\newblock In A.~M. Pitts, editor, \emph{Proceedings of FoSSaCS 2015}, volume
  9034 of \emph{LNCS}, pages 279--294. Springer, 2015.

\bibitem[Borgstr{\"o}m et~al.(2015)Borgstr{\"o}m, Dal~Lago, Gordon, and
  Szymczak]{mhlambda-arxiv}
J.~Borgstr{\"o}m, U.~Dal~Lago, A.~D. Gordon, and M.~Szymczak.
\newblock A lambda-calculus foundation for universal probabilistic programming
  (long version).
\newblock \emph{CoRR}, abs/1512.08990, 2015.

\bibitem[Cousot and Monerau(2012)]{CousotMonerau-ESOP2012-PAI}
P.~Cousot and M.~Monerau.
\newblock Probabilistic abstract interpretation.
\newblock In \emph{Proceedings of {ESOP} 2012}, volume 7211 of \emph{LNCS},
  pages 166--190. Springer, 2012.

\bibitem[Danos and Ehrhard(2011)]{DanosEhrhard11}
V.~Danos and T.~Ehrhard.
\newblock Probabilistic coherence spaces as a model of higher-order
  probabilistic computation.
\newblock \emph{Information and Computation}, 209\penalty0 (6):\penalty0
  966--991, 2011.

\bibitem[Danos and Harmer(2002)]{DanosHarmer02}
V.~Danos and R.~Harmer.
\newblock Probabilistic game semantics.
\newblock \emph{{ACM} Transactions on Computational Logic}, 3\penalty0
  (3):\penalty0 359--382, 2002.

\bibitem[Ehrhard et~al.(2014)Ehrhard, Tasson, and
  Pagani]{EhrhardTassonPagani14}
T.~Ehrhard, C.~Tasson, and M.~Pagani.
\newblock Probabilistic coherence spaces are fully abstract for probabilistic
  {PCF}.
\newblock In \emph{Proceedings of {POPL} 2014}, pages 309--320. ACM, 2014.

\bibitem[Ferrer~Fioriti and
  Hermanns(2015)]{Hermanns15.terminationProbabilistic}
L.~M. Ferrer~Fioriti and H.~Hermanns.
\newblock Probabilistic termination: Soundness, completeness, and
  compositionality.
\newblock In \emph{Proceedings of POPL 2015}, pages 489--501. ACM, 2015.

\bibitem[Gallay(2009)]{gallay2009mesure}
T.~Gallay.
\newblock \emph{Th{\'e}orie de la mesure et de l'int{\'e}gration}.
\newblock 2009.
\newblock Course notes available online at
  http://im2ag-webmath.e.ujf-grenoble.fr/enseignement2/IMG/pdf/integrationa.pdf.

\bibitem[Goldwasser and Micali(1984)]{GoldwasserMicali}
S.~Goldwasser and S.~Micali.
\newblock Probabilistic encryption.
\newblock \emph{Journal of Computer and System Sciences}, 28\penalty0
  (2):\penalty0 270--299, 1984.

\bibitem[Goodman(2013)]{DBLP:conf/popl/Goodman13}
N.~D. Goodman.
\newblock The principles and practice of probabilistic programming.
\newblock In \emph{Proceedings of {POPL} 2013}, pages 399--402. ACM, 2013.

\bibitem[Goodman and Stuhlm\"{u}ller(2014)]{dippl}
N.~D. Goodman and A.~Stuhlm\"{u}ller.
\newblock The design and implementation of probabilistic programming languages.
\newblock \url{http://dippl.org}, 2014.

\bibitem[Goodman and Tenenbaum(2014)]{probmods}
N.~D. Goodman and J.~B. Tenenbaum.
\newblock Probabilistic models of cognition.
\newblock \url{http://probmods.org}, 2014.

\bibitem[Goodman et~al.(2008)Goodman, Mansinghka, Roy, Bonawitz, and
  Tenenbaum]{DBLP:conf/uai/GoodmanMRBT08}
N.~D. Goodman, V.~K. Mansinghka, D.~M. Roy, K.~Bonawitz, and J.~B. Tenenbaum.
\newblock Church: a language for generative models.
\newblock In D.~A. McAllester and P.~Myllym{\"{a}}ki, editors,
  \emph{Proceedings of {UAI} 2008}, pages 220--229. {AUAI} Press, 2008.

\bibitem[Gordon et~al.(2014)Gordon, Henzinger, Nori, and
  Rajamani]{DBLP:conf/icse/GordonHNR14}
A.~D. Gordon, T.~A. Henzinger, A.~V. Nori, and S.~K. Rajamani.
\newblock Probabilistic programming.
\newblock In M.~B. Dwyer and J.~Herbsleb, editors, \emph{Proceedings of {FOSE}
  2014}, pages 167--181. ACM, 2014.

\bibitem[Hastings(1970)]{hastings70:MH}
W.~K. Hastings.
\newblock {Monte} {Carlo} sampling methods using {Markov} chains and their
  applications.
\newblock \emph{Biometrika}, 57:\penalty0 97--109, 1970.

\bibitem[Homan and Gelman(2014)]{Homan2014NUTS}
M.~D. Homan and A.~Gelman.
\newblock The no-u-turn sampler: Adaptively setting path lengths in
  {H}amiltonian {M}onte {C}arlo.
\newblock \emph{Journal of Machine Learning Research}, 15:\penalty0 1593--1623,
  2014.

\bibitem[Hur et~al.(2015)Hur, Nori, Rajamani, and Samuel]{corsamp}
C.~Hur, A.~V. Nori, S.~K. Rajamani, and S.~Samuel.
\newblock A provably correct sampler for probabilistic programs.
\newblock In P.~Harsha and G.~Ramalingam, editors, \emph{Proceedings of
  {FSTTCS} 2015}, volume~45 of \emph{LIPIcs}, pages 475--488. Schloss Dagstuhl
  - Leibniz-Zentrum f{\"u}r Informatik, 2015.

\bibitem[Jones(1990)]{jones90:PhD}
C.~Jones.
\newblock \emph{Probabilistic Non-determinism}.
\newblock PhD thesis, University of Edinburgh, 1990.
\newblock Available as Technical Report CST-63–90.

\bibitem[Jones and Plotkin(1989)]{JonesPlotkin89}
C.~Jones and G.~D. Plotkin.
\newblock A probabilistic powerdomain of evaluations.
\newblock In \emph{Proceedings of {LICS} 1989}, pages 186--195. ACM, 1989.

\bibitem[Kiselyov(2016)]{problems-kiselov16}
O.~Kiselyov.
\newblock Problems of the lightweight implementation of probabilistic
  programming, 2016.
\newblock Poster at PPS'2016 workshop.

\bibitem[Kozen(1979)]{DBLP:conf/focs/Kozen79}
D.~Kozen.
\newblock Semantics of probabilistic programs.
\newblock In \emph{Proceedings of {FOCS} 1979}, pages 101--114. {IEEE} Computer
  Society, 1979.

\bibitem[Manning and Sch{\"u}tze(1999)]{manning1999foundations}
C.~D. Manning and H.~Sch{\"u}tze.
\newblock \emph{Foundations of statistical natural language processing}.
\newblock MIT Press, 1999.

\bibitem[Mansinghka et~al.(2014)Mansinghka, Selsam, and
  Perov]{DBLP:journals/corr/MansinghkaSP14}
V.~K. Mansinghka, D.~Selsam, and Y.~N. Perov.
\newblock Venture: a higher-order probabilistic programming platform with
  programmable inference.
\newblock \emph{CoRR}, abs/1404.0099, 2014.

\bibitem[Metropolis et~al.(1953)Metropolis, Rosenbluth, Rosenbluth, Teller, and
  Teller]{metropolis53}
N.~Metropolis, A.~W. Rosenbluth, M.~N. Rosenbluth, A.~H. Teller, and E.~Teller.
\newblock Equations of state calculations by fast computing machines.
\newblock \emph{Journal of Chemical Physics}, 21\penalty0 (6):\penalty0
  1087--–1092, 1953.

\bibitem[Nori et~al.(2014)Nori, Hur, Rajamani, and
  Samuel]{DBLP:conf/aaai/NoriHRS14}
A.~V. Nori, C.~Hur, S.~K. Rajamani, and S.~Samuel.
\newblock {R2:} an efficient {MCMC} sampler for probabilistic programs.
\newblock In C.~E. Brodley and P.~Stone, editors, \emph{Proceedings of {AAAI}
  2014}, pages 2476--2482. {AAAI} Press, 2014.

\bibitem[Panangaden(1999)]{panangaden99markovkernels}
P.~Panangaden.
\newblock The category of {Markov} kernels.
\newblock \emph{Electronic Notes in Theoretical Computer Science}, 22:\penalty0
  171--187, 1999.
\newblock In proceedings of PROBMIV 1998.

\bibitem[Panangaden(2009)]{panangaden2009labelled}
P.~Panangaden.
\newblock \emph{Labelled Markov Processes}.
\newblock Imperial College Press, 2009.

\bibitem[Park et~al.(2008)Park, Pfenning, and Thrun]{park08sampling}
S.~Park, F.~Pfenning, and S.~Thrun.
\newblock A probabilistic language based upon sampling functions.
\newblock \emph{ACM Transactions on Programming Languages and Systems},
  31:\penalty0 1, 2008.

\bibitem[Pearl(1988)]{pearl1988probabilistic}
J.~Pearl.
\newblock \emph{Probabilistic reasoning in intelligent systems: networks of
  plausible inference}.
\newblock Morgan Kaufmann, 1988.

\bibitem[Ramsey and Pfeffer(2002)]{DBLP:conf/popl/RamseyP02}
N.~Ramsey and A.~Pfeffer.
\newblock Stochastic lambda calculus and monads of probability distributions.
\newblock In \emph{Proceedings of {POPL} 2002}, pages 154--165, 2002.

\bibitem[Russell(2015)]{DBLP:journals/cacm/Russell15}
S.~J. Russell.
\newblock Unifying logic and probability.
\newblock \emph{Communications of the {ACM}}, 58\penalty0 (7):\penalty0 88--97,
  2015.

\bibitem[Scibior et~al.(2015)Scibior, Ghahramani, and
  Gordon]{DBLP:conf/haskell/ScibiorGG15}
A.~Scibior, Z.~Ghahramani, and A.~D. Gordon.
\newblock Practical probabilistic programming with monads.
\newblock In B.~Lippmeier, editor, \emph{Proceedings of Haskell 2015}, pages
  165--176. {ACM}, 2015.

\bibitem[Staton et~al.(2016)Staton, Yang, Heunen, Kammar, and
  Wood]{DBLP:journals/corr/StatonYHKW16}
S.~Staton, H.~Yang, C.~Heunen, O.~Kammar, and F.~Wood.
\newblock Semantics for probabilistic programming: higher-order functions,
  continuous distributions, and soft constraints.
\newblock \emph{CoRR}, abs/1601.04943, 2016.

\bibitem[Tao(2011)]{tao2011measure}
T.~Tao.
\newblock \emph{An Introduction to Measure Theory}.
\newblock AMS, 2011.

\bibitem[Thrun(2002)]{thrun2002robotic}
S.~Thrun.
\newblock Robotic mapping: A survey.
\newblock In G.~Lakemeyer and B.~Nebel, editors, \emph{Exploring artificial
  intelligence in the new millennium}, pages 1--35. Morgan Kaufmann, 2002.

\bibitem[Tierney(1994)]{tierney1994}
L.~Tierney.
\newblock Markov chains for exploring posterior distributions.
\newblock \emph{The Annals of Statistics}, 22\penalty0 (4):\penalty0
  1701--1728, 1994.

\bibitem[Tolpin et~al.(2015)Tolpin, van~de Meent, and
  Wood]{DBLP:conf/pkdd/TolpinMW15}
D.~Tolpin, J.~van~de Meent, and F.~Wood.
\newblock Probabilistic programming in {Anglican}.
\newblock In A.~Bifet, M.~May, B.~Zadrozny, R.~Gavald{\`{a}}, D.~Pedreschi,
  F.~Bonchi, J.~S. Cardoso, and M.~Spiliopoulou, editors, \emph{Proceedings of
  {ECML} {PKDD} 2015, Part {III}}, volume 9286 of \emph{LNCS}, pages 308--311.
  Springer, 2015.

\bibitem[Toronto(2014)]{TorontoPhD}
N.~Toronto.
\newblock \emph{Useful Languages for Probabilistic Modeling and Inference}.
\newblock PhD thesis, Brigham Young University, Provo, UT, 2014.
\newblock URL
  \url{https://www.cs.umd.edu/~ntoronto/papers/toronto-2014diss.pdf}.

\bibitem[Toronto et~al.(2015)Toronto, McCarthy, and
  Horn]{DBLP:conf/esop/TorontoMH15}
N.~Toronto, J.~McCarthy, and D.~V. Horn.
\newblock Running probabilistic programs backwards.
\newblock In J.~Vitek, editor, \emph{Proceedings of {ESOP} 2015}, volume 9032
  of \emph{LNCS}, pages 53--79. Springer, 2015.

\bibitem[Wingate et~al.(2011)Wingate, Stuhlm\"{u}ller, and
  Goodman]{wingate2011lightweight}
D.~Wingate, A.~Stuhlm\"{u}ller, and N.~D. Goodman.
\newblock Lightweight implementations of probabilistic programming languages
  via transformational compilation.
\newblock In G.~J. Gordon and D.~Dunson, editors, \emph{Proceedings of AISTATS
  2011}, volume~15 of \emph{JMLR: Workshop and Conference Proceedings}, pages
  770--778. JMLR, 2011.

\end{thebibliography}
\end{document}